%% file: samplepaper.tex
\documentclass[runningheads]{llncs}
\usepackage{graphicx}
\usepackage[colorlinks]{hyperref}
\usepackage[numbers]{natbib}
\usepackage{float}
\usepackage{pgfplots}
\usetikzlibrary{pgfplots.groupplots}
\usepackage{tikz}
\usepackage{subfig}
\usepackage{amssymb}
\usepackage{graphicx}
\usepackage{algorithm}
\usepackage{algpseudocode}
\algnewcommand{\OR}{\textbf{or~}}
\algnewcommand{\AND}{\textbf{and~}}
\usepackage{mathrsfs}
\usepackage{amsmath}

\usepackage{amsthm}
\usepackage{multirow}
\usepackage{times}
\usepackage{url}
\usepackage{placeins}
\usepackage{wrapfig}
\usepackage{hyphenat}
\usepackage{booktabs}
\usepackage{subfig}

\algnewcommand{\Param}[1]{%
	\State \textbf{Parameter:}
	\Statex \hspace*{\algorithmicindent}\parbox[t]{.8\linewidth}{\raggedright #1}
}
\algnewcommand{\GlobalState}[1]{%
	\State \textbf{Global state:}
	\Statex \hspace*{\algorithmicindent}\parbox[t]{.8\linewidth}{\raggedright #1}
}
\algnewcommand{\Output}[1]{%
	\State \textbf{Output:}
	\Statex \hspace*{\algorithmicindent}\parbox[t]{.8\linewidth}{\raggedright #1}
}
\algnewcommand{\Inputs}[1]{%
	\State \textbf{Input:}
	\Statex \hspace*{\algorithmicindent}\parbox[t]{.8\linewidth}{\raggedright #1}
}
\algnewcommand{\Initialize}[1]{%
	\State \textbf{Initialize:}
	\Statex \hspace*{\algorithmicindent}\parbox[t]{.8\linewidth}{\raggedright #1}
}

\allowdisplaybreaks 

\begin{document}
\title{Mining Top-k Trajectory-Patterns from Anonymized Data
	}
%
%

\author{Anuj S. Saxena\inst{1} \and
Siddharth Dawar\inst{2} \and
Vikram Goyal\inst{2} \and
Debajyoti Bera\inst{2}
}

%

\institute{
National Defense Academy, Khadakwasla, Pune, India \\
\email{anujs@iiitd.ac.in}\\
 \and
Indraprastha Institute of Information Technology, Delhi, India\\
\email{\{siddharthd,vikram,dbera\}@iiitd.ac.in}}
\maketitle              
\begin{abstract} 
The ubiquity of GPS enabled devices result into the generation of an enormous amount of user movement data consisting of a sequence of locations together with activities performed at those locations. Such data, commonly known as {\it activity-trajectory data}, can be analysed to find common user movements and preferred activities, which will have tremendous business value. However, various countries have now introduced data privacy regulations that make it mandatory to anonymize any data before releasing it. This immediately makes it difficult to look for patterns as the existing mining techniques may not be directly applicable over anonymized data. User locations in an activity-trajectory dataset are anonymized to regions of different shapes and sizes making them uncomparable; therefore, it is unclear how to define suitable patterns over those regions. In this paper, we propose a  top-k pattern mining technique called TopKMintra that employs a pre-processing step to transform anonymized activity-trajectory into an intermediate representation that address the above problem. Unlike usual sequential data, activity-trajectory data is 2-dimensional that can lead to generation of duplicate patterns. To handle this, TopKMintra restricts arbitrary extensions in the two dimensions by imposing an order over the search space; this also helps in addressing the common problem of updating the threshold in top-k pattern mining algorithms during various stages. We also address the issue to raise the initial threshold that avoid exploring substantial amount of invalid candidates. We perform extensive experiments on real datasets to demonstrate the efficiency and the effectiveness of TopKMintra.
Our results show that even after anonymization, certain patterns may remain in a dataset and those could be mined efficiently.

\keywords{Trajectory Patterns, Sequential Pattern Mining, Data Privacy, Anonymized trajectory, Anonymized database.}
\end{abstract}
\section{Introduction}
\input{intro.tex}

\section{Related Work}
\label{sec:RW}
\input{rel-work.tex}

\section{Preliminaries, Problem Statement and Solution Approach}
\label{sec:pps}
\input{preliminaries.tex}

\section{Computational Framework for Mining Patterns}
\label{sec:TF}
\input{framework.tex}

\section{Top-K Pattern Mining Framework}
\label{sec:top-k}
\input{top-k.tex}

\section{Experimental Evaluation}
\label{sec:exp}
\input{exp.tex}

\section{Conclusion}
\label{sec:conclusion}
\input{conclusion.tex}

\section{Conflict of Interest}
The authors declare that they have no conflict of interest.

\bibliographystyle{unsrt}
\bibliography{pattern_mining_bib}
\end{document}

%% file: intro.tex
Are you a company in the business of data crunching? Do you believe in user's privacy to keep your users-base intact? Do you follow General Data Privacy Regulations to avoid heavy penalties? To remain in the business days are not far when partner companies would not provide their customer's data in raw format. Instead, they would first anonymize the data before sharing it due to privacy regulations. In the last decade, there has been a lot of work in the area of data privacy enforcement \cite{ardagna07,bamba08,beresford03,chow07,mokbel08,gruteser03,Hoh05,dummy05,Machanavajjhala07,Samarati01,xu08,You07}, and as a result, there are many tools available that can anonymize a given dataset to enforce privacy in a customizable manner. However, there is a lack of tools and techniques that can extract useful knowledge from the anonymized data \cite{Mintra16}, partly due to the challenge that a generic tool that cuts across multiple anonymization techniques and extracts a diverse set of information may not at all be possible.

Imagine travelers that are visiting several places as part of a single journey. An activity-trajectory data they generate consists of the geographic coordinates (referred to as spatial data) of the important locations visited by several people along with the textual descriptions of the activities (referred to as categorical data) partook by them at each location. Such datasets are generated by many service providers that track locations and are fairly commonplace due to the availability of mobile, GPS and internet services at very cheap costs. Erstwhile researchers have demonstrated the use of activity-trajectory data for a personalized recommendation \cite{ZHENG12}, user behaviour analysis \cite{BIERLAIRE08,Rai07,Zheng08}, traffic management \cite{Horvitz05,HERRERA10,Goh12}, disaster recovery \cite{MONTOYA03,MANSOURIAN06}, route recommendation \cite{Dai15,Kurashima10}, etc. However, growing user concerns about privacy threat in sharing too much personal information such as locations, activities performed, etc., while availing services jeopardize their future. This has forced business ventures to adopt privacy-preserving techniques while rendering service to protect the data-privacy needs of a user.

 A commonly used technique to achieve data-privacy is `anonymization'~\cite{ardagna07,Samarati01,gruteser03,nergiz2008towards} that hide sensitive information by generalizing data attributes and thereby delinking the precise information with the unique user identifier. The salient feature of these techniques is to group multiple similar records from the data into a publishable record that blur the sensitive information, which for our scenario are users' visited locations and activities performed at those locations.  What this means for us is that every published anonymized-trajectory is produced by grouping at least a minimum number of activity-trajectory sequences. We observe that most of the privacy-preserving techniques that are proposed along these lines have been designed to enforce a certain privacy notion whereas the other important goal of data utility is either overlooked or not considered practically. It is rather assumed that an optimal (or a sub-optimal) grouping of records which results into comparatively less information loss may provide reasonable utility. Due to this limitation, despite ease in ensuring privacy, the anonymization based techniques were phased out by differential privacy approaches which ensures privacy by adding noise using a predefined distribution. The advantage of the differential privacy over anonymization is that the former provides a formal guarantee over the data usage. However, ensuring an addition of a controlled noise for a differential private query is challenging and may not be possible always. This suggest that by proposing a data utility model for anonymized data, an alternate framework can be designed for those queries that may not be easily protected through differential privacy. Our proposed framework is one such model that can be used to assess the utility of the data in the context of pattern mining from an anonymized activity-trajectory data. In our mining framework, we assume that the data is anonymized using popular location privacy enforcement techniques, viz., location obfuscation~ \cite{ardagna07,gruteser03,OUR_ICISS}, $k$-anonymity~\cite{gedik08,nergiz2008towards,xu08}, $l$-diversity~\cite{Machanavajjhala07,Saxena13}, etc. Some of the current authors proposed a preliminary answer to this problem in an earlier work,  named Mintra \cite{Mintra16}, that could mine common patterns from an anonymized activity-trajectory data. In this paper, we propose a framework `TopKMintra' for mining \textit{top-k most relevant trajectory-patterns}.  To the best of our knowledge, no other work exists in this direction.

Our mining framework is based on the assumption that frequent patterns in the actual trajectory data will be contained in many anonymized trajectories, and therefore, can be extracted as a frequent common pattern from the anonymized data. For this it is required to define patterns over the anonymized data containing  actual trajectory patterns, and to provide a computational framework to measure their frequency in the anonymized data. These common patterns may have high resolution\footnote{The resolution of a sequence of the boxed-locations (or blurred locations) is defined as the average area of the boxed-locations.} and may not be of much use. Therefore, while computing the frequency of the patterns, we need to give more weight to the finer frequent patterns. To motivate further on the mining objective consider the following example.

\begin{figure*}[!ht]
	\begin{minipage}{.5\textwidth}
		\centering \includegraphics[width=0.9\textwidth]{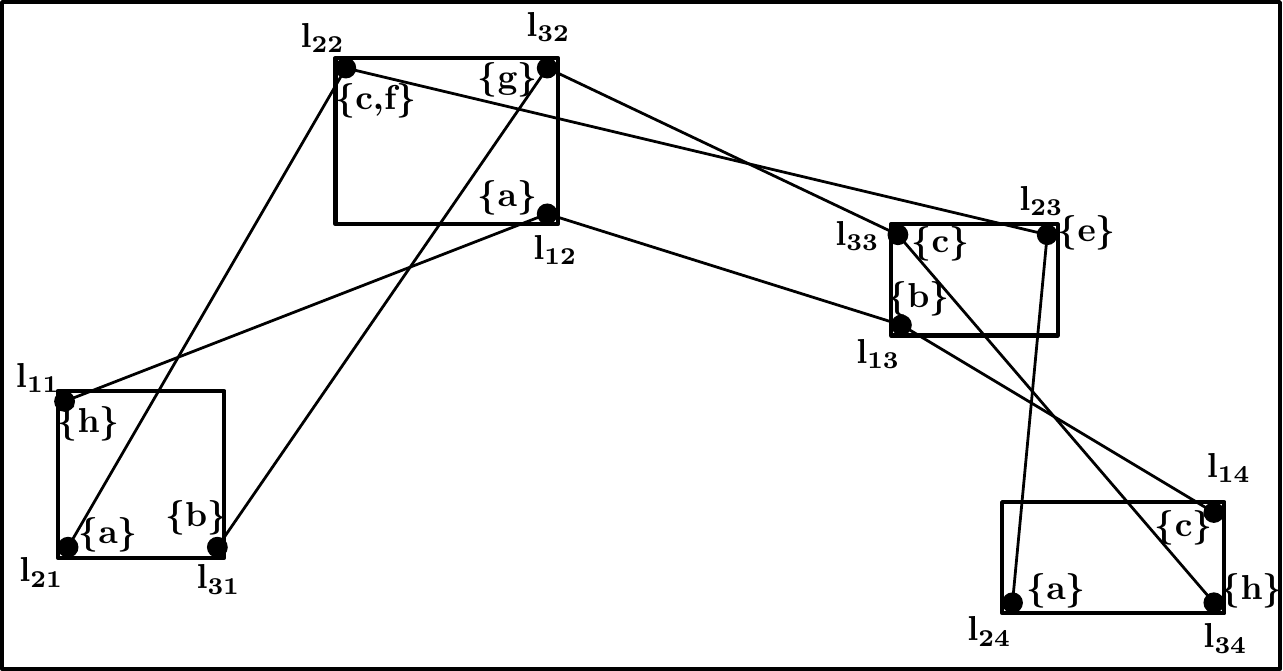}
		\caption{\label{fig:example} An anonymized annotated trajectory}
	\end{minipage}%
	\begin{minipage}{.5\textwidth}
		\centering
		\includegraphics[width=1\textwidth]{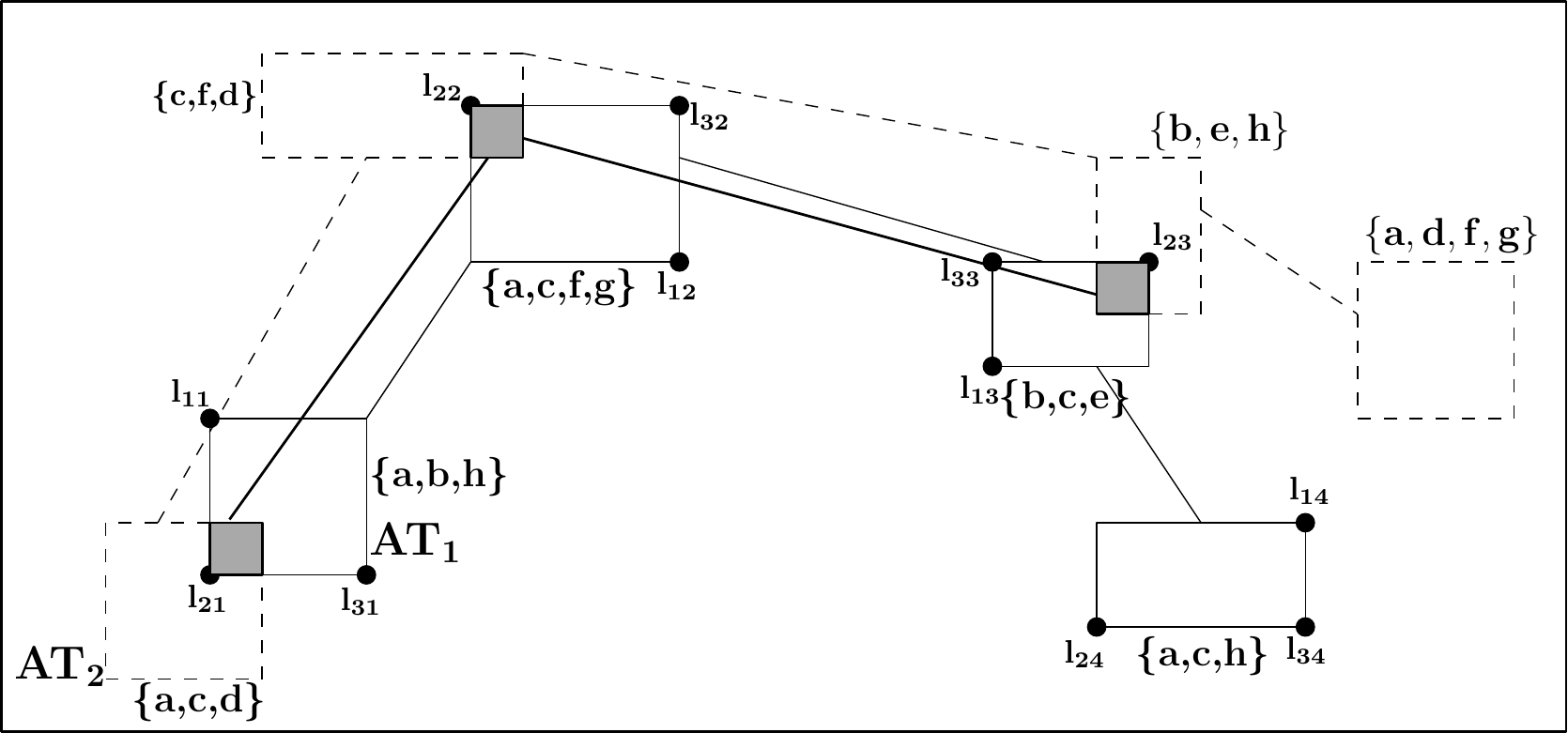}
		\caption{\label{fig:trajPattern} Example of a trajectory-pattern}
	\end{minipage}
\end{figure*}
\noindent \textbf{Example:} Fig.\ \ref{fig:example} shows an anonymized trajectory of length four obtained by grouping three users trajectories. The first record of an anonymized trajectory shows a {\it minimum bounding rectangle} (MBR) region of three user locations $\{l_{11},l_{21},l_{31}\}$ and their activities $\{\mathrm{a,b,h}\}$. After privacy enforcement, the database would have a set of such anonymized trajectories. The following questions arise regarding the utility of anonymized data: {\it Given previous route and current location of a user what activity and venue can be recommended to the user?}, {\it What routes are highly preferred by users who perform a certain set of activities?}, {\it Given a route query of a user what activities can be recommended to the user?} 

 We believe that the patterns can be mined from anonymized trajectory data that may help to answer questions similar to the given above. However, the accuracy of the answer would depend highly on the quality of the mined patterns. Fig.~\ref{fig:trajPattern} shows a common pattern that we intend to mine and in this case it is obtained from the two anonymized trajectories $AT_1$ \& $AT_2$ as a sequence of common shaded regions along with the most likely activities performed therein, e.g., $\{\mathrm{a}\}$, $\{\mathrm{c,f}\}$ and $\{\mathrm{b,e}\}$. If this pattern is contained in many other anonymized trajectories also, it is highly likely that many users prefer to visit a location in the first shaded region to perform activity `$\mathrm{a}$' then visit a location in the second shaded region and perform activity `$\mathrm{c}$' or `$\mathrm{f}$' and so on. We can also conclude that the activity `$\mathrm{a}$' takes place in the first shaded region -- a drastic improvement from the view portrayed by the anonymized trajectories $AT_1$ or $AT_2$. Our mining goal is to find such sufficiently finer patterns whose regions \& activities are part of many anonymized trajectories.

Trajectory data being sequential in nature, many sequence mining algorithms have been adapted for mining trajectory patterns. In general, the techniques for sequence mining can be classified into two categories, i.e., (1) frequent sequence pattern mining (FSM) \cite{han2000freespan,han2004mining,pei2001prefixspan}, and (2) high-utility sequence pattern mining (HSM) \cite{array_eswa,mmu,tseng2010up,inc_huspm,kais,yin2012uspan,icdm}. FSM techniques consider only the presence and absence of an item in a sequence and mines patterns with a frequency greater than a threshold. Whereas, HSM techniques consider importance/relevance of an individual item and mines high utility patterns. Both FSM and HSM have been adapted over trajectory data recently \cite{arya15,cao2005mining,giannotti2007trajectory,tsoukatos2001efficient}. However, the same can not be employed directly to mine patterns from anonymized data. There are several challenges to this problem compared to trajectory data mining.

{\bf Challenges:}   First, each record of an anonymized trajectory is in the form of a pair consisting of a spatial-region (a location is blurred to a spatial region) and an activity-set. It is not possible to compare two spatial regions for equality whereas they can be checked for containment or overlapping; this is in contrast to the classical FSM and HSM techniques where two items ought to be comparable. There are other problems arising of the anonymized nature of the data. Pattern mining techniques require a notion of ``weight of a pattern'' --- it is further desired that to support early termination weight should satisfy the downward closure property\cite{Agrawal1994} (discussed in detail in Section \ref{sec:top-k}); however, such a notion is not trivially defined for anonymized data. Secondly, the activity-trajectory data has a complex structure being two dimensional (namely, spatial and categorical). On the other hand, the existing FSM and HSM techniques consider only one-dimensional data, i.e., categorical data. It is, therefore, challenging to adapt existing pattern mining algorithms to mine activity-trajectory data. Thirdly, similar to the classical techniques for pattern mining, the trajectory-pattern mining over anonymized data has to deal with the issue of choosing a suitable threshold value for finding enough patterns.

We model the anonymized trajectory mining problem as a high-utility sequence pattern mining task. For this, we discretize anonymized regions as a set of smaller regions named as cells. Any mined location in a pattern can not be smaller than a cell, and therefore, defining the size of a cell is crucial which we will discuss in  section \ref{sec:summary}. However, the basic understanding of a cell is that it is considered as a possible region of activity. We associate an activity with a cell based on the frequency of the (cell, activity) pair in the anonymous data. In the cell representation of an anonymized region, we associate a weight with each cell that is the fraction of the cell region overlapped with the anonymized region. Assuming that activities within an anonymized region are uniformly spread over the region, we give a practical model to measure the total overlap ratio of a pattern over an anonymized data. We named this weight of a pattern as its relevance score. We discuss more detail on this in section \ref{sec:relevance}. All of this gives a score function that associates high score with patterns that are finer and are contained in many anonymized trajectories, thus convert the pattern mining over anonymized data into a high-relevance sequential mining.  

In a preliminary version of this work, the authors proposed a technique to mine all patterns~\cite{Mintra16} for a user provided threshold. However, that technique has a fundamental limitation that is common to pattern mining algorithms, i.e., they either produce a large number of patterns or very few patterns due to the difficulty in choosing a suitable threshold value for different datasets. Therefore, in this work, we propose a top-k solution for our scenario, namely TopKMintra. Many top-k versions of FSM and HSM algorithms have been designed that raise the utility threshold during the search space exploration. However designing similar techniques for a 2-dimensional data space such as ours is not straightforward. For this, we first adopt a couple of strategies proposed in top-k high utility pattern mining literature \cite{icdm,tks}, namely {\it Threshold Initialization (TI)} and {\it Threshold Updation (TU)}. The TI strategy basically pre-inserts some patterns as candidates to raise the initial threshold before starting the mining process. As the search space is not fixed and depends upon the varying threshold value thereby not starting with a zero threshold value a substantial amount of invalid candidates can be avoided. The TU technique maintains an order over the search space by adding those candidates early that may sharply increase the threshold value during the execution. This curtails the search space and may find top-k patterns efficiently. In addition, we have also adopted a couple of pruning strategies, namely {\it depth pruning} and {\it width pruning}. 

We experimentally evaluate the performance of TopKMintra over real datasets and study the impact of our proposed heuristics that are (i) TI, i.e.,initialization of threshold to prune the search space and (ii) TU, i.e., sorting of encoded small regions (cells) as per their relevance score across the database. 
We also compare TopKMintra with the state-of-the-art approach MintraBL which is an adapted version of the USpan algorithm \cite{yin2012uspan} designed for mining top-k patterns from the activity-trajectory data and also analyze the effect of the above two mentioned heuristics by adopting it in the baseline algorithm as MintraBL+i  (TI with baseline) and MintraBL+s (TU with baseline). 

\noindent {\bf Contributions:} The specific contributions of this paper are:
\begin{enumerate}
	\item We come up with a novel problem of mining anonymized activity trajectory data where the users trajectories are anonymized using popular $k$-anonymity and $l$-diversity privacy enforcement techniques.
	\item We introduce TopKMintra, a technique to mine top-k patterns from anonymized trajectories data. TopKMintra uses a spatial encoding technique to model the problem as a high-utility sequence pattern mining problem. It uses heuristics for efficiently initializing and updating threshold value and for early termination. 
	\item We conduct extensive experiments with real datasets to establish the efficiency of TopkMintra. The experimental study shows that TopkMintra outperforms all the three baseline algorithms in terms of memory usage, the number of candidates explored and absolute clock time.
\end{enumerate}

\noindent {\bf Outline:} The outline of the paper is as follows. Section \ref{sec:RW} presents related work in the area of data-privacy techniques (Section~\ref{sec:rw:dpt}), sequential pattern mining (Section~\ref{sec:rw:spm}) and top-k high utility sequential pattern mining (Section~\ref{sec:rw:top-k_huspm}). In Section~\ref{sec:pps}, we introduce the activity-trajectory data (Section~\ref{sec:pps:td}), the anonymized-trajectory data (Section~\ref{sec:pps:atd}), problem definition stating the mining goal over the anonymized data (Section~\ref{sec:pps:ps}) and  the solution approach (Section~\ref{sec:summary}) stating the need for the preprocessing of the data, patterns that we intend to mine and the weight function to measure the worth of a pattern. In Section ~\ref{sec:TF}, we first use the preprocessing steps to convert the data into the wLAS-sequence form (Section~\ref{sec:preprocessing}), and then we discuss the computational framework for measuring the worth of a trajectory-pattern in a wLAS-sequence data (Section~\ref{sec:TPM} and Section~\ref{sec:relevance}). Our proposed framework can be employed to incrementally compute the relevance of the child node from the relevance of the parent node.  We present the top-k pattern mining approach $TopKMintra$ in Section~\ref{sec:top-k} that uses a couple of efficiency strategies to reduce the search space. Finally, we present our experiment results in Section~\ref{sec:exp}  where we compare our top-k approach with three baseline approaches.  We conclude in Section~\ref{sec:conclusion} with the possible future directions.     
%

%% file: rel-work.tex
In this section, we first review studies on {\it data privacy techniques} with an emphasis on techniques protecting sensitive information in spatial queries through anonymization. As we intend to mine useful patterns from the anonymized-data generated through these privacy-preserving techniques, we next discuss models proposed for {\it sequential patterns mining} over spatial-data that might be close to our mining framework over anonymized data. Further, having identified the challenges of setting an appropriate threshold in data mining that can mine sufficient patterns, we aim for a top-k pattern mining framework. Thus a review of models on {\it the top-k framework for sequential pattern mining} is also briefly presented. 
 
\subsection{Data Privacy Techniques}
\label{sec:rw:dpt}
Enforcing privacy of personal information in public databases has been in focus since the 1990s~\cite{badge92}. The issue of privacy for a spatial database is very similar to the above in the sense that user's spatial queries contain sensitive information open to misuse, and in fact, research in location privacy adopt most of the concepts from the area of privacy in database publishing~\cite{Samarati01,Machanavajjhala07}. One among the initial problems that discuss the privacy over spatial data was proposed by Gruteser {\it et al.}~\cite{gruteser03} where they suggested modification in the spatial and the temporal information to achieve privacy. Since then many techniques have been proposed to protect the location information. Most of them mask the exact spatial location using one of the following: {\it imprecision} (i.e., by blurring the region)~\cite{dummy05,bamba08,gruteser03}, {\it inaccuracy} (i.e., by providing nearby landmark)~\cite{ji,MLY08} or {\it anonymization} (i.e., by generating bounding region containing several users)~\cite{chow06,casper06, gedik08}. This process of masking precise location makes it difficult to correlate a user with her location. 

Methods to protect privacy through imprecise or inaccurate locations have been discussed extensively by Chow {\it et al.}~\cite{mokbel08}. Kido{\it et al.}~\cite{dummy05} suggested introducing an imprecision by sending a set of false locations along with the exact user location. On the other hand, solutions that are based on anonymization do so based on the idea adopted from the database research: $k$-anonymity~\cite{Samarati01} and $l$-diversity~\cite{Machanavajjhala07}. Over a spatial database, the $k$-anonymity is achieved by reporting {\it a region as an anonymous location} of a user that contains at-least $k-1$ other users' location. The anonymization of location makes the location of the requesting user indistinguishable among $k$ user-locations in the anonymous region~\cite{chow06,hu09}. Over a spatial-textual date, the k-anonymity ensures location privacy, however, the activity information may get disclosed due to activity homogeneity attack. An improvement on $k$-anonymity that protects activity-privacy is proposed using  $l$-diversity~\cite{bamba08}, i.e., by maintaining at least $l$ different activities in the $k$-anonymous spatial query. Thus $k$-anonymity together with $l$-diversity may satisfy the location as well as the activity privacy needs of a user.

Mokbel {\it et al.}~\cite{casper06} proposed a grid-based algorithm to determine the optimal obfuscated area in his Casper framework. One of the notable works on location privacy through anonymization was the `Clique-Cloak' algorithm by Gedik {\it et al.}~\cite{gedik08} that suggested implementing personalized k-anonymity to optimize the obfuscated location information as per the individual privacy need. All of these models discuss the location privacy. Bamba {\it et al.}~\cite{bamba08} proposed the PrivacyGrid to preserve both $k$-anonymity and $l$-diversity.  Most of these earlier works address the snapshot query scenario where location information is disclosed one-time only for getting service. However, they may not necessarily ensure the privacy of a user who is moving continuously and reporting its location information from different places-- called as a {\it continuous query scenario.} 

The challenge in continuous query scenario is to ensure that the location of a user should not be exactly computable by correlating information from consecutive queries. Chow {\it et al.}~\cite{chow07} proposed two necessary properties namely \textit{cloaking region sharing} and \textit{memorization} that needs to ensure for providing privacy in continuous settings. 
As extensions of the notions of $k$-anonymity and $l$-diversity, other notions like historical $k$-anonymity and $m$-invariance impose other constraints of invariance of users and queries throughout the session, respectively. Dewri {\it et al.}~\cite{DewriRRW10} showed that for preserving query privacy in continuous query scenario by enforcing m-invariance is sufficient in the case of a location-aware adversary. However, for the case of location unaware adversary where location privacy is equally important, there is a need to enforce historical $k$-anonymity. In~\cite{xu08}, the authors consider past visited locations (footprints) of the user as well as trajectories of other users to anonymize the trajectory of the current user and enforce historical $k$-anonymity.

Beresford {\it et al.} tackle location privacy by proposing a concept called mix zones~\cite{beresford03}. Users in a mix zone are assigned identities in such a way that it is not possible to relate entering individuals with exiting individuals. Similar to mix zone, where user essentially does not get service to ensure location privacy, Saxena {\it et al.}~\cite{OUR_ICISS} considered skipping service once in a while towards the same goal. Their Rule-Based approach gave a heuristic which determined whether it is safe to take service or not at a particular location. Chow {\it et al.}~\cite{Chow2011} proposed a location anonymization algorithm that is designed specifically for the road network environment using $k$-anonymity, i.e., location is cloaked into a set of connected road segments of a fixed minimum total length and including at least $k$ users. Ghinita {\it et al.} \cite{GHINITA10} proposed protecting identity privacy by avoiding location based attacks and by using k-anonymity to mix user's identity with other nearby users. Lee {\it et al.}~\cite{lee11} presents a
location privacy protection technique that protects the location semantics by cloaking with semantically heterogeneous nearby locations. Xue {\it et al.}~\cite{Xue09} modeled the similar problem as location $l$-diversity by associating each query with at least $l$ different semantic locations, and proposed a solution by mapping observed queries to a linear equation of semantic locations. He formally justified privacy by relating privacy guarantee with existence of an infinite solutions in the resulting system of linear equations.

\subsection{Sequential Pattern Mining}
\label{sec:rw:spm}
Based on the current state-of-the-art pattern mining can broadly be divided into two classes-- frequent sequence pattern mining (FSM) and high-utility sequence pattern mining (HSM). The most essential and prior algorithm was proposed by Agrawal {\it et al.}~\cite{agrawal1995mining} in which they gave the {\it AprioriAll} algorithm which mines frequent patterns by candidates generation and test approach. The algorithm was based on {\it downward closure property} which states that a frequent pattern contains a sub-pattern that are in turn frequent. A list of algorithm based on data projection have been proposed which includes {\it FreeSpan} \cite{han2000freespan} and {\it PrefixSpan} \cite{pei2001prefixspan}. In these approaches, a sequence database is recursively projected into a set of smaller projected databases, and sequential patterns are grown in each projected database by exploring only locally frequent fragments. FreeSpan offers bi-level projection technique to improve performance. PrefixSpan is a projection-based pattern-growth approach for the mining of sequential patterns.  PrefixSpan outperforms FreeSpan in that only relevant postfixes are projected. Han {\it et al.}~\cite{han2000mining} propose a technique using a data structure called FP-Tree(Frequent pattern tree) which is an extended prefix tree structure for storing information required by depth-first search travel algorithm. SPADE algorithm by Zaki {\it et al.}~\cite{zaki2001spade} searches the lattice formed by id-list intersections and completes mining in three passes of database scanning. 
It may be noted that all these works focus mainly on itemsets sequences without any spatial dimension.

For high-utility itemset mining, several techniques have been proposed. Candidate generation and test-based techniques, such as Two-Phase~\cite{han2004mining} and IHUP~\cite{ahmed2009efficient}, run in two phases. Candidates generated in the first phase are being verified in the second phase. The efficiency of such an algorithm depends upon the candidate size in the first phase. Tseng {\it et al.}~\cite{tseng2010up}, in {\it UP-Growth} and {\it UP-Growth+}, have proposed an efficient data structure called UP-Tree to generate a candidate set that is much smaller than Two-Phase and IHUP. They have used heuristics to prune nodes and to decrease utility at nodes, and therefore results in less number of candidates.

Frequent sequence mining assumes that the items have binary presence or absence in the database and have equal importance. Ahmed {\it et al.}~\cite{ahmed} introduced the concept of high-utility sequence mining. High-utility sequence mining associates different importance with items and items occur in a sequence with different quantity. The authors proposed a level-wise algorithm UtilityLevel and a pattern growth approach UtilitySpan for mining high-utility sequential patterns in three database scans. UtilitySpan scans the database to find the 1-length promising sequences. The 1-length sequences are extended using the pattern growth approach to generate candidate high-utility sequences. A third database scan is performed to find the exact high-utility sequences from the candidates. Yin {\it et al.}~\cite{yin2012uspan} introduced the lexicographic tree structure to store sequences and extension mechanisms called i-concatenation and s-concatenation to grow the prefix pattern for the proposed USPAN algorithm. The authors proposed some pruning strategies to reduce the search space compared to UtilitySpan. Aklan {\it et al.}~\cite{crom} proposed another upper bound called Cumulated Rest of Match (CRoM) and an algorithm called HuspExt which estimates the utility of future extensions based on the current prefix. Le {\it et al.}~\cite{array_eswa} propose an array-based structure which takes less memory compared to the utility-matrix structure used by the USPAN algorithm. The authors propose an algorithm called AHUS-P which utilizes shared memory to parallelize the mining task. Some work has also been done on mining high-utility sequential rules \cite{seq_rules}, mining sequential patterns with multiple minimum utility thresholds \cite{mmu}, and incremental high-utility sequential pattern mining \cite{inc_huspm}.   

Koperski {\it et al.}~\cite{koperski1995discovery} extended the concept given by Srikanta {\it et al.} to spatial databases by introducing spatial predicates like close\_to and near\_by. Tsoukatos {\it et al.}~\cite{tsoukatos2001efficient} proposed an apriori-based method for mining spatial regions sequences. Cao {\it et al.}~\cite{cao2005mining} proposed a trajectory mining algorithm. Chung {\it et al.}~\cite{chung2004evolutionary} proposed an Apriori-based method to mine movement patterns, where the moving objects are transformed using a grid system, and then frequent patterns are generated level by level. Giannotti {\it et al.}~\cite{giannotti2007trajectory} focuses on a method for extracting trajectory patterns containing both spatial and temporal information. These works focus mainly on spatial and temporal dimension. Mining over Spatial-Textual data is being proposed for the first time by Arya {\it et al.}~\cite{arya15}. They have studied mining over the two dimension data in three ways -- namely Spatial-Textual (Spatial in the first Phase then Textual), Textual-Spatial (Textual in the first phase then Spatial) and Hybrid (Both dimensions simultaneously).

\subsection{Top-k High-Utility Sequential Pattern Mining}
\label{sec:rw:top-k_huspm}
The number of high-utility sequences output by the sequential pattern mining algorithm depends on the minimum utility threshold for a dataset. A large number of sequences can be produced if the threshold is set too low and vice-versa. The characteristics of the database need to be analyzed by the user to set a threshold for getting a reasonable number of patterns. Yin {\it et al.}~\cite{icdm} introduced the problem of mining top-k high utility sequential patterns and an algorithm based on USPAN. Top-k mining algorithms start with a threshold zero and update it during the mining process only. However, it is beneficial to raise the threshold during pre-processing to remove unpromising items before exploring the search space. Yin {\it et al.} proposed a pre-insertion strategy to insert 1-length and full sequence in the top-k buffer. The order in which search space is explored doesn't matter for high-utility sequence mining as result will remain the same. However, the order matters for top-k algorithms as processing a branch in the search space with a higher possibility of getting a top-k pattern will raise the threshold quickly. The authors proposed another strategy to sort the extensions of a sequence in decreasing order according to Sequence-Projected Utilization. 

Wang {\it et al.}~\cite{kais} proposed an algorithm called TKHUS-Span to find top-k high-utility sequential patterns. Search strategies like Guided Depth-first Search, Best-first search, and Hybrid search are proposed to intelligently explore the search space. The authors design a utility-chain structure for faster computation compared to the lexicographic tree employed by USpan. A prefix extension utility strategy is also proposed to bound the utility of any extension of the current prefix.

%% file: preliminaries.tex
This section presents preliminaries, defines the problem of top-k trajectory-pattern mining over an anonymized activity-trajectory data and present a solution approach. For this, we first discuss activity-trajectory data and its anonymization. Next, we discuss the sequential patterns that are of interest over the activity-trajectory data. However, dissociation of location and activity information in the anonymized data makes the mining of such a pattern challenging. To address this issue, we suggest a solution that encodes an anonymous-trajectory data as a weighted-location activity-set trajectory data (detail in section \ref{sec:preprocessing}), and then by defining trajectory-patterns associate a weight with the patterns in the encoded data. The weight of a trajectory-pattern is named as its relevance score. This reduces the problem of finding useful patterns from the anonymized-data into {\it high relevant pattern mining} over encoded data. We now present the necessary theory that is required for understanding our mining framework.

\subsection{Trajectory Database}
\label{sec:pps:td}
Let ${\cal L}=\{l_1,l_2,\ldots l_m\}$  denote the collection of all the locations on the map at which some of the activities from the set $I=\{i_1,i_2,\ldots, i_n\}$ can be performed. For us, a location $l \in {\cal L}$  is a discrete spatial point $(latitude,longitude)$ that represents a point of interest such as a store, a restaurant, a movie theater, etc. A user is allowed to perform multiple activities at one location. For example, whilst in a shopping mall, one can perform activities such as purchasing items from an outlet, watching a movie in a multiplex, ordering eatables, etc. The set of activities performed by a user at a location is named as {\it activity-set} and is denoted by $X (\subseteq I)$. Thus, an {\it activity-trajectory} of a user represents her movement along locations from ${\cal L}$ together with activities performed at those locations. Formally, 
\begin{definition}[Activity-Trajectory]
	An activity-trajectory ${\cal T}$ of length $t$ is a sequence of (location, activity-set) pair i.e., $${\cal T}= \langle (l_1,X_1) (l_2,X_2) \ldots (l_t,X_t) \rangle \qquad l_i\in{\cal L} \mbox { and } X_i \subseteq I \mbox{ for all } i\in\{1,2,\ldots,t\}$$
\end{definition}
For pattern mining, two patterns having the same (location, activity) sequence but different time-stamp are considered same, and therefore, we do not associate a time-stamp with the (location,activity-set) pair in the activity-trajectory. Our focus in this work is on activity-trajectory datasets only and therefore, for the sake of brevity in the rest of the paper, we will call an activity-trajectory as a  {\it trajectory} and an activity-trajectory database as a {\it trajectory database} (denoted by $\cal D$). In Fig.~\ref{fig:example}, $\langle(l_{11},\{\mathrm{h}\})\;(l_{12},\{\mathrm{a}\})$ \\$\;(l_{13},\{\mathrm{b}\})\;(l_{14},\{\mathrm{c}\})\rangle$ is one such trajectory of a user. 

Nowadays, users are more concerned about their privacy and prefer disclosing sensitive information like locations,  activities, etc., only in a privacy-preserving manner. For example, a user may consider not disclosing an activity associated with a particular location that may reveal her private association such as health status, liking, etc. To satisfy such privacy needs, it is desired that the database should be published only after anonymizing the sensitive parameters. 
\subsection{Anonymized Trajectory Database}
\label{sec:pps:atd}
We assume that there exists a privacy-preserving technique such as {\it grid-based location blurring~\cite{Gidofalvi07}, $k$-anonymity~\cite{nergiz2008towards,Abul08,Samarati01}, $l$-diversity~\cite{Saxena13, bamba08}}, etc., that anonymize trajectory-data to satisfy the privacy need of a user. We denote the {\it anonymized trajectory}  by ${\cal AT}$ and the {\it anonymized trajectory database} by ${\cal PD}$.
\begin{definition}[Anonymous-Trajectory]
\label{anontraj_d}
For a trajectory ${\cal T} \in {\cal D}$, where $ {\cal T} = \langle (l_i, X_i)\rangle _{i=1}^n$, its anonymization is  $${\cal AT}= \langle (AL_k,AX_k)\rangle_{k=1}^m~\qquad(m\le n)$$ where $AL_k$ is a spatial-region (called as anonymous location) and $AX_k$ is a set of activities (called as anonymous activity-set), if there exists a sub-trajectory ${\cal T}^\prime=\langle (l_{i_k}, X_{i_k})\rangle _{k=1}^m$ of ${\cal T}$ such that the following are satisfied:    
	\begin{enumerate}
		\item {\it (Location containment constraint)} Every location of the sub-trajectory ${\cal T}'$ must be contained in the corresponding anonymous location of ${\cal AT}$ i.e.,  $l_{i_k} \in AL_k$ for each k such that $1\le k \le m$.
		\item {\it (Location privacy constraint)} Each anonymous-location $AL_k$ should satisfy the location privacy criteria as stated by the privacy definition.
		\item {\it (Activity containment constraint)} Every activity-set of the sub-trajectory ${\cal T}'$ must be contained in the corresponding anonymous activity-set of ${\cal AT}$, i.e.,  $X_{i_k} \subseteq AX_k$ for each $k$ such that $1\le k \le m$.
		\item {\it (Activity privacy constraint)} Each anonymous activity-set $AX_i$ should satisfy the activity privacy criteria as stated by the privacy definition. 
	\end{enumerate}	
\end{definition}
Location and activity privacy constraints, as in the point $2$ and in the point $4$ above,  
may vary for different privacy enforcement techniques. For example, anonymization satisfying $k$-anonymity as a location privacy constraint and $l$-diversity as an activity privacy constraint requires that each anonymous-trajectory ${\cal AT}$ is an anonymization of a collection of trajectories from ${\cal D}$ (called an {\it anonymity-set}) such that each anonymous-location $AL_i$ in ${\cal AT}$ contains at least $k$ distinct locations from trajectories in the anonymity-set and each anonymous activity-set $AX_i$ contains at least $l$ distinct activities along those locations in $AL_i$. 
In Figure\ \ref{fig:example},   $\langle(AL_1,\{\mathrm{a,b,h}\})$ $(AL_2,\{\mathrm{a,c,f,g}\})$ $(AL_3,\{\mathrm{b,c,e}\})$ $(AL_4,\{\mathrm{a,c,h}\})\rangle$ is a 3-anonymous and 3-diverse activity-trajectory. Here $AL_1$ (shown as a rectangular box) is a minimum bounding rectangle ($MBR$) that anonymize locations $l_{11},l_{21}$ and $l_{31}$ along the three user trajectories and $AX_1=\{\mathrm{a,b,h}\}$ is the respective anonymous activity-set such that $|AL_1|\ge 3$ and $|AX_1|\ge 3$. Similarly, it can be seen that the other anonymous location and anonymous activity-set also satisfies the 3-anonymity and the 3-diversity condition. $k$-anonymity and $l$-diversity ensure the location and the activity privacy by making it difficult for an adversary to associate an actual location and an actual activity with the user trajectory based on the disclosed anonymized dataset with a high probability. 

Though anonymization provides privacy, it affects the data-usability. For example, mining useful information from anonymized activity-trajectories such as frequently visited trajectories for path recommendation, frequent (location,activity) pair for point of interest, etc., may not be trivial due to dissociation of location and activity attributes. Anonymization based privacy-preserving techniques (as in Definition~\ref{anontraj_d} above)~\cite{Abul08, bamba08,DewriRRW10, Gidofalvi07, Hoh05, nergiz2008towards, xu08,Saxena13, Machanavajjhala07, You07, kalnis07} mostly overlooked this usability part where it is assumed that reasonable utility can be achieved by keeping an optimal or a suboptimal blurring level while anonymizing. It remains mostly unclear whether the anonymized data is of much use.

\subsection{Problem Statement: Knowledge Mining over Anonymous-trajectory}
\label{sec:pps:ps}
An ideal mining goal over a trajectory data would be to retrieve a sequence of (location, activity) pairs taken by many users, i.e., `frequent activity-trajectory patterns'. As discussed,  dissociation of location, activity and their sequentiality in anonymized data makes the mining of such patterns challenging. For example, from data containing anonymous trajectories similar to $AT_1$ in Figure \ref{fig:trajPattern}, it is not trivial to mine a frequent pattern such as $\langle (l_{21},a)(l_{22},c)$ $(l_{23},e)\rangle$. However, {\it if such a trajectory-pattern is frequent, it must be a part in many other anonymous trajectories also.} This raises a question as to `{\it how can we mine frequently visited sequential-pattern that is common across several anonymous-trajectories in the anonymized data}'. 

We now give an outline of our mining framework. We assume that there exists an anonymization technique which generates anonymized data in the form as defined in Definition \ref{anontraj_d}. We also emphasize that the proposed mining framework do not link user or its trajectory with the mined patterns with any better probabilistic guarantee than in the anonymized data, i.e., the privacy of a user remains unaffected with the mined knowledge that predict frequent location-activity pairs and their sequentiality. 

\subsection{Solution Approach: Summary of the framework}
\label{sec:summary}
Our solution for pattern mining from an anonymized data is based on the assumption that the actual patterns of an activity-trajectory data that are frequent will be contained in many anonymized trajectories after anonymization, and therefore, can be extracted as a common knowledge within those anonymized trajectories. The challenges, however, is to come up with a computational framework that can mine common regions having similar activities, and the sequences of such frequent region-activity pairs from the anonymized data. As the actual user activity-trajectory data is not known due to anonymization, the best we can do is to predict the most likely patterns of location-activity pair from the anonymized data that are taken by many users. For such a pattern-mining task over an anonymized-trajectory data, first we propose an encoding of a data into a {\it weighted Location-Activity-Set sequence (wLAS-sequence)} format, and thereafter, by defining patterns over the wLAS-sequences, we propose its relevance score as a cumulative overlap ratio of the pattern over the trajectories in the anonymized data. The mined patterns from the anonymized data that have high total overlap ratio (or relevance score) are most likely the patterns containing the actual frequent location-activity sequential patterns from the users activity-trajectory data.  We now discuss the steps of our solution approach.

We first encode the anonymous data into wLAS-sequence data. For this, we discretize the region ${\cal R}$, containing all the anonymous-trajectories in ${\cal PD}$, by partitioning it into rectangular cells--called as cell-partition of ${\cal R}$. We denote the collection of all the cells in the region ${\cal R}$ by ${\cal P}=\{p_1,p_2,\ldots, p_n\}$. Formally,
\begin{definition}[Partition of a region] A finite collection ${\cal P}=\{p_1,p_2,\ldots p_n\}$ of cells in ${\cal R}$ is a partition of ${\cal R}$ if 
	\begin{itemize}
		\item (cells are mutually disjoint) $~p_r \cap p_s = \phi$, for any  $r,s \in \{1,2,\ldots n\}, r\not=s$. 
		\item (cells cover the whole region ${\cal R}$) $~\cup_{k=1}^n p_k = {\cal R} $ 
	\end{itemize}
\end{definition}
The cells are fairly small in size as compared to the anonymous locations and represent the basic unit of location, i.e., {\it the location in a mined trajectory-patterns cannot be finer than the cells in the discretized space.} This also means that the location of an activity that we mine in a predicted pattern can at-least be of the cell-size. The size of the cell can be precomputed either using the actual activity-trajectory data before anonymization and is provided with the anonymized data for the mining task or it may be computed by using the trade-off between pattern quality and mining time over the anonymized data. If it is precomputed over actual activity-trajectory data, the size should be tuned so as to avoid the closest reported consecutive locations of the user or the nearby activities on the ground to lie in the same cell. However, we only avoid it whereas some of the locations or the activities may lie in the same cell for the chosen size. This is due to the existence of closeby GPS locations (when a user is almost still) in the activity trajectory data or extremely nearby activities on the ground. In that case, those locations or the activities that are in the same cell cannot be seen independently as a frequent pattern due to the modeling constraint. Further, if we tuned cell-size over anonymized data, for a smaller cell-size, the patterns we get are finer. But if we keep cell-size too small, it may only increase the computational complexity of the pattern mining without increasing the quality of the mined patterns. This is because on reducing the size of the cell even further we may not get substantially finer patterns as compared to patterns with bigger cell-size; however, the computational cost for the pattern-mining is inversely related to cell-size (as shown experimentally also). 

For our high-relevant pattern mining framework, we restrict ${\cal P}$ as a partition of the region ${\cal R}$. However, for other mining task formalization of ${\cal P}$ as collections of cells over ${\cal R}$ that cover ${\cal R}$ but may not necessarily disjoint, or a collection of finite disjoint cells that represents point-of-interest in the region ${\cal R}$ but may not cover whole ${\cal R}$, etc., can be analyzed as  separate problems.\\

\noindent We now discuss the weighted Location Set representation of an anonymous location.
\begin{definition}[weighted Location Set (wLS)]
\label{def:cell-weight}
For a discretization $ {\cal P}=\{p_1,p_2,\ldots, p_n\}$ of ${\cal R}$, the {\it weighted Location Set} (wLS) representation of an anonymous location $AL$ is $$\{p_{i_1}:w_{i_1},p_{i_2}:w_{i_2},\ldots, p_{i_k}:w_{i_k}\}$$ 
where $w_{i_s}$, for $s=1,\ldots, k$, denotes the (non-zero) weight of the cell $p_{i_s} \in {\cal P}$ having overlap with $AL$.
\end{definition}
The weight of a cell in wLS representation must be defined based on the knowledge that we intend to mine from the anonymized data. For predicting frequent patterns from an anonymized data, one of the way to define the weight of a cell in the anonymous location $AL$ is the chance of the user being present in the cell given that she is in $AL$. If the cumulative weight of a pattern across trajectories in the anonymized data is high, it is most likely a high probable frequent pattern. While reporting cells for $AL$, it is meaningful to report only those cells that have a non-zero weight, which in our case are the cells having non-empty overlap with the anonymous location. Also, for computational requirement, we enforce that the cells in the wLS representation are ordered over the cells index.  A formal discussion on the weight of the cells for our mining task is presented in Section \ref{sec:preprocessing}. We now discuss the weighted Location Activity Set (wLAS) sequences.

\begin{definition}[weighted Location Activity Set (wLAS) sequence]
	An anonymized trajectory ${\cal AT}=\langle (AL_i,AX_i)\rangle_i$ is said to be a wLAS-sequence if each of its anonymized locations $AL_i$ is in a wLS form. 
\end{definition}
\noindent For computational requirements, we restrict that the activities in the activity-set $AX_i$ are sorted in lexicographic order. Table \ref{table:wLAS-seq-rep} shows an example of an anonymized trajectory in a wLAS-sequence representation.  Since the weights are chances in our framework, their values are in between $0$ and $1$. However, it is not a requirement for our mining task to keep weights as a fraction. If required, they can be scaled to an integer by multiplying with an appropriate factor $N$. In this example, $N=18$ can convert weights into integer weights.
\begin{table}[H]
	\begin{center}
		\begin{tabular}{||l l l ||} 
			\hline
			$wLAS$-Sequence & ${\cal AT} = \langle\alpha_1,\alpha_2,\alpha_3, \alpha_4\rangle$ & \\ [0.5ex] 
			\hline\hline
			$wLAS$-Terms 
			& $\alpha_1 =(AL_1,AX_1)$ & $AL_1 \{p_1:\frac{2}{9},$
			~{\boldmath${ p_2}$} $:\frac{1}{9}, ~p_5:\frac{4}{9},$
			~{\boldmath$p_6$}$:\frac{2}{9}\} $ \\
			&& $AX_1=\{\mathrm{a,b,h}\}$\\ 
			\cline{2-3}
			& $~\alpha_2 =(AL_2,AX_2)$ &   $AL_2=\{p_{11}:\frac{1}{3}, 		 ~p_{12}:\frac{1}{6}, ~p_{15}:\frac{1}{3}, 		 ~p_{16}:\frac{1}{6}\}$\\ 
			& &   $AX_2=\{\mathrm{a,c,f,g}\})$ \\
			\cline{2-3}
			&$~\alpha_3 =(AL_3,AX_3)$ &  $AL_3=\{p_{22}:\frac{1}{6},
			~p_{23}:\frac{1}{6}, ~\mbox{\boldmath{$p_{26}$}}:\frac{1}{3},
			~p_{27}:\frac{1}{3}\}$ \\ 
			& &  $AX_3=\{\mathrm{b,c,e}\})$ \\
			\cline{2-3}
			&$~\alpha_4 =(AL_4,AX_4)$ &  $AL_4=AL_4=\{p_{25}:\frac{1}{2},~\mbox{\boldmath{$p_{26}$}}:\frac{1}{2}\}$ \\ 
			& &  $AX_4=\{\mathrm{a,c,h}\})$ \\
			\hline\hline
			
		\end{tabular}
		\caption{ wLAS-sequence Representation of Anonymous Trajectory}
		\label{table:wLAS-seq-rep}
	\end{center}
\end{table}

\noindent We now formally define the trajectory-patterns that we intend to mine from the anonymized data.  Our example of a trajectory-pattern is the one as shown by the shaded common information in the figure \ref{fig:trajPattern}.
\begin{definition}[Trajectory-Pattern]
	A trajectory-pattern, denoted by the symbol $T$, is a sequence of (cellSet, ActivitySet) pair, i.e., $T= \langle (P_1,X_1), (P_2,X_2), \ldots , (P_n,X_n) \rangle$, where each $P_i$ and $X_i$ denote the collection of cells and the activities respectively.
\end{definition}
For example, $t=\langle~ (\{p_2,p_{6}\},\{a,b\}) ~ (p_{26},c)~ \rangle$ is one such  trajectory-pattern. We denote by the size of $P_i$ (or $X_i$) as the number of cells (or activities) in the collection. 

We are now ready to discuss the relevance-score of a pattern in an anonymized data that is in the wLAS-sequence form. The purpose of the relevance-score is to capture the common knowledge across trajectories in the wLAS-sequnece form, and therefore, assist in the pattern mining. Formally, we require that the relevance of a trajectory pattern should satisfy the following conditions.    



\begin{definition}[High-Relevance trajectory pattern] \label{defn:high-relevance} A trajectory-pattern $T=\langle (P_i,X_i) \rangle_i$ is of {\it high-relevance} if 
	\begin{enumerate}
		\item {\it The association between locations and activities within the trajectory-pattern is highly likely}, i.e.,  based on the information within the anonymized data it is more likely that the activities in $X_i$ are performed at some of the locations in $P_i$ for each $i$.
		\item {\it The patterns having relatively smaller size of $P_i$ and $X_i$ (for each $i$) are comparatively more relevant}. For large size of $P_i$ and $X_i$, the trajectory-pattern is just like an anonymized trajectory, and therefore, not of much use.
		\item {\it The pattern must be a frequent one in an anonymized trajectory database}, i.e., the mined knowledge must be present in several anonymized trajectories as a common knowledge. 
		  
	\end{enumerate}
\end{definition}

The first two points in the above definition define the relevance of a trajectory-pattern whereas the third point discusses its significant presence across different anonymized trajectories in the database. We formally define the relevance-score of a trajectory-pattern in a $wLAS$-sequence database later in Section \ref{sec:relevance}. Informally, the relevance-score of a cell-activity pair in a wLAS-sequence is the chance of the activity being performed at the cell location based on the information in the wLAS-sequence. For a trajectory-pattern, its relevance in a wLAS-sequence is the cumulative chance of individual pairs based on the sequentiality knowledge of the pairs in the wLAS-sequence, and the relevance in a wLAS-sequence data is the sum of the relevance over matching wLAS-sequences. Patterns for which the relevance-score in the anonymized data is higher than the prespecified threshold are considered {\it relevant-patterns}. 

We define the relevance score of a (cell, activity) term in a matching wLAS (anonymized region) using the cell-weight (Definition~\ref{def:cell-weight}), which is the chance of an activity being performed in the cell.  For computing the relevance of a pattern in a wLAS-sequence, we aggregate weights of  individual matching terms; and over an anonymized data, we aggregate relevance in the matching wLAS-sequences of the data. We discuss more detail about the relevance computation in Section~\ref{sec:relevance}, Definition~\ref{defn:relevance}, however, our proposed relevance score satisfies the three conditions as mentioned in the definition above. The aggregation of relevance over the matching sequences in the data ensures point (2) and point (3) as above. The finer patterns are expected to have comparatively more matching in anonymized-trajectories that increases its frequency, and being frequent its relevance in different matching trajectories get added to make it high. However, the point (1) is satisfied by the choice of the relevance score as the cell weight.

We now explain relevance of a pattern by an example. For a trajectory-pattern $t_1=(\{p_2,p_6\},a)$, its relevance in ${\cal AT}$ (as in Table \ref{table:wLAS-seq-rep}) is the maximum of the weights among all the possible matches of $t_1$ with different terms of ${\cal AT}$, i.e., $\alpha_1,\alpha_2,\alpha_3$ and $\alpha_4$. Since the match of $t_1$ is in $\alpha_1$ only, we define $w(t_1,\alpha_1) = w(p_2) + w(p_6) = 1/3$ (as $\mathrm{a} \in AX_1$), and therefore $w(t_1,{\cal AT})=\max{w(t,\alpha_1)=1/3}$. Similarly, for the weight of a pattern $t=\langle t_1 t_2\rangle$ where $t_1$ is as before and $t_2=(\{p_{26},c\})$, we get that second term $t_2$ has a match in $\alpha_3$ and $\alpha_4$ both, and $w(t_2, \alpha_3)=1/3, w(t_2,\alpha_4)=1/2$. Therefore, $w(t,{\cal AT})=max\{w(t_1,\alpha_1)+w(t_2,\alpha_3),~ w(t_1,\alpha_1)+w(t_2,\alpha_4)\}=\{1/3+1/3, 1/3+1/2\}=5/6$. The weight of a pattern in an anonymized database consisting of several such trajectories is the cumulative weight of the pattern in different trajectories. 


This transformation of an anonymized data into a wLAS-sequences data maps the problem of finding `patterns from an anonymized data' into an equivalent problem of `high-relevance pattern mining' over a two-dimensional wLAS-sequence data. In pattern mining, finding an appropriate threshold that can mine enough number of relevant-trajectories is data specific and is not easy to decide beforehand. 
To overcome this difficulty, in this paper, we propose a {\it top-k approach for mining high-relevant trajectory patterns from a wLAS-sequence data}. We adopt a couple of efficiency strategies for our proposed algorithm $TopKMintra$, such as {\it Threshold Initialization Strategy (TI)} and {\it Threshold Updation (TU)}. Our computational framework is motivated by the $USpan$ \cite{yin2012uspan}, a technique that mines high utility sequential patterns from one-dimensional itemset sequences. Similar to $USpan$, for early termination, we adopt a couple of pruning strategies, namely {\it Depth Pruning} and  {\it Width Pruning}. \\ \ \\
\noindent The outline of the rest of the sections in light of the abovementioned solution approach are as follows:
\begin{itemize}
	\item First, we propose a {\it preprocessing step} (Section \ref{sec:preprocessing}) that converts an anonymous-trajectory data into a corresponding wLAS-sequence data.  This assists in mining `high-relevant trajectory-patterns' from anonymized data.
	\item The {\it relevance-score} of a {\it trajectory-pattern} over the processed data is defined in Section \ref{sec:relevance}. The proposed score in Definition~\ref{defn:high-relevance} quantify both the relevance and the repetitiveness of the common information within different wLAS-sequences. 
	\item  We present an efficient top-k trajectory mining approach TopKMintra in Section \ref{sec:top-k} with two effective efficiency strategies, namely  {\it Threshold Initialization (TI)} and {\it Threshold Updation (TU)} and two pruning strategies, namely {\it Depth Pruning} and  {\it Width Pruning}.
	\item Finally, our experimental results are presented in Section \ref{sec:exp}. As there exists no algorithm that discusses mining over anonymized trajectory data, we compare TopKMintra with three algorithms which are designed using our baseline algorithm $MintraBL$ (Section \ref{sec:mintraBL}) by combining with two efficiency strategies $TI$ and $TU$.
\end{itemize}

%% file: framework.tex
As a preprocessing step, we convert an anonymized trajectory data into a wLAS-sequence data for a given cell-partition which is precomputed using the actual trajectory data and is known in advance. We suggest the cell-partition to be as much refined+ over activities as possible. At the same time, we do not restrict having more than one activities at a single cell. By separating nearby activities over different cells, we create cells with a very small size. This increases the size of $P_i$ in wLS representation, and therefore, increases the computational complexity of pattern mining, i.e., there is a trade-off between the finer patterns and the computational complexity to mine them. Note that reporting nearby activities at a single cell is equally useful information in such a scenario. Therefore, we suggest deciding the size of the cells keeping both of these constraints in mind.  One such cell-partitioning is shown in Figure \ref{fig:prepro} that divide the region into cells ranging from $P_1$ to $P_{32}$. The cells in the figure are taken quite bigger to keep our example size in the subsequent discussion smaller. 

\subsection{Preprocessing Step: wLAS-sequence data}
\label{sec:preprocessing}

\begin{figure}
	\centering
	\includegraphics[width=.6\textwidth]{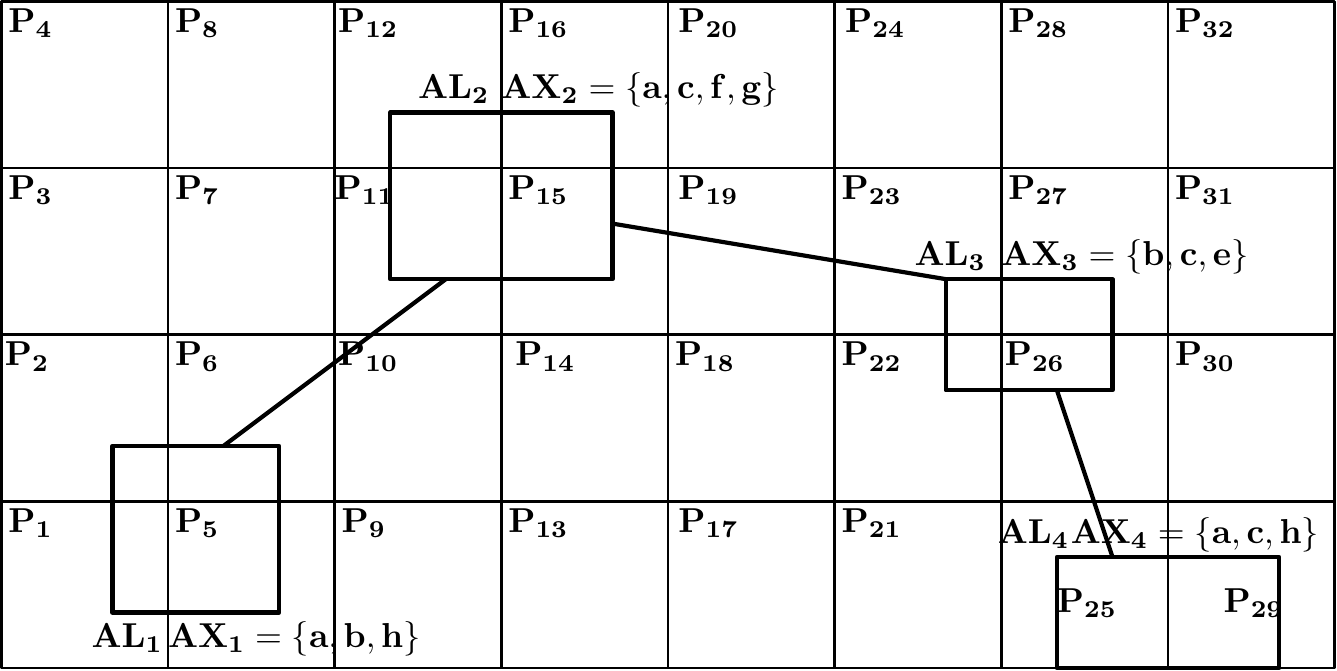}
	\captionof{figure}{Pre-processing using grid of square cells}
	\label{fig:prepro}
\end{figure}
Let $AL_i$ be a MBR of the $i^{th}$ anonymous location in a given anonymous trajectory that overlaps $n$ cells $p^1_i,p^2_i, \ldots p^n_i$ in ${\cal P}$. We define the weight $w^j_i$ of a cell $p^j_i$ within a MBR as the ratio of area of overlap region between the cell and the MBR, i.e.,
\begin{center}
	$w^j_i =\frac{\mbox{area of overlap region between $AL_i$
			and $p^j_i$}}{\mbox{area of $AL_i$}}$
\end{center} 
The weighted locationSet representation of $AL_i$ (aka. wLS form), therefore, is
\begin{center}
	$AL_i \equiv \{p^1_i:w^1_i, p^2_i:w^2_i, \ldots, p^n_i:w^n_i \}$	
\end{center}
\noindent  For convenience, locations in weighted representation will always be ordered in the increasing order of cell-index. Consider an example of anonymized trajectory as in Figure~\ref{fig:prepro} where anonymized-region $AL_1$ has a non-empty overlap with cells $p_1,p_2,p_5$ and $p_6$. Based on the overlap-ratio of these cells with $AL_1$, we have $w_1=\frac{2}{9}, w_2=\frac{1}{9}, w_5=\frac{4}{9}, w_6=\frac{2}{9}$. The wLS representation of $AL_1$, wLAS of $(AL_1,AX_1)$ and wLAS-sequence for anonymous-trajectory in Figure~\ref{fig:prepro} are shown in Table~\ref{table:preprocessing}. In wLAS representation, we also maintain the collection of activities in $AX_i$, for each $i$, sorted in lexicographic order.

\begin{table}[H]
	\begin{center}
		\resizebox{\textwidth}{!}{\begin{tabular}{ll}
			\hline
			\\[-1em]
			weighted locationSet (wLS) & $AL_1\equiv \{p_1:\frac{2}{9}, ~p_2:\frac{1}{9}, ~p_5:\frac{4}{9}, ~p_6:\frac{2}{9}\}$ \\
			\\[-1em]
			\hline \hline
			\\[-1em]
			weighted location-activity-set  & $\alpha_1=(AL_1,AX_1)$, $\;$ where\\ 
			\\[-1em]
			(wLAS)& $AL_1$ as above, $\;\;$  $AX_1=\{a,b,h\}$\\
			\\[-1em]
			\hline \hline
			\\[-1em]
			weighted location-activity-set  & $\alpha =\langle \alpha_1,\alpha_2,\alpha_3,\alpha_4  \rangle$, $\;$ where \\
			\\[-1em]
			Sequence (wLAS-sequence) &
			$\alpha_1=(AL_1,AX_1)$,$\alpha_2=(AL_2,AX_2)$,\\
			\\[-1em]
			& $\alpha_3=(AL_3,AX_3)$,$\alpha_4=(AL_4,AX_4)$\\
			\\[-1em]  
			&
			$AL_1=\{p_1:\frac{2}{9},$
			~{\boldmath${ p_2}$} $:\frac{1}{9}, ~p_5:\frac{4}{9},$
			~{\boldmath$p_6$}$:\frac{2}{9}\}$ and $AX_1=\{\mathrm{a,b,{\bf h} }\}$\\
			\\[-1em]
			& $AL_2=\{p_{11}:\frac{1}{3}, 		 ~p_{12}:\frac{1}{6}, ~p_{15}:\frac{1}{3}, 		 ~p_{16}:\frac{1}{6}\}$ and $AX_2=\{\mathrm{a,c,f,g}\}$\\
			\\[-1em]
			& $AL_3=\{p_{22}:\frac{1}{6},
			~p_{23}:\frac{1}{6}, ~p_{26}:\frac{1}{3},
			~p_{27}:\frac{1}{3}\}$ and $AX_3=\{\mathrm{b,c,e}\}$\\
			\\[-1em]
			& $AL_4=\{\mbox{\boldmath{$p_{25}$}}:\frac{1}{2},~\mbox{\boldmath{$p_{26}$}}:\frac{1}{2}\}$ and $AX_4=\{\mathrm{{\bf a,c},h}\}$\\
			\\[-1em]
			\hline
		\end{tabular}}
		\caption{wLAS representation of the Anonymous Trajectory in Fig~\ref{fig:prepro}}
		\label{table:preprocessing}
	\end{center}
\end{table}
By $w(p,\alpha_i)$,  we denote the weight of a cell $p$ in wLAS $\alpha_i$. For example, as in Table~\ref{table:preprocessing}, we have $w(p_2,\alpha_1)=\frac{1}{9}$. In subsequent sections, we define the relevance of the trajectory-patterns using weights of cells as in wLAS representation.
We do not consider any separate weight with activities and assume that a (cell, activity) pair within an anonymized trajectory has the same weight as that of the weight of the cell in the matching trajectory. This is equivalent of assuming that activities within the MBRs are equally likely, and therefore, associated with a cell with the same chance as that of a user being present in the cell. The use of additional knowledge about activities like activity ranking based on user feedback, most-frequently visited activities, etc., are expected to refine the quality of the patterns. However, we leave such extensions on the modeling of activity-weights and its effect on trajectory-pattern mining over anonymized data as a work to be explored in the future. This being a basic model, the emphasis is on exploring the possibility of mining patterns from the anonymized data than enhancing the pattern quality.

For a fixed cell partitioning of the region ${\cal R}$,  we convert each anonymous trajectories in ${\cal PD}$ into  a weighted location-activity-set sequence (aka. wLAS-sequence). For notational convenience, in rest of the paper, we still denote the wLS by $AL$, wLAS by $\alpha_i$, wLAS-sequence by $\alpha$ and processed private data by $\cal PD$.
For a particular wLAS $\alpha_i$, the set of locations in $\alpha_i$ are denoted by $\alpha_i.loc$, and the set of activities are denoted by $\alpha_i.act$. For example, consider $\alpha_1$ as in the table above, we have $\alpha_1.loc = \{p_1,p_2,p_5,p_6\}$ and
$\alpha_1.act = \{\mathrm{a,b,h}\}$.

\begin{definition}
	Consider two wLAS-sequences $\alpha^a, \alpha^b $ of length $p$ and $q$ respectively, say $\alpha^a  = \langle \alpha^a_1, \alpha^a_2, \ldots,\alpha^a_p \rangle$ and $ \alpha^b = \langle \alpha^b_1, \alpha^b_2, \ldots,\alpha^b_q \rangle$, where $i^{th}$ wLAS in $\alpha^a$ is $\alpha^a_i = (AL^a_i, AX^a_i)$ and  $j^{th}$ wLAS in $\alpha^b$ is $\alpha^b_j=(AL^b_j,AX^b_j)$. 
	\begin{enumerate}
		\item (wLS containment) wLS $AL^a_i$ is contained in $AL^b_j$, denoted by
		$AL^a_i \subseteq AL^b_j$, if locations in $AL^a_i$ are all locations in $AL^b_j$ i.e., $AL^a_i.loc \subseteq AL^b_j.loc$. 
		\item (wLAS containment) wLAS $\alpha^a_i$ is said to be contained in $\alpha^b_j$, denoted by $\alpha^a_i \subseteq \alpha^b_j$, if $AL^a_i \subseteq AL^b_j$ and $AX^a_i \subseteq AX^b_j$. It should be noted that the latter containment is a standard set-containment
		\item (wLAS-sequence containment) $\alpha^a$ is contained in $\alpha^b$, denoted by $\alpha^a \subseteq \alpha^b$,  if every wLAS of $\alpha^a$ is contained in some wLAS of $\alpha^b$. Formally,
		
		$\alpha^a \subseteq \alpha^b$ if $p \le q$ and there exist $1 \le  j_1 < j_2 \ldots < j_p \le q$ such that for every $k=1$ to $p$ we have $\alpha^a_k \subseteq \alpha^b_{j_k}$. 
		\item (wLAS subsequence) $\alpha^a$ is a  subsequence of $\alpha^b$ if $\alpha^a \subseteq \alpha^b$.
	\end{enumerate}
\end{definition}
\begin{example} Consider wLAS-sequences $\alpha$ in Table~\ref{table:preprocessing} and $\beta$ as below. The respective matching of $\beta$ in $\alpha$ are highlighted.
	\begin{table}[H]
		\begin{center}
			\begin{tabular}{ll}
				\hline
				\\[-1em]
				wLAS-sequence $\beta=\langle \beta_1,\beta_2 \rangle$&
				$\beta_1=(AL^b_1,AX^b_1)$,$~\beta_2=(AL^b_2,AX^b_2)$\\
				\\[-1em]  
				&
				$AL^b_1=\{~p_2:\frac{1}{3},			~p_6:\frac{2}{3}\}$ and $AX^b_1=\{\mathrm{h}\}$\\
				\\[-1em]
				& $AL^b_2=\{p_{25}:\frac{1}{2}, 		~p_{26}:\frac{1}{2}\}$ and $AX_2=\{\mathrm{a,c}\}$\\
				\\[-1em]
				\hline \hline
				wLS containment & $AL^b_1 \subseteq AL_1$ and $AL^b_2 \subseteq AL_4$ \\
				\\[-1em]
				\hline\hline
				wLAS containment & $\beta_1 \subseteq \alpha_1$ and $\beta_2 \subseteq \alpha_4$ \\
				\\[-1em]
				\hline\hline
				wLAS-sequence containment & $\beta \subseteq \alpha$\\
				\\[-1em]
				\hline\hline
			\end{tabular}
			\caption{wLAS subsequence}
			\label{table:subsequence}
		\end{center}
	\end{table}
\end{example}

\subsection{Trajectory-pattern Matching from a wLAS-sequence}
\label{sec:TPM}
We now introduce some of the basic notions and definitions related to trajectory-pattern and  $wLAS$-sequence, which is useful for explanations in the subsequent sections.


\begin{definition}
\label{defn:match}
	For a trajectory-pattern $T=\langle (P_i,X_i) \rangle_{i=1}^n$ of length $n$ and a wLAS-sequence $\alpha = \langle \alpha_1,\alpha_2,\ldots \alpha_m \rangle$,
	\begin{enumerate}
		\item  $\alpha$  is said to be an {\bf exact-match} of $T$, denoted by $T \sim \alpha $, if $m=n$ and for all $i$, $P_i=\alpha_i.loc$ and $X_i=\alpha_i.act$.
		\item $\alpha $ is said to be a {\bf matching sequence} of $T$, denoted by $T \precsim \alpha$, if there exists a subsequence of $\alpha$ which is an exact-match of $T$. Formally, 	$ T \precsim \alpha ~$ if $\exists~ \alpha^\prime \subseteq \alpha$ such that $ ~T \sim \alpha'$. 
	\end{enumerate}	
\end{definition}



For a discussion on pattern mining from a wLAS-sequence database, we need to consider all subsequences of a wLAS-sequence having an exact-match with the pattern. We also define notions such as {\it pivot-match} and {\it projected subsequence} of a pattern in wLAS-sequences, as discussed below.

\begin{definition}
\label{defn:pivot-match}
	Consider a trajectory-pattern $T$, its matching sequence $\alpha = \langle \alpha_1,\alpha_2,\ldots ,\alpha_m \rangle$ ($|T| \le m$) and a subsequence $\beta = \langle \beta_{k_1}, \beta_{k_2},\ldots,\beta_{k_{|T|}}\rangle$ of $\alpha$ of the length same as that of $T$. 
	\begin{enumerate}	
		\item $\beta$ is said to be a {\bf pivot-match} of $T$ in $\alpha$ if
		\begin{itemize}
			\item $\beta$ is an exact-match of $T$, i.e., $T \sim \beta $, and
			\item It is the first exact-match of $T$ in $\alpha$, i.e., for any other exact-match of $T$ in $\alpha$, say $\gamma = \langle \gamma_{s_1}, \gamma_{s_2},\ldots,\gamma_{s_{|T|}}\rangle$, the last matching term of $T$ in $\beta$ (having index $k_{|T|}$) comes before the last matching term in $\gamma$ (having index $s_{|T|}$). 
		\end{itemize}
		\item For a pivot-match $\beta = \langle \beta_{k_1}, \beta_{k_2}, \ldots, \beta_{k_{|T|}}\rangle$ of $T$ in $\alpha$, the term of $\alpha$ containing the last term of the pivot-match, i.e., $\beta_{k_{|T|}}$, is called the {\bf pivot-term} of $T$ in $\alpha$.
		\item If pivot term of  $T$ in  $\alpha $ is $\alpha_{a_k}$, the {\bf projected subsequence} of $T$ in $\alpha$, denoted by $\overline{\alpha}_T$  is the subsequence of $\alpha$ which contains the remaining terms of $\alpha$ after the pivot-term, including the remaining weighted-locations and activities from the pivot term. Formally, $\overline{\alpha}_T =<\overline{\alpha}_{a_k}, \alpha_{a_k+1},\ldots ,\alpha_n>$, where $\overline{\alpha}_{a_k}$ denotes remaining wLAS in $\alpha_{a_k}$.
	\end{enumerate}	
\end{definition}
\begin{example} Consider the following example of wLAS-sequence $\gamma$ and trajectory-patters $T^a,T^b,T^c,T^d$ and $T^e$. It can be seen that $\gamma$ is a matching sequence for $T^a,T^b, T^c$ and $T^d$ but not for the pattern $T^e$. A matching sequence may have more than one subsequences which contains the exact match to the given trajectory-pattern. For example, $T^b$ has a match in subsequence $\langle \gamma_1, \gamma_2 \rangle$ and $\langle \gamma_1 , \gamma_3 \rangle$. Similarly, $T^c$ has match in $\langle \gamma_1, \gamma_2 \rangle$, $\langle \gamma_2, \gamma_3 \rangle$, $\langle \gamma_1 , \gamma_3 \rangle$ and $T^d$	has a match in $\langle \gamma_1, \gamma_3 \rangle$, $\langle \gamma_2 , \gamma_3 \rangle$.
	\begin{table}[H]
		\begin{center}
			\resizebox{\textwidth}{!}{\begin{tabular}{l|lll}
				\hline\hline
				\\[-1em]
				\boldmath{$\gamma=\langle \gamma_1, \gamma_2, \gamma_3\rangle$} & $\gamma_1 =  (AL_1,AX_1)$ & $\gamma_2 =(AL_2,AX_2)$& $\gamma_3=(AL_3,AX_3)$\\ 
				(wLAS-sequence) &  $AL_1=\{p_2:0.6,p_3:0.4\}$ & $AL_2=(\{p_{3}:0.2, p_{4}:0.8\}$& $AL_3=\{ p_4:0.3, p_6:0.7\}$\\ &$AX_1=\{\mathrm{b,e,h}\})$ & $ AX_2=\{\mathrm{a,b,f,g}\})$&$AX_3=\{c,f,h\}$\\
				\\[-1em]
				\hline \hline
				\\[-1em]
				Trajectory-patterns &Patterns Match in $\gamma$&&\\ 
				\hline
				\\[-1em]
				\boldmath{$T^a=\langle t^a_1, t^a_2 \rangle \precsim \gamma$} &		$t^a_1=(p_2,\mathrm{b})~$&  \boldmath{$t^a_2=(p_{3},\mathrm{b})$}& $\;\;\;\;\;\;\;\;\;\;\;\;\;\;\;\;\;\;\;\;\;\;\leftarrow$ pivot match\\
				\\[-1em]
				\hline
				\\[-1em]
				\boldmath{$T^b=\langle t^b_1, t^b_2 \rangle \precsim \gamma$}&	$t^b_1 =(p_2,\{\mathrm{b,e}\})$&\boldmath{$t^b_2=(p_{4},\mathrm{f})$}&$\;\;\;\;\;\;\;\;\;\;\;\;\;\;\;\;\;\;\;\;\;\;\leftarrow$ pivot match\\
				& $t^b_1=(p_2,\mathrm{\{b,e\}})~$&&  $t^b_2=(p_{4},\mathrm{f})$\\
				\\[-1em]
				\hline
				\\[-1em]
				\boldmath{$T^c=\langle t^c_1, t^c_2 \rangle \precsim \gamma$}&	$t^c_1 =(p_3,\{\mathrm{b}\})$&\boldmath{$t^c_2=(p_{4},\mathrm{f})$}&$\;\;\;\;\;\;\;\;\;\;\;\;\;\;\;\;\;\;\;\;\;\;\leftarrow$ pivot match\\
				& &$t^c_1=(p_3,\mathrm{b})~$&  $t^c_2=(p_{4},\mathrm{f})$\\
				
				& $t^c_1=(p_3,\mathrm{b})~$&&  $t^c_2=(p_{4},\mathrm{f})$\\
				\\[-1em]
				\hline
				\\[-1em]
				\boldmath{$T^d=\langle t^d_1, t^d_2 \rangle \precsim \gamma$}&	$t^d_1 =(p_3,\{\mathrm{b}\})$&&{\boldmath{$t^d_2=(p_{6},\mathrm{c})$}}$\;\;\;\leftarrow$ pivot match\\
				& &$t^d_1=(p_3,\mathrm{b})~$&  {\boldmath $t^d_2=(p_{6},\mathrm{c})$}$\;\;\;\leftarrow$ pivot match\\
				\\[-1em]
				\hline
				\\[-1em]
				\boldmath{$T^e=\langle t^e_1\rangle \not \precsim \gamma$} & No Match & No Match& No Match\\
				$t^e_1=(p_{4},\mathrm{e})$&&&\\
				\\[-1em]
				\hline\hline
				\\[-1em]
				Projected Subsequence of &&&\\
				\\[-1em]
				\hline
				\\[-1em]
				$T^a$ in $\gamma$ &&$\overline{\gamma_2}=(p_4,\{\mathrm{f,g}\})$& $\gamma_3$\\
				\\[-1em]
				\hline
				\\[-1em]
				$T^b$ in $\gamma$ &&$\overline{\gamma_2}=$null&$\gamma_3$\\
				\hline
				\\[-1em]
				$T^c$ in $\gamma$ &&$\overline{\gamma_2}=$null&$\gamma_3$\\
				\hline
				\\[-1em]
				$T^d$ in $\gamma$ &&&$\overline{\gamma_3}=$null\\
				\\[-1em]
				\hline\hline
			\end{tabular}}
			\caption{Matching Sequences, Pivot Match and Projected Subsequence}
			\label{table:matching}
		\end{center}
	\end{table}
\end{example}

From the above example, the pivot-match of trajectory-patterns $T^a$, $T^b$ and $T^c$ are in the same subsequence $\langle \gamma_1, \gamma_2\rangle$ and the pivot term is $\gamma_2$. 
In all these cases, the pivot-match is unique. However, there can be more than one pivot-match of a trajectory-pattern. This is because pivot-match only restricts the pivot-term which is unique. The sub-pattern of the trajectory-pattern excluding the last term may have multiple exact-match. For example, the trajectory-pattern $T^d$ has two  matches in $\gamma$, namely $\langle\gamma_1,\gamma_3\rangle$ and $\langle\gamma_2,\gamma_3\rangle$, and both are pivot-match. The projected subsequence of $T^a$ in $\gamma$ is $\langle\overline{\gamma_2}, \gamma_3\rangle$, which is rest of the wLAS-sequence in $\gamma$ after the patter $T^a$. Here $\overline{\gamma_2}$ contains location and activities of the pivot term $\gamma_2$, excluding the locations and activities of the last term of $T^a$. For the projected sequence of $T^b$ and $T^c$, the term $\overline{\gamma_2}$ is null. This is because either the location or the activity set becomes empty.

For a trajectory-pattern, we now introduce two more terms, namely {\it $1$-pattern} and {\it $r$-pattern}. For a wLAS-sequence, a $1$-pattern is a $(location, activity)$ pair for which wLAS-sequence is a matching. For example, $(p_2,b)$ is a $1$-pattern of $\gamma$ as in Table~\ref{table:matching}, whereas $(p_5,c)$ is not. The complete trajectory-pattern of a given wLAS-sequence is called a raw pattern or  $r$-pattern. For example, 
$\langle (\{p_2,p_3\},\{\mathrm{b,e,h}\}) (\{p_{3}, p_{4}\},\{\mathrm{a,b,f,g}\})$ \\ $(\{p_4,p_6\},\{c,f,h\})\rangle$ is an $r$-pattern for the $\gamma$. We need $1$-sequence and $r$-sequence for discussion in threshold initialization.

\begin{definition}[Sequence-relevance]
\label{defn:sr}
	For a given wLAS-sequence $\alpha$, the relevance of its $r$-pattern is called the sequence-relevance. We denote it by ${\cal R}(\alpha)$	
\end{definition} 
\subsection{Relevance of a Trajectory-Pattern}	 
\label{sec:relevance}
 
%

Our objective is to find  high-relevant trajectory-patterns as discussed in Definition~\ref{defn:high-relevance}. For this, we assign a weight to a trajectory-pattern,  
\begin{itemize}
	\item by using the weight associated with cells in a wLAS-sequence representation of anonymous trajectory, and
	\item by accumulating all the weights due to an overlap of a pattern with different anonymous trajectories in ${\cal PD}$.
\end{itemize}
The weight of a trajectory-pattern in the wLAS-sequence estimate the chance of a pattern in the database. We call this weight of a trajectory-pattern as its {\it relevance-score} in the database or just as its {\it relevance}. Next, we formally define the procedure to estimate the relevance of a trajectory-pattern. We first define the relevance of a single (location,activity) pair in a wLAS-sequence (Definition \ref{defn:relevance}, point-1). Using this, we define the relevance of (locationSet,activitySet), i.e., a trajectory-term in a wLAS (Definition \ref{defn:relevance}, point-2). Finally, we discuss the relevance of a trajectory-pattern in a wLAS-sequence (Definition \ref{defn:relevance}, point-3) and in a wLAS-sequence database (Definition \ref{defn:relevance}, point-4). For our formalization, we assume that user is equally likely to be at any place of locations in a wLAS; and may perform any set of activities out of activitySet in wLAS. 

\begin{definition}[Relevance of Trajectory-pattern]
\label{defn:relevance}
Consider a wLAS-sequence database ${\cal PD}$, where $\alpha=\langle\alpha_i\rangle_i$ denote a wLAS-sequence in ${\cal PD}$ and $\alpha_i$ as $i^{th}$ wLAS-term in $\alpha$.
Then  	
	\begin{enumerate}
		\item The relevance of a single location-activity pair $(p,\mathrm{a})$ in  wLAS $\alpha_t$ is equal to weight of $p$ in the $\alpha_t$, if $\alpha_t$ is a matching sequence of $(p,\mathrm{a})$. i.e.,
		\begin{center}
			\begin{tabular}{lll}
				$\displaystyle{\cal R}((p,\mathrm{a}), \alpha_t)$ & = & $\left\{
				\begin{array}{ll}
				w(p, \alpha_t)  & \text{ if } (p,a) \precsim \alpha_t \\
				0 & \text{ otherwise }
				\end{array}
				\right.$
			\end{tabular}
		\end{center}
		\item The relevance of locationSet-activitySet pair $(L,X)$ in some wLAS $\alpha_t$ is
		\begin{enumerate}
			\item When $L = \{p_1,p_2\ldots, p_n\} $, $X= \{\mathrm{a}\}$ and $(L,X) \precsim \alpha_i$. The relevance of $(L,X)$ in $\alpha_i$ is the sum of the overlap ratio of locations in  $\alpha_i.loc$ with that of $L$. 
			
			\begin{center}
				\begin{tabular}{lll}
					$\displaystyle{\cal R}((L,X), \alpha_t)$  =   $\left\{
					\begin{array}{ll}
					\sum_{i=1}^n{\cal R}(~(p_i,\mathrm{a}),~\alpha_t~)  & (L,X)\precsim \alpha_t \\
					0 & \mbox{otherwise }
					\end{array}
					\right.$
				\end{tabular}
			\end{center}
			\item When $L= \{p_1,\ldots,p_n\}$ and $X= \{\mathrm{a_1,a_2,\ldots, a_m}\}$. The relevance of $(L,X)$ in $\alpha_t$ in term of activities depends solely on the fact whether activities in $X$ belong to $\alpha_t.act$. This is because, we do not associate weight with activities.
			\begin{center}
				$\displaystyle{\cal R}((L,X), \alpha_t)  = \min_{j}\{~{\cal R}(L,\mathrm{a_j})~\vert~j=1,\ldots, m~\}$  
			\end{center}
		\end{enumerate}
		
		\item The relevance of a trajectory pattern $T=\langle (L_i,X_i) \rangle$ in  wLAS-sequence.
		\begin{enumerate}
			\item The relevance of a trajectory-pattern $T$ in a wLAS-sequence $\alpha$ is the collection of all relevance values of $T$ for possible multiple exact-match of $T$ in $\alpha$. 
			\begin{center}
				$\displaystyle {\cal R}(T, \alpha) =
				\{ ~\; \sum_i {\cal R}((L_i,X_i), \beta_i) ~~|~~   \beta= \langle
				\beta_i \rangle \subseteq \alpha \ ~\mbox{ and }~  T \sim \beta )\} $.
			\end{center}
			\item The max-relevance of a trajectory-pattern $T$ in a wLAS-sequence $\alpha$, denoted by ${\cal R}_{\max}(T, \alpha)$, is the maximum over all relevance values of $T$ in $\alpha$.
			\begin{center}
				${\cal R}_{max}(T, \alpha) = \max {\cal R}(T,\alpha) $.
			\end{center}
		\end{enumerate} 
		\item The relevance score of a trajectory-pattern $T$ in an anonymous trajectory data ${\cal PD}$ is the summation of all the relevance values of $T$ over matching sequences in ${\cal PD}$. 
		\begin{center}
			$\displaystyle {\cal R}(T, {\cal PD}) = \sum_{\alpha \in {\cal PD}} {\cal R}_{max}(T,\alpha )$
		\end{center}
	\end{enumerate}	
\end{definition}


We define relevance of a trajectory-pattern $T=\langle (L_i,X_i) \rangle$ in a wLAS-sequence $\alpha \in {\cal PD}$ using its match in $\alpha$; and by accumulating the relevance of trajectory-term $(L_i,X_i)$ (called {\it term-relevance})  with the corresponding matching-term. For term-relevance of $(L,X)$ in $\alpha^t$, if any activity in $X$ is not in $\alpha_t.act$, then the relevance is zero; otherwise, the relevance is same as that of $(L,\mathrm{a}_i)$ for any activity
$\mathrm{a}_i$ in $X$. It is easy to verify that ${\cal
	R}(L,\mathrm{a}_i) = {\cal R}(L,\mathrm{a}_j)$, for any
$\mathrm{a}_i,\mathrm{a}_j \in X$.
In the definition of term-relevance, we first consider the relevance of the locations in $L$ and then the relevance of the activities in $X$. We could have done the other way round and obtained the same value.

The relevance of a trajectory-pattern $T$ with the matching wLAS-sequence $\alpha$ is defined as the summation of the term-relevance values. An alternative theoretically more sound definition for the relevance of a trajectory-pattern may be to consider the product of the term-relevance of the independent terms in the exact-match in place of summation. Being an overlap ratio of locations, term-relevance values are fractions, and their product will be even smaller. To avoid working with very small threshold values, we have considered summation of individual term-relevances which has the same pruning power as the product in term of finding useful patterns. Since a trajectory-pattern may have multiple matches in a wLAS-sequence, its relevance in wLAS-sequence is the collection of all relevance values with distinct matches. For defining the relevance of a trajectory-pattern in a wLAS-sequence data, we consider the maximum relevance value of a trajectory-pattern among its distinct matches in wLAS-sequence. The summation of maximum relevances of a trajectory pattern over the matching sequences in a wLAS-sequence data is, therefore, the relevance of a trajectory-pattern in wLAS-sequence data.   If wLAS-sequence data is clear from the context, we sometimes denote the relevance of a trajectory-pattern $T$ by just writing ${\cal R}(T)$. 



\begin{example}
	\label{ex1}
	For explaining relevance of a trajectory-pattern, as discussed above, we consider a toy example of an anonymized database ${\cal PD} = \langle\alpha_1,\alpha_2,\alpha_3\rangle$ in Table \ref{tab:1}. The evaluation of the relevance of a trajectory pattern $T= \langle \eta_1, \eta_2 \rangle$, where $\eta_1=(\{p_1,p_2\},\{\mathrm{a,b}\})$ and $\eta_2= (p_5,\mathrm{g})$, in $\alpha_1$ and in ${\cal PD}$ is shown in Table \ref{tabel:relevance}.
	\begin{table}[H]
		\centering
		\resizebox{\textwidth}{!}{\begin{tabular}{l|l}
			\hline\hline
			Private Data &\\
			\hline
			\\[-1em]
			$\alpha_1 = \langle\alpha^a_1,\alpha^a_2,\alpha^a_3\rangle$ 
			& $\alpha^a_1 =(AL^a_1,AX^a_1),~\alpha^a_2 =(AL^a_2,AX^a_2),~\alpha^q_3 =(AL^a_3,AX^a_3)$\\ & $AL^a_1=\{p_1:0.25,p_2:0.25, p_5:0.25,p_6:0.25\}$ and $ AX^a_1=\{\mathrm{a,b,h}\}$\\
			&$AL^a_2=\{ p_1:0.2, p_2:0.2, p_{5}:0.4, p_{7}:0.2\}$ and $AX^a_2=\{\mathrm{a,b,g,j}\})$\\
			&$AL^a_3=\{ p_{3}:0.2, p_{5}:0.1, p_{7}:0.25,  p_{11}:0.45\}$ and $AX^a_3=\{\mathrm{a,c,d,g}\})$\\
			\\[-1em]
			\hline
			\\[-1em]
			$\alpha_2 =  \langle\alpha^b_1,\alpha^b_2,\alpha^b_3 \rangle$&
			$\alpha^b_1 = (AL^b_1,AX^b_1),\alpha^b_2 = (AL^b_2,AX^b_2),\alpha^b_3 = (AL^b_3,AX^b_3)$\\
			&$AL^b_1=\{p_3:0.26, p_4:0.22, p_7:0.3, p_8:0.22,\}$ and $AX^b_1=\{\mathrm{d,e,h}\}$\\
			&$AL^b_2=\{p_6:0.13, p_7:0.2, p_{10}:0.2, p_{11}:0.47,\}$ and $AX^b_2=\{\mathrm{e,f,g}\}$\\
			&$AL^b_3=\{p_9:0.24, p_{10}:0.34, p_{13}:0.22, p_{14}:0.2\}$ and $AX^b_3=\{\mathrm{d,f}\}$\\
			\\[-1em]
			\hline
			\\[-1em]
			$\alpha_3 = \langle\alpha^c_1,\alpha^c_2\rangle$&
			$\alpha^c_1 = (AL^c_1,AX^c_1),~\alpha^c_2 = (AL^c_2,AX^c_2)$\\
			& $AL^c_1=\{p_1:0.2,p_2:0.4, p_6:0.1, p_7:0.3,\}$ and $AX^c_1=\{\mathrm{a,b,h}\}$\\
			&$AL^c_2=\{p_{5}:0.2, p_{6}:0.3, p_{10}:0.2,p_{11}:0.3,\}$ and $AX^c_2=\{\mathrm{a,g,h}\}$\\
			\\[-1em]
			\hline\hline
    		\end{tabular}}
		\caption{Sample wLAS-sequence database}
		\label{tab:1}
	\end{table}
	
\end{example}	

\begin{table}[H]
	\centering
	\resizebox{\textwidth}{!}{\begin{tabular}{l|lll}
		\hline\hline
		Relevance of   & $\alpha^a_1$ &$\alpha^a_2$ & $\alpha^a_3$\\
		$T= \langle \eta_1, \eta_2 \rangle$ in $\alpha_1$ &&&\\
		\cline{2-4}
        \hline
		\\[-1em]
		&  $\eta_1 \precsim \alpha^a_1$ & $\eta_1 \precsim \alpha^a_2$ & $\eta_1 \not \precsim \alpha^a_3$ \\
		\\[-1em]
		$\eta_1=(\{p_1,p_2\},\{\mathrm{a,b}\})$ & ${\cal R}((p_1,a),\alpha^a_1)=0.25$  & ${\cal R}((p_1,a),\alpha^a_2)= 0.2$ &${\cal R}(\eta_1,\alpha^a_3)=0$\\ 
		& ${\cal R}((p_2,a),\alpha^a_1)=0.25$&$ {\cal R}((p_2,a),\alpha^a_2)=0.2$&\\
		&${\cal R}((\{p_1,p_2\},a),\alpha^a_1)=0.5$&${\cal R}((\{p_1,p_2\},a),\alpha^a_2)=0.4$&\\
		&${\cal R}((\{p_1,p_2\},b),\alpha^a_1)=0.5$&${\cal R}((\{p_1,p_2\},b),\alpha^a_2)=0.4$&\\
		&${\cal R}((\eta_1,\alpha^a_1)=0.5$&${\cal R}(\eta_1,\alpha^a_2)=0.4$&\\
		\\[-1em]
		\\[-1em]
        \hline
		$\eta_2= (p_5,\mathrm{g})$ &$\eta_2 \not \precsim \alpha^a_1$&$\eta_2 \precsim \alpha^a_2$&$\eta_2 \precsim \alpha^a_3$\\
		\\[-1em]
		& ${\cal R}(\eta_2,\alpha^a_1)=0$  & ${\cal R}(\eta_2,\alpha^a_2)= 0.4$ &${\cal R}(\eta_2,\alpha^a_3)=0.1$\\	
		\\[-1em]
        \hline 
		\\[-1em]
		$1^{st}$ Match of T in $\alpha_1$, say $\beta_1$ & $\eta_1\precsim\alpha^a_1$ & $\eta_2\precsim \alpha^a_2$ & \\
		\\[-1em]
		${\cal R}(T, \beta_1) =0.9$ &${\cal R}(\eta_1,\alpha^a_1)$& + ${\cal R}(\eta_2,\alpha^a_2)$&\\ 
        \hline
		$2^{nd}$ Match of T in $\alpha_1$, say $\beta_2$& $\eta_1\precsim \alpha^a_1$ & & $\eta_2\precsim\alpha^a_3$ \\
		\\[-1em]
		${\cal R}(T, \beta_2) =0.6$ & ${\cal R}(\eta_1,\alpha^a_1)$ & & $+ {\cal R}(\eta_2,\alpha^a_3)$ \\ 
        \hline 
		$3^{rd}$ Match of T in $\alpha_1$,say $\beta_3$&& $\eta_1\precsim \alpha^a_2$ & $\eta_2\precsim \alpha^a_3$\\
		\\[-1em]
		${\cal R}(T, \beta_3) =0.5$ & & ${\cal R}(\eta_1,\alpha^a_2)$  & $+ {\cal R}(\eta_2,\alpha^a_3)$ \\\ \\
		&\multicolumn{3}{l}{
			${\cal R}_{max}(T,\alpha_1) = \max \{  {\cal R}(T, \beta_1), {\cal R}(T, \beta_2),{\cal R}(T, \beta_3)\}= \max\{0.9,0.6,0.5\}=0.9$ 
		}\\
		\\[-1em]
		\hline
		\\[-1em]
		&\multicolumn{3}{l}{
			${\cal R}_{max}(T,\alpha_2) =  0$    \hspace{10mm}($T \not \precsim \alpha_2$) 
		}\\
		\\[-1em]	
		&\multicolumn{3}{l}{
			${\cal R}_{max}(T,\alpha_3) = \max \{ \{ {\cal R}(\eta_1, \alpha^c_1)+ {\cal R}(\eta_2, \alpha^c_2)\}\}= \max\{\{0.6+0.2\}\}=0.8$ 
		}\\
		\\[-1em]
		&	\multicolumn{3}{l}{
			${\cal R}(T,{\cal PD}) =   {\cal R}_{max}(T, \alpha_1) +  {\cal R}_{max}(T, \alpha_2)+ {\cal R}_{max}(T, \alpha_3)= 1.7$ 
		}\\
		\hline \hline
	\end{tabular}}
	\caption{Relevance of Trajectory-pattern in wLAS-sequence data}
	\label{tabel:relevance}
\end{table}

%% file: top-k.tex
Patterns from an anonymized data can be mined for a user specified $min\_relevance$ threshold. However, it is often difficult to set a proper threshold. A small threshold may produce thousands of patterns, whereas a high value may lead to no findings. Since the value of threshold is data-specific, setting a proper threshold for desired result size may be computationally challenging. We, therefore, propose a top-k framework for mining patterns from an anonymized data. In top-k mining framework, for a user specified pattern count--an interger value $k$, we intend to mine {\it top-k high-relevant trajectory patterns} from an anonymized data ${\cal PD}$ which is in a wLAS-sequence form. 
\begin{definition}[Top-k high relevant pattern]
	A trajectory-pattern $T$ is said to be a top-k high-relevant pattern if there are less than $k$ trajectory-patterns whose relevance score is more than the relevance of $T$ in the database ${\cal PD}$.  
\end{definition}

Our computational procedure for pattern mining over a wLAS-sequence data is influenced by a pattern growth approach $USpan$~\cite{yin2012uspan} which is designed for mining high-utility sequential patterns from a weighted itemset sequence data. The data in $USpan$ consists of one-dimensional sequences having terms as a set of items with their term-utility (that is quantity $\times$ unit-price). However, in our framework, a wLAS-sequence data contains two-dimensional sequences where each term is a pair having spatial-information (a weighted-locations-set) and textual information (an activity-set). The similarity in one of the dimensions of a wLAS-sequence data (i.e., weighted-location-set) with that of sequences in $USpan$ suggests a possible adoption of $USpan$ as a three-step procedure. First, transform a two-dimensional wLAS-sequence data into an equivalent one-dimensional data (required as an input for $USpan$). Second, apply $USpan$ to find one-dimensional relevant patterns from the transformed data. Third, map the mined patterns to the corresponding two-dimensional patterns. Though this adaptation of $USpan$ as $USpan2D$~\cite{Mintra16} can find high-relevant patterns from wLAS-sequence data, the procedure is computationally expensive. As shown \cite{Mintra16}, $USpan2D$ generates many {\it invalid} one-dimensional patterns that are not equivalent to any two-dimensional patterns in the wLAS-sequence data. Moreover, these invalid patterns cannot be discarded as their extension may generate valid two-dimensional high-relevant patterns. Inability to discard invalid patterns in the transformed space makes the search space substantially large and the $USpan2D$ procedure computationally inviable.

To overcome this difficulty of finding patterns efficiently, a two-dimensional pattern growth approach, namely $Mintra$ \cite{Mintra16}, is proposed.  
 $Mintra$, a non-trivial extension of the one-dimensional pattern growth technique, avoids duplicate pattern generation by restricting repetition in exploration. Though $Mintra$ outperforms $USpan2D$ in term of execution time, finding enough number of patterns in $Mintra$ is not always easy. A threshold value that returns a desired number of patterns is highly data dependent.
This is because the relevance score of patterns may be small numbers for a sparse-data whereas it can be larger numbers for a dense-data. This variation in the relevance scores of patterns makes it hard to choose a proper threshold value for the desired result size over different data sets. In this work, we resolve this issue  by proposing a top-k framework that can efficiently mine top-k high relevant trajectory-patterns for a user-provided parameter $k$. Since there exists no prior work on top-k pattern mining from an anonymized trajectory data, for comparison of our proposed technique, we first discuss a baseline approach that is designed using $Mintra$. An advantage of using $Mintra$ in defining the baseline algorithm is that $Mintra$ is efficient than the one-dimensional pattern growth algorithm $USpan2D$. 

\subsection{$MintraBL$ : A Baseline Algorithm}
\label{sec:mintraBL}
In our baseline algorithm $MintraBL$, we maintain a list `$TopKList$' that contains the best top-k patterns during the exploration using $Mintra$. The patterns in the $TopKList$ are sorted over their relevance score. At any stage of $MintraBL$, the minimum relevance score of the patterns in the  $TopKList$  is set as the current $threshold$. Thus, the $threshold$ value in the $MintraBL$ keeps changing while exploring the search space. The steps of our baseline algorithm are as follows:
\begin{enumerate}
	\item  To start the $TopKList$ is initialized as empty and the threshold value is initialized as $0$.
	\item The wLAS-sequence data is scanned using $Mintra$ to find patterns with their relevance score. 
	\item A pattern having relevance score more than the current $threshold$ is inserted into the $TopKList$ (at an appropriate position); the length of the list is maintained to $k$ if it has become more than $k$ (by dropping a pattern with the least relevance in the list); and, the $threshold$ is modified (to the updated minimum relevance of the current list) for the future exploration.
	\item $MintraBL$ terminates by exploring all the patterns generated using $Mintra$ and returns with the top-k patterns in the $TopKList$.
\end{enumerate}

\subsubsection{Disadvantage of Baseline Algorithm}
Though $MintraBL$ mines top-k patterns correctly, it traverses too many candidate trajectory-patterns. The limitation of the baseline algorithm is that it fails to capitalize the advantage of top-k requirement. There are two main disadvantages to the execution of the baseline approach.
\begin{itemize}
	\item First, in the baseline the $threshold$ value is dynamic, i.e., it changes during the execution. Initially, the $threshold$ is set at $0$, and it gets modified with the addition of patterns in the $TopKList$. This results in the exploration of a high number of patterns initially in the search tree and until the $threshold$ value is not raised substantially.
	\item Second, the value of the $threshold$ is modified based upon the relevance of the  candidates inserted in the $TopKList$. Therefore, the order in which the search tree is explored and the candidates are inserted into the $TopKList$ affects the size of the search tree. In between various possibilities for growing the current pattern,  choosing those patterns early that may have high relevance are expected to increase the threshold sharply and to have relatively smaller search space. However, in $MintraBL$ patterns are picked based upon its lexicographic order than that of its relevance score.  
\end{itemize}    
Based on these observations, we propose two efficiency heuristics, namely {\it threshold initialization (TI)} that resolve the issues with `low initial threshold value', and {\it threshold updation (TU)} that `update threshold optimally at each stage'. Our proposed top-k pattern mining approach over anonymized data uses these two heuristics together with two early termination strategies, namely {\it depth pruning} and  {\it width pruning} ~\cite{yin2012uspan}. We discuss the details of the $TopKMintra$ in the next section.

\subsection{TopKMintra}
\label{sec:topKMintra}

\noindent TopKMintra grows trajectory-patterns in both the spatial and the textual dimensions simultaneously. For this, there are three basic term-concatenation operations, namely location\hyp{}concatenation, activity\hyp{}concatenation, and sequential\hyp{}concatenation, denoted by   l\hyp{}\textit{concatenation}, a\hyp{}{\it concatenation}, and s\hyp{}{\it concatenation}  respectively. In s\hyp{}concatenation, a trajectory-pattern is extended by appending a term $(location, activity)$ at the suffix of the current trajectory-pattern. Whereas in a l\hyp{}concatenation or in an a\hyp{}concatenation, a trajectory-pattern is extended by appending a location or an activity to the last term of the current trajectory-pattern, respectively. An application of these term-concatenation operations on a trajectory-pattern and the resultant trajectory-pattern is shown in Table~\ref{table:concatenation}. 

\begin{table}[h!]
	\begin{center}
		\begin{tabular}{lll}
			\hline\hline
			\\[-1em]
			Trajectory-Pattern & Operation & Resultant Trajectory Pattern\\
			\\[-1em]
			\hline
			\\[-1em]
			$(p_1,\mathrm{a})(p_2,\{\mathrm{b,c}\})$& l\hyp{}concatenation by $p_3$ & $(p_1,\mathrm{a})(\{p_2,p_3\},\{\mathrm{b,c}\})$\\
			& a\hyp{}concatenation by $\mathrm{d}$&$(p_1,\mathrm{a})(p_2 ,\{\mathrm{b,c,d}\})$ \\
			& s\hyp{}concatenation by $(p_4,\mathrm{e})$&$(p_1,\mathrm{a})(p_2,\{\mathrm{b,c}\})(p_4,\mathrm{e})$\\
			\\[-1em]
			\hline\hline
		\end{tabular}
		\caption{Term concatenation operations over Trajectory-patterns}
		\label{table:concatenation}
	\end{center}
\end{table}  
 From a current-pattern, for a given location and an activity term, an l-concatenation followed by an a-concatenation or an a-concatenation followed by an l-concatenation will both generate the same child pattern. To avoid this repetition of patterns, we imposed a restriction on location extension after an activity extension. For example, in the Fig.~\ref{fig:tree-trav} depicting pattern-growth, the pattern $(p_1,\{\mathrm{a,b}\})$ is not extended in location dimension as it is generated through activity extension. 
\begin{figure}[h!]
	\centering\includegraphics[width=0.5\textwidth]{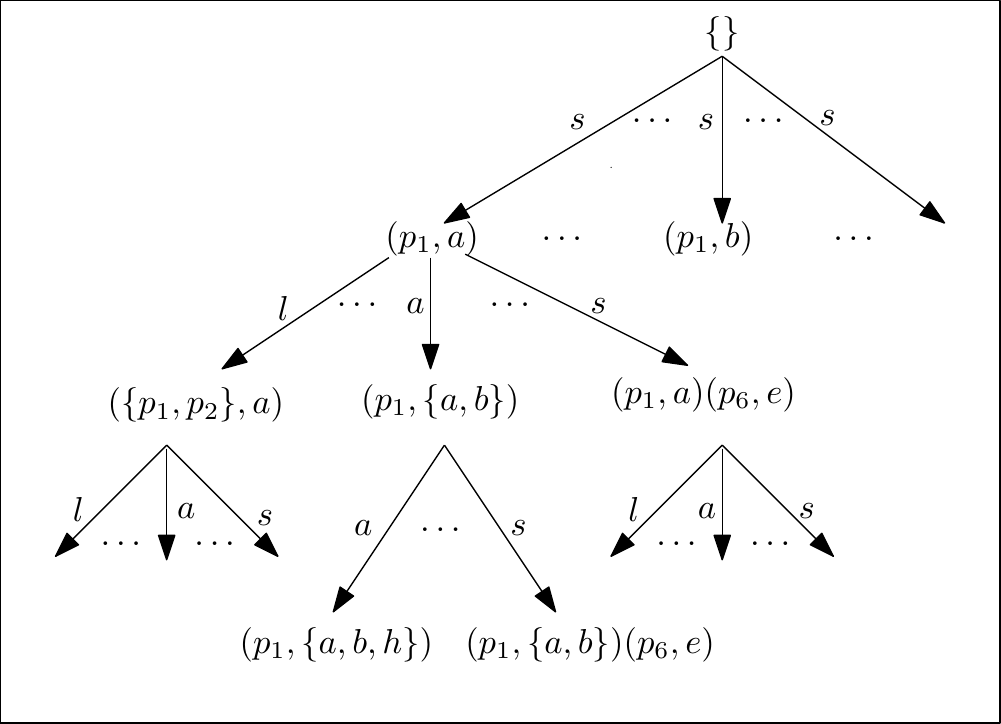}
	\caption{\label{fig:tree-trav} Snapshot of a search tree }
\end{figure}
This restriction on pattern generation not only improves on execution time by avoiding repetition but it is also necessary for finding correct top-k patterns by restricting duplicates. However, the actual efficiency of TopKMintra over the baseline is due to the two heuristics, namely {\it threshold initialization (TI)} and {\it threshold updation (TU)}, that improves on the above mentioned two limitations of the baseline approach--initializing and updating the threshold value; and due to the couple of pruning strategies, namely {\it depth pruning} and {\it width pruning}, that ensures early termination.  

To overcome the limitation of the threshold initialized with $0$ in $MintraBL$, that results in generating a high number of unpromising patterns, we use the following heuristic.
\paragraph*{\bf Strategy 1: Threshold initialization (TI).}  {\it   We initialize the $TopKList$ with patterns of length $1$, length $2$ and complete wLAS-sequences from the database}, sorted based on their relevance-score. The initialization takes place before the actual mining starts and requires a single data scan that can be performed together with other initialization tasks, such as relevance matrix formation (discussed later while explaining the algorithmic procedure). 
The minimum relevance score of these initial patterns in the $TopKList$ is assigned as the initial $threshold$.  As shown through experiments in Section \ref{sec:exp}, the TI heuristic raises the initial threshold value substantially, restrict the unpromising candidates early, and therefore, reduces the search space and the running time.  \qed \\

TI strategy initializes the threshold value to some non-zero value based on the pre-insertion of patterns as mentioned above. Let us consider a pattern which is discarded to be a candidate in $TopKList$ due to this change (say $TopKList_1$); however, assume that it is selected as a candidate in the $TopKList$ when the initial threshold was set to be zero (say $TopKList_2$). Clearly, such a pattern cannot be a top-k pattern. For if it is a valid top-k pattern then its score should be more than the lowest score candidate in the current $TopKlist_1$. This is contradictory to the fact that it was discarded to be a candidate in $TopKList_1$. Such a pattern even if it enters in the $TopKList_2$ has to leave it at some stage in the future. Thus, in both the cases, it can not be a valid candidate eventually. 

Further, we are working on a fixed search space where we are pruning or reordering it based on some fixed strategies to make exploration efficient. If we assume that there is no other strategy except threshold initialization, then the patterns will be picked from the search space in the same order. In that case, it is easy to see that the threshold due to preinsertion will always remain greater or equal to the threshold when we start with zero initial threshold. As a result, any pattern which is a candidate in $TopKList_1$ will also be a candidate in $TopKList_2$. Thus, if other strategies are correct TI neither generates false positive nor it generates false negative.\\
   
\noindent In a fixed-threshold based pattern growth algorithms, the search space is independent of the exploration order. However, for the correctness of such an algorithm, non-repetition of patterns need to be ensured. In $Mintra$, the same is done by introducing a lexicographic order over the term-concatenation, and therefore, by maintaining an order over the search tree. For the top-k approach,  the threshold value changes during the execution. Therefore, an early addition of candidates in the $TopKList$ that raises the threshold sharply may result in generating less unpromising patterns. This leads us to our second heuristic-- {\it threshold update}. We first introduce PTR-order for trajectory patterns and for the term-concatenation operation that is instrumental in defining our next heuristic. We start with  {\it Projectod Trajectory-pattern Relevance (PTR)} which in turn is defined using {\it pivot-match relevance} and {\it projected-subsequence relevance}.

\begin{definition}[Pivot-match relevance]
	\label{defn:pm-relevance}
	For a pivot-match (Definition~\ref{defn:pivot-match}) of a trajectory pattern $T$ in a wLAS-sequence $\alpha$, its maximum relevance is called the pivot-match relevance, and is denoted by ${\cal R}_{PM} (T,\alpha)$.
\end{definition} 
\noindent For trajectory-pattern  $T=\langle \eta_1 ,\eta_2 \rangle$ and its matching sequence $\alpha_1$ in Table~\ref{tabel:relevance}, there are three exact match $\beta_1$, $\beta_2$ and $\beta_3$. Clearly, $\beta_1=\langle \alpha^a_1, \alpha^a_2 \rangle$ is a pivot-match, where ${\cal R}_{PM} (T,\alpha_1) = \max \{ {\cal R}(T,\beta_1)\}=\max\{0.9\}=0.9$. The pivot-match of a trajectory pattern in a matching sequence need not to be unique. For example consider the matching sequence $\alpha_2$ of a trajectory-pattern $T^\prime=\langle \nu_1 \nu_2 \rangle$, where $\nu_1=(p_7,\mathrm{e})$, $\nu_2= (p_{9},\mathrm{d})$. The pattern $T^\prime$ have exact match in subsequences of $\alpha_2$, i.e., in  $\alpha^2_1\cdot\alpha^2_3$, say $\beta_1$, and in $\alpha^2_2\cdot\alpha^2_3$, say $\beta_2$. Clearly, both of them are pivot-match, where ${\cal R}_{PM} (T^\prime,\alpha_2) = \max\{{\cal R}(T^\prime,\beta_1),{\cal  R}(T^\prime,\beta_2)\} =\max\{\{0.3+0.24\},\{0.2+0.24\}\}=0.54$

\begin{definition}[Projected-subsequence relevance]
	For a trajectory pattern $T$ and a wLAS-sequence $\alpha$, the  projected-subsequence relevance of $T$ in $\alpha$ is the sequence-relevance of its projected-subsequence (Definition~\ref{defn:pivot-match}). We denote it by ${\cal R}_{rest} (T,\alpha)$ or by ${\cal R}(\overline{\alpha})$. 
\end{definition}
The {\it projected subsequence} of trajectory-pattern $T$ (Table~\ref{tabel:relevance}) in matching sequence $\alpha_1$ (Table~\ref{tab:1}) is $\langle (\{p_{7}:0.2,\{\mathrm{j}\}),(\{p_3:0.2,p_{5}:0.1, p_{7}:0.25,p_{11}:0.45\},\mathrm{\{a,c,d,g\}})\rangle$, and the projected-subsequence relevance, i.e.,  ${\cal R}_{rest} (T,\alpha_1)$  is $1.2$. \\

\noindent We now define the {\it projected trajectory-pattern relevance (PTR)} of a trajectory pattern in a matching wLAS-sequence and in the wLAS-sequence database.

\begin{definition}[Projected Trajectory-pattern Relevance (PTR)] 
	\label{defn:PTR}
	For a trajectory pattern $T$, a wLAS-sequence data ${\cal PD}$ and a matching wLAS-sequence $\alpha$ of $T$ in ${\cal PD}$,
	\begin{itemize}
		\item The projected trajectory-pattern relevance of $T$ in $\alpha$, denoted by $PTR(T,\alpha)$, is the sum of the pivot-match relevance and the projected subsequence relevance, i.e.,
		$$ PTR(T, \alpha) = {\cal R}_{PM}(T,\alpha) + {\cal R}_{rest}(T,\alpha)$$
		\item The projected trajectory-pattern relevance of $T$ in ${\cal PD}$, denoted by $PTR(T)$, is the sum of projected trajectory-pattern relevance of $T$ over all the matching sequences $\alpha$ in ${\cal PD}$. i.e,
		$$ PTR(T) = \sum_{\alpha\in{\cal PD} }PTR(T,\alpha)$$
	\end{itemize}	  	
\end{definition}
 For a trajectory-pattern $T$ (Table~\ref{tabel:relevance}) and wLAS-sequence database (Table~\ref{tab:1}), the projected trajectory-pattern relevance of $T$ in $\alpha_1$ is $PTR(T,\alpha_1)=2.1$ (as ${\cal R}_{PM} (T,\alpha_1) =0.9$ and ${\cal R}_{rest} (T,\alpha_1) =1.2$). Also, $PTR(T,\alpha_2) =0$ ($T \not \precsim \alpha_2$) and $PTR(T,\alpha_3)=1.6$ (as ${\cal R}_{PM} (T,\alpha_3) =0.8$, ${\cal R}_{rest} (T,\alpha_3) =0.8$ ), and  therefore, $PTR(T)=3.7$.\\ 

\noindent We now define the PTR-order for trajectory patterns and for the concatenation operation.

\begin{definition}(PTR-order for trajectory-patterns)
	Let $T$ be a trajectory-pattern which is extended to $T_x$ and $T_y$ by a single term-concatenation operation by term $x$ and $y$ respectively. Then we say pattern $T_x$ is prior to $T_y$, denoted by $T_x < T_y$, if 
\begin{itemize}
	\item Both $T_x$ and $T_y$ are concatenated through the same operation, and
	\begin{itemize}
		\item either $PTR(T_x) > PTR (T_y)$.
		\item or if $PTR(T_x) = PTR (T_y)$ then $x$ is lexicographically smaller than $y$.
	\end{itemize} 
	\item $T_x$ is l-concatenated and $T_y$ is either a-concatenated or s-concatenated.
	\item $T_x$ is a-concatenated and $T_y$ is s-concatenated.
\end{itemize} 
\end{definition}
PTR-order inserts those candidates early in the $TopKList$ for which PTR (Definition~\ref{defn:PTR}) is high as compare to other candidates at that stage. In the PTR-order, the execution preference is given to l-concatenation, then to a-concatenation and then to s-concatenation. We use PTR-order for  trajectory-patterns to induce an order over term-concatenation.
\begin{definition}(PTR-order based term-concatenation)
For a trajectory-pattern $T$, a term $x$ is prior to term $y$ for term-concatenation, denoted by $x \prec y$, if $T_x < T_y$.
\end{definition}		
\paragraph*{\bf Strategy 2: Threshold update.}
We propose to use the PTR-order to sort the terms in the term-concatenation lists, namley s-list, a-list and l-list. The proposed PTR-order not only ensure correctness, it also reduce the search space. This is due to quick rise in the threshold value during the stages by choosing candidates that are expected to have high relevance score, and which is also experimentally verified in Section~\ref{sec:exp}. \qed

Now we justify that threshold update do not generate any false negative or false positive pattern as a top-k pattern. At any stage, the threshold update explores the search space by giving preference to the patterns having higher relevance as compared to other candidates at that stage. Thus, if a pattern is discarded to be a candidate it must have a relevance score lower than the lowest relevance of the current candidate list. Thus, it can not be a valid top-k pattern as we have already a set of $k$ patterns having relevance more than this pattern. If we assume that the other strategies are correct, threshold update only changes the order in which the search space is explored, i.e., all the patterns will be explored to be a valid candidate but in a different order. Now if a pattern is in the actual top-k list, it has to be in the final candidate list reported through threshold update. If it is not, then the lowest threshold of the candidates in the reported list will be lesser than this missing pattern, and this contradicts the fact that all the candidates were explored in the search space. Thus if all other strategies are correct, the threshold update returns the correct top-k list. 

The above two strategies manage threshold value so as to avoid generating many unpromising patterns. However, we still need to scan the search tree  substantially to decide on top-k relevant patterns.  For this, we suggest two pruning strategies,  namely {\it depth pruning} and {\it width pruning}~\cite{yin2012uspan}, that help in early termination with the best top-k patterns. 

In $TopKMintra$, the projected data is scanned to find terms for concatenation. Three lists, viz., $l$-list, $a$-list and $s$-list, are maintained that contains terms for respective concatenation operation. In width pruning, while maintaining the $l$-list and $s$-list for the future exploration corresponding to the current trajectory-pattern, only promising items are appended to the respective list. However, in a depth pruning strategy, further exploration of a trajectory-pattern is stopped when no more high relevant pattern as an extension of the current pattern in possible based on the information from the projected data, i.e., it  stopped going deeper into the search tree. 


\paragraph*{\bf Strategy 3: Width Pruning~\cite{yin2012uspan}.}
To avoid selecting unpromising items for the $s$-list and $l$-list, we adopt a width pruning strategy. The strategy is based on sequence-relevance downward closure property (SDCP) similar to USpan. To define SDCP, we first discuss the matching sequence-relevance of a trajectory-pattern.
\begin{definition}(Matching Sequence-relevance) For a pattern $T$, its matching sequence-relevance, denoted by $MSR(T)$, is defined as the sum of the sequence-relevance (Definition \ref{defn:sr}) of all its matching wLAS-sequences  $\alpha \in {\cal PD}$ (Definition \ref{defn:match}). i.e.,
$$ MSR(T) = \sum_{T \precsim \alpha, \alpha\in{\cal PD}} {\cal R}(\alpha)$$
\end{definition}
\begin{theorem}(Sequence-relevance Downward Closure Property) For a given wLAS-sequence database, and two trajectory patterns $T_1$ and $T_2$ such that $T_1 \subseteq T_2$, we have $$ MSR(T_2) \le MSR(T_1)$$
\end{theorem}
\begin{proof}
Let $\alpha \in{\cal PD}$ such that $T_2 \precsim \alpha$, i.e., there exists $\alpha_2 \subseteq \alpha$ s.t. $T_2 \sim \alpha_2$. Since $T_1 \subseteq T_2$, there must exists a $\alpha_1 \subseteq \alpha_2$ such that $T_1 \sim \alpha_1$. This implies that $T_1 \precsim \alpha$, i.e.,
$$ \{\alpha ~|~ \alpha \in {\cal PD} \wedge  T_2 \precsim \alpha\} \subseteq \{\beta ~|~ \beta \in {\cal PD} \wedge  T_1 \precsim \beta\}$$
Hence,
$$ MSR(T_2) \le MSR(T_1)$$
\end{proof}
For the term under consideration for the concatenation operation, we compute the matching sequence-relevance corresponding to all the matching sequences. If the matching sequence-relevance is more than the current threshold, the term is promising. Otherwise, the term is unpromising. 
\paragraph*{\bf Strategy 4: Depth Pruning \cite{yin2012uspan}.}  If the projected trajectory relevance of a trajectory pattern $T$ in wLAS-sequence data ${\cal PD}$, i.e., $PTR(T)$ (Definition~\ref{defn:PTR}), is less than the current threshold then $T$ and its offspring cannot be relevant pattern. In this situation, we can stop going deeper into the search tree and rather backtrack. The depth pruning strategy is based on the following result.

\begin{theorem}
Given a trajectory pattern $T$ and wLAS-sequence data ${\cal PD}$, the relevance of $T$ and $T$'s childern in the search tree are no more than $PTR(T)$. 
\end{theorem}
\begin{proof}
For $T$, a given trajectory-pattern, and for any $\alpha \in {\cal PD}$ such that $T\precsim \alpha$, say, the projected subsequence of $T$ in $\alpha$ is $\overline{\alpha_T}$ (Definition~\ref{defn:pivot-match}). Now, for any term $t$ in the projected subsequence $\overline{\alpha_T}$, the relevance of any of the match of $t$ in $\overline{\alpha_T}$ is no more than the total relevance of the projected subsequence ${\cal R}_{rest}(T,\alpha)$. 

Let $T^\prime$ be the pattern generated by concatenating the term $t$ in $T$. From the above argument, it is easy to observe that the maximum relevance of $T^\prime$ in $\alpha$, i.e., ${\cal R}_{max}(T^\prime,\alpha)$ can be no more that the  ${\cal R}_{PM}(T,\alpha) + {\cal R}_{rest} (T,\alpha)$ ($= PTR (T,\alpha)$). Therefore,
\begin{eqnarray*}
{\cal R}(T^\prime,{\cal PD})& = & \sum_{\alpha\in{\cal PD}}{\cal R}_{max}(T^\prime,\alpha)\\ 
            & \le & \sum_{\alpha\in{\cal PD}}{PTR(T,\alpha)}\\
            &=& PTR(T)
\end{eqnarray*}
Hence the result.
\end{proof}

\begin{algorithm*}[!h]
	\caption{TopKMintra}
    \label{algo:TopKMintra}
	\begin{algorithmic}[1]
    	\Inputs{Anonymized database ${\cal PD}$ in wLAS-sequence form,\\ user provided top-k parameter `k'}
    	\Output{ $TopKList$ containing Top-k spatio-textual sequential pattern}
    	\GlobalState{$TopKList$: list of size `k'to store top-k patterns\\ $threshold$: To compare relevance of a patterns}
		\Initialize{\strut Projected-relevance Matrix (PRM)\\
		            \strut lookup table for activities (Act)\\
		            \strut $TopKList \gets \phi$\\
		            \strut $threshold \gets preInsertion({\cal PD}, k,TopKList)$
		}
		\State \hspace*{1em}
  		\State $Node \gets root$   
		\State $s\_loc-list \gets \mbox{Scan PRM to find location for serial concatenation}$ 
		\State Perform width pruning to remove unpromising items ($u\_loc$) from $s\_loc-list$
		\For{each $loc \in s\_loc-list$}
		\State $Node \gets Node + Serial(loc)$ 
		\State Compute $Node.PTR$  \Comment{Using pivot match, pivot relevance and subsequence relevance }
        \If{$(Node.PTR >=  threshold)$}
		\State Call $TopKMintra\_update\_act(Node,PRM,Act,k,TopKList,empty\-list)$
		\State Call $TopKMintra\_update\_i-loc(Node,PRM,Act,k,TopKList)$
		\EndIf
		\EndFor
		\State return $TopKList$
	\end{algorithmic}
\end{algorithm*}

\paragraph*{\bf TopKMintra Algorithm.} We are now ready to discuss our TopKMintra algorithm. The algorithm is illustrated in Algorithm~\ref{algo:TopKMintra} with two subroutines Algorithm~\ref{algo:TopKMintraUpdateAct} and Algorithm~\ref{algo:TopKMintraUpdate-x-loc}. The input for TopKMintra is an anonymous database ${\cal PD}$ in wLAS-sequence form and the user provided top-k parameter $k$; the output contains the top-k high relevant trajectory-patterns from ${\cal PD}$ in the $TopKList$. 

The main algorithm TopKMintra first initializes data structure projected-relevance matrix (PRM) similar to the utility matrix in \cite{yin2012uspan}, a lookup table for activities (ACT), TopKList and initial threshold. The projected-relevance matrix for each wLAS-sequence contains a matrix in loc-ids (present in the sequence) and term-id. Each entry in the matrix is a tuple where the first value is the relevance of the matching loc-id in the term-id (zero if it is not present) and the second value is the remaining subsequence relevance.  The two advantage of computing PRM is that-- first, we can find list of term concatenation (line 7 in Algorithm~\ref{algo:TopKMintra}, line 1 in Algorithm~\ref{algo:TopKMintraUpdateAct} and Algorithm~\ref{algo:TopKMintraUpdate-x-loc}); and second, we can compute $Node.relevance$ and $Node.PTR$ which is the relevance and projected trajectory relevance of the current trajectory-pattern respectively (line 11 in Algorithm~\ref{algo:TopKMintra}, line 4 in Algorithm~\ref{algo:TopKMintraUpdateAct} and line 5 in Algorithm~\ref{algo:TopKMintraUpdate-x-loc}). The data structure $Act$ is an inverted list of (activity,term-id) to the list of wLAS-sequence-ids. This stores the information of the list of wLAS-sequence in which the activity is present at a particular term-id. 

The two subroutines Algorithm~\ref{algo:TopKMintraUpdateAct} and Algorithm~\ref{algo:TopKMintraUpdate-x-loc} are required to maintain the order of execution (line 13-14 in Algorithm~\ref{algo:TopKMintra}, line 10-11 in Algorithm~\ref{algo:TopKMintraUpdateAct} and line 7-8 in ALgorithm~\ref{algo:TopKMintraUpdate-x-loc}) and restricting $l$-concatenation after the $a$-concatenation (line 10-11 in Algorithm~\ref{algo:TopKMintraUpdateAct}) to avoid  repetition. The algorithm~\ref{algo:TopKMintraUpdate-x-loc} is a common routine for $l$-concatenation and $s$-concatenation where $x$-loc represent both $l$-loc and $s$-loc. The width pruning while extending the location is applied (line 8 in Algorithm~\ref{algo:TopKMintra} and line 2 in Algorithm~\ref{algo:TopKMintraUpdate-x-loc}) to consider only promising terms for concatenation. To backtrack instead of going deeper and returning with nothing, we have applied depth pruning (line 12-15 in Algorithm~\ref{algo:TopKMintra}, line 9-12 in Algorithm~\ref{algo:TopKMintraUpdateAct} and line 6-9 in Algorithm~\ref{algo:TopKMintraUpdate-x-loc}). The pattern is valid only when we append an activity (line 2-3 in Algorithm~\ref{algo:TopKMintraUpdateAct}. For every such valid pattern, we check if the relevance of the generated pattern is more than the current threshold (line 5 in Algorithm~\ref{algo:TopKMintraUpdateAct}). For relevant pattern, we update the $TopKList$ by appending the trajectory-pattern at an appropriate place in TopKList and update the current threshold (line 6-7 in Algorithm~\ref{algo:TopKMintraUpdateAct}. The routine of the TopKMintra is self explanatory which recursively calls various routines for updating $TopKList$ and return with  top-k most relevant trajectory-patterns in wLAS-sequence data ${\cal PD}$.

\begin{algorithm*}[!t]
	\caption{TopKMintra\_update\_act}
    \label{algo:TopKMintraUpdateAct}
	\begin{algorithmic}[1]
  		\State $a-list \gets$ Scan PRM to find activities for i-concatenation
		\For{each $act \in a-list$}
		\State $Node \gets Node + I(act)$ 
		\State Compute $Node.relevance$ and $Node.PTR$ using PRM 
		\If{$(Node.relevance >=  threshold)$}
		\State Insert $Node.seq$ in $TopKList$ maintaining size at-most $k$
		\State Update $threshold$
		\EndIf
		\If{$(Node.PTR >=  threshold)$}
		\State Call $TopKMintra\_update\_act(Node,PRM,Act,k,TopKList, tail(a-list))$
		\State Call $TopKMintra\_update\_s-loc(Node,PRM,Act,k,TopKList)$
		\EndIf
		\EndFor
	\end{algorithmic}
\end{algorithm*}

\begin{algorithm*}[!t]
	\caption{TopKMintra\_update\_x-loc}
    \label{algo:TopKMintraUpdate-x-loc}
	\begin{algorithmic}[1]
  		\State $x\_loc-list \gets$ Scan PRM to find location for x-concatenation
  		\State Perform width pruning to remove Unpromising items ($u\_loc$) from $x\_loc-list$
		\For{each $lox \in x\_loc-list$}
		\State $Node \gets Node + X(loc)$ 
		\State Compute $Node.PTR$ using PRM 
		\If{$(Node.PTR >=  threshold)$}
		\State Call $TopKMintra\_update\_act(Node,PRM,Act,k,TopKList, empty-list)$
		\State Call $TopKMintra\_update\_\overline{x}-loc(Node,PRM,Act,k,TopKList)$
		\EndIf
		\EndFor
	\end{algorithmic}
\end{algorithm*}


%% file: exp.tex
\begin{table}
	
	\begin{center}
		\caption{$ Characteristics \, of \, real \, datasets$ \label{Tab:datasets}}
		\resizebox{\textwidth}{!}{\begin{tabular}{ l l  l l l l}
			\hline
			\bfseries{Dataset} & \bfseries{\#Trajectories} & \bfseries{Avg. trajectory length (A)}& \bfseries{\#Locations (D)} & \bfseries{\#Activity (D)} & \bfseries{Max. trajectory length}  \\ \hline
			TDrive & 10,219 & 9.6 & 15,481 & 631 & 10  \\
			FourSquare & 10,567 & 8.2 & 13,946 & 631 & 12 \\
			 \hline
		\end{tabular}}
	\end{center}
\end{table}
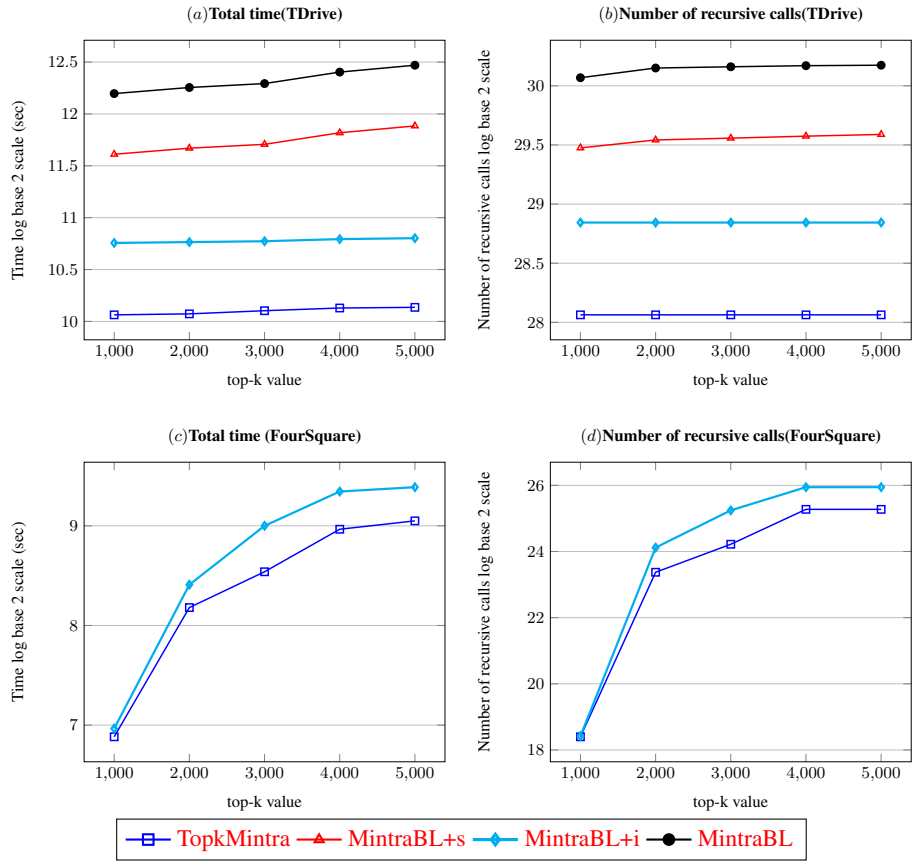
\begin{figure}[tbp]
	\begin{center}
		\begin{tabular}{cc}

			\begin{minipage}{.5\textwidth}
				\begin{tikzpicture}[scale=0.7]
				\begin{axis}[
				title={$(a)$\textbf{Total time(TDrive)}},
				xlabel={top-k value },
				ylabel={Time log base 2 scale (sec)},
				enlarge y limits=true,
				ymajorgrids=true,
				legend to name=t1,
				legend style={at={($(0,0)+(1cm,1cm)$)},legend columns=4,fill=none,draw=black,anchor=center, align=center},
				]
				
				\addplot
				[
				color=blue,
				mark=square,
				style=thick,
				]
				coordinates {
					(1000,{log2(1070)})(2000,{log2(1077)})(3000,{log2(1100)})(4000,{log2(1120)})(5000,{log2(1125)})
				};
				\addlegendentry{TopkMintra}
				\addplot
				[
				color=red,
				mark=triangle,
				style=thick,
				]
				coordinates {
					(1000,{log2(3129)})(2000,{log2(3259)})(3000,{log2(3344)})(4000,{log2(3612)})(5000,{log2(3782)})
				};    
				\addlegendentry{MintraBL+s}
				\addplot
				[
				color=cyan,
				mark=diamond,
				style=very thick,
				]
				coordinates {
					(1000,{log2(1730)})(2000,{log2(1740)})(3000,{log2(1750)})(4000,{log2(1775)})(5000,{log2(1786)})
				};    
				\addlegendentry{MintraBL+i}
				\addplot
				[
				color=black,
				mark=*,
				style=thick,
				]
				coordinates {
					(1000,{log2(4693)})(2000,{log2(4888)})(3000,{log2(5018)})(4000,{log2(5418)})(5000,{log2(5673)})
				};    
				\addlegendentry{MintraBL}

				\end{axis}
				
				\end{tikzpicture}
			\end{minipage}
			
			&
			
			\begin{minipage}{.5\textwidth}
				\begin{tikzpicture}[scale=0.7]
				\begin{axis}[
				title={$(b)$\textbf{Number of recursive calls(TDrive)}},
				xlabel={top-k value },
				ylabel={Number of recursive calls log base 2 scale},
				enlarge y limits=true,
				ymajorgrids=true,
				]
				
				\addplot
				[
				color=blue,
				mark=square,
				style=thick,
				]
				coordinates {
					(1000,{log2(280489663)})(2000,{log2(280489763)})(3000,{log2(280489863)})(4000,{log2(280489963)})(5000,{log2(280489980)})
				};
				\addplot
				[
				color=red,
				mark=triangle,
				style=thick,
				]
				coordinates {
					(1000,{log2(746369331)})(2000,{log2(781850069)})(3000,{log2(790342520)})(4000,{log2(799582850)})(5000,{log2(807498744)})
				};    
				\addplot
				[
				color=cyan,
				mark=diamond,
				style=very thick,
				]
				coordinates {
					(1000,{log2(481811638)})(2000,{log2(481811738)})(3000,{log2(481811838)})(4000,{log2(481811938)})(5000,{log2(481811968)})
				};    
				\addplot
				[
				color=black,
				mark=*,
				style=thick,
				]
				coordinates {
					(1000,{log2(1126038310)})(2000,{log2(1191974370)})(3000,{log2(1200619938)})(4000,{log2(1208384678)})(5000,{log2(1211651571)})
				};  
				
				\end{axis}
				
				\end{tikzpicture}
			\end{minipage}\\ \\

			\begin{minipage}{.5\textwidth}
				\begin{tikzpicture}[scale=0.7]
				\begin{axis}[
				title={$(c)$\textbf{Total time (FourSquare)}},
				xlabel={top-k value },
				ylabel={Time log base 2 scale (sec)},
				enlarge y limits=true,
				ymajorgrids=true,
				]
				
				\addplot[
				color=blue,
				mark=square,
				style=thick,
				]
				coordinates {
					(1000,{log2(118)})(2000,{log2(290)})(3000,{log2(372)})(4000,{log2(500)})(5000,{log2(530)})
				};
				\addplot[
				color=cyan,
				mark=diamond,
				style=very thick,
				]
				coordinates {
					(1000,{log2(125)})(2000,{log2(340)})(3000,{log2(512)})(4000,{log2(650)})(5000,{log2(670)})
				};    
				
				\end{axis}
				
				\end{tikzpicture}
			\end{minipage}
			
			&
			
			\begin{minipage}{.5\textwidth}
				\begin{tikzpicture}[scale=0.7]
				\begin{axis}[
				title={$(d)$\textbf{Number of recursive calls(FourSquare)}},
				xlabel={top-k value },
				ylabel={Number of recursive calls log base 2 scale},
				enlarge y limits=true,
				ymajorgrids=true,
				]
				
				\addplot[
				color=blue,
				mark=square,
				style=thick,
				]
				coordinates {
					(1000,{log2(345337)})(2000,{log2(10859397)})(3000,{log2(19531356)})(4000,{log2(40576850)})(5000,{log2(40576900)})
				};
				\addplot[
				color=cyan,
				mark=diamond,
				style=very thick,
				]
				coordinates {
					(1000,{log2(350357)})(2000,{log2(18212375)})(3000,{log2(39746629)})(4000,{log2(64576850)})(5000,{log2(64576900)})
				};

				\end{axis}
				
				\end{tikzpicture}
			\end{minipage}
			
		\end{tabular}
		\ref{t1}
		\caption{Performance evaluation on TDrive and FourSquare datasets for cell size (0.0020,0.0040). MintraBL+s and MintraBL didn't terminate on FourSquare dataset for more than 24 hours. Hence, we don't report their readings for FourSquare dataset.}
		\label{fig:tdrive_foursquareperf1}
	\end{center}
\end{figure}
In this section, we compare the performance of our proposed technique TopKMintra against three algorithms-- a baseline algorithm MintraBL (refer \ref{sec:mintraBL}) and two other algorithms designed using baseline algorithm, namely MintraBL+i, MintraBL+s. The difference between the three algorithms with that of TopKMintra is the inclusion or non-inclusion of the two pruning strategies as discussed in section \ref{sec:topKMintra}, i.e., {\it threshold initialization} (to overcome the limitation of initial threshold value as 0), and {\it threshold update} (to optimally update threshold at each step using PTR based sorting order). In the baseline algorithm, none of the two strategies are included, whereas in MintraBL+i only threshold initialization strategy is included, and in MintraBL+s only threshold update strategy is included. We have implemented our algorithms in Java, and all the experiments were conducted on a Windows 2012 Server with Intel Xeon(R) CPU of 2.00 GHz and 64 GB RAM. 

We conduct experiments on TDrive \cite{tdrive_2,tdrive_1} and FourSquare \cite{ankita17} trajectory datasets. The characteristics of the datasets are shown in Table \ref{Tab:datasets}. The TDrive dataset contains one-week trajectories of 10,357 taxis. We preprocess the datasets to construct anonymized user trajectories with both $k$-anonymity and $l$-diversity set to 3. We have used HaverSine method to compute the distance between two GPS coordinates for any comparisons. Since we extract common patterns over anonymized data, meaningful patterns can be extracted only if there is enough overlap between the trajectories. This requires filtering of trajectories for a fixed region to find enough patterns. To incorporate this in TDrive dataset, we have considered the taxi locations with longitude and latitude between (115 to 120) and (37 to 42). For the specified region, we have considered taxi trajectories of length at most ten where every location of the trajectory is a stoppage, i.e., the taxi has stopped at the location for the considerate amount of time. As a preprocessing step, we discretized the region into $1000 \times 1000$ square cells and mapped the taxi trajectories into the corresponding cell-trajectories. To create  activity-trajectory data, some of the cell locations of the discretized region are assigned a randomly chosen activity, and the same is associated with the corresponding cell-trajectories. The extracted cell-trajectories are anonymized using the $k$-anonymity and $l$-diversity as specified above. 

For FourSquare dataset, check-in locations with longitude and latitude between (35 to 36) and (139 to 140) are selected. At most twelve points are extracted for every trajectory and mapped to a $200 \times 200$ square cells. Every check-in location in FourSquare dataset has an activity associated with it. This cell-trajectory data is anonymized using $k$-anonymity and $l$-diversity. 

\begin{figure}[tbp]
	\begin{center}
		\begin{tabular}{cc}

			\begin{minipage}{.5\textwidth}
				\begin{tikzpicture}[scale=0.7]
				\begin{axis}[
				title={$(a)$\textbf{Memory consumed(TDrive)}},
				xlabel={top-k },
				ylabel={Memory (MB) log base 2 scale},
				enlarge y limits=true,
				ymajorgrids=true,
				legend to name=t2,
				legend style={at={($(0,0)+(1cm,1cm)$)},legend columns=4,fill=none,draw=black,anchor=center, align=center},
				]
				
				\addplot
				[
				color=blue,
				mark=square,
				style=thick,
				]
				coordinates {
					(1000,{log2(5963)})(2000,{log2(5964)})(3000,{log2(5966)})(4000,{log2(5968)})(5000,{log2(5970)})
				};
				\addlegendentry{TopkMintra}
				\addplot
				[
				color=red,
				mark=triangle,
				style=thick,
				]
				coordinates {
					(1000,{log2(6004)})(2000,{log2(6008)})(3000,{log2(6010)})(4000,{log2(6012)})(5000,{log2(6015)})
				};    
				\addlegendentry{MintraBL+s}
				\addplot
				[
				color=cyan,
				mark=diamond,
				style=very thick,
				]
				coordinates {
					(1000,{log2(5975)})(2000,{log2(5977)})(3000,{log2(5978)})(4000,{log2(5980)})(5000,{log2(5985)})
				};    
				\addlegendentry{MintraBL+i}
				\addplot
				[
				color=black,
				mark=*,
				style=thick,
				]
				coordinates {
					(1000,{log2(6020)})(2000,{log2(6030)})(3000,{log2(6040)})(4000,{log2(6045)})(5000,{log2(6050)})
				};    
				\addlegendentry{MintraBL}

				\end{axis}
				
				\end{tikzpicture}
			\end{minipage}
			
			&
			
			\begin{minipage}{.5\textwidth}
				\begin{tikzpicture}[scale=0.7]
				\begin{axis}[
				title={$(b)$\textbf{Memory Consumed(FourSquare)}},
				xlabel={top-k },
				ylabel={Memory (MB) log base 2 scale},
				enlarge y limits=true,
				ymajorgrids=true,
				]
				
				\addplot
				[
				color=blue,
				mark=square,
				style=thick,
				]
				coordinates {
					(1000,{log2(1162)})(2000,{log2(1505)})(3000,{log2(1600)})(4000,{log2(1958)})(5000,{log2(2022)})
				};
				\addplot
				[
				color=cyan,
				mark=diamond,
				style=very thick,
				]
				coordinates {
					(1000,{log2(1241)})(2000,{log2(1743)})(3000,{log2(2088)})(4000,{log2(2292)})(5000,{log2(2300)})
				};

				\end{axis}
				
				\end{tikzpicture}
			\end{minipage}			
		\end{tabular}
		\ref{t2}
		\caption{Memory consumption on TDrive and FourSquare datasets. MintraBL+s and MintraBL didn't terminate on FourSquare dataset for more than 24 hours. Hence, we don't report their readings for FourSquare dataset.}
		\label{fig:tdrive_foursquarememory}
	\end{center}
\end{figure}
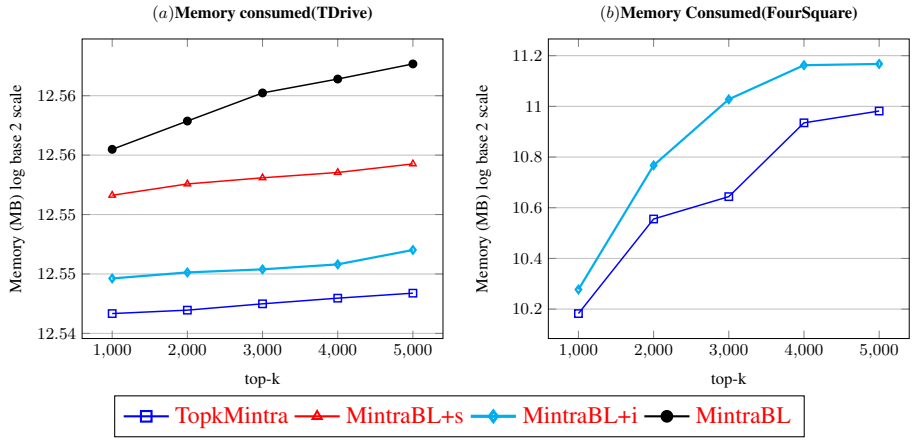

\subsection{Performance Evaluation}
We study the performance gain of ToKMintra with the other three Algorithm $MintraBL$, $MintraBL+i$ and $MintraBL+s$ in terms of absolute execution time, candidate set size and main memory for different top-k values. We also study the effect of dataset-size, trajectory-length, cell size, $k$-anonymity, and $l$-diversity over execution time and the number of candidate trajectory-patterns generated measured by the number of recursive calls.
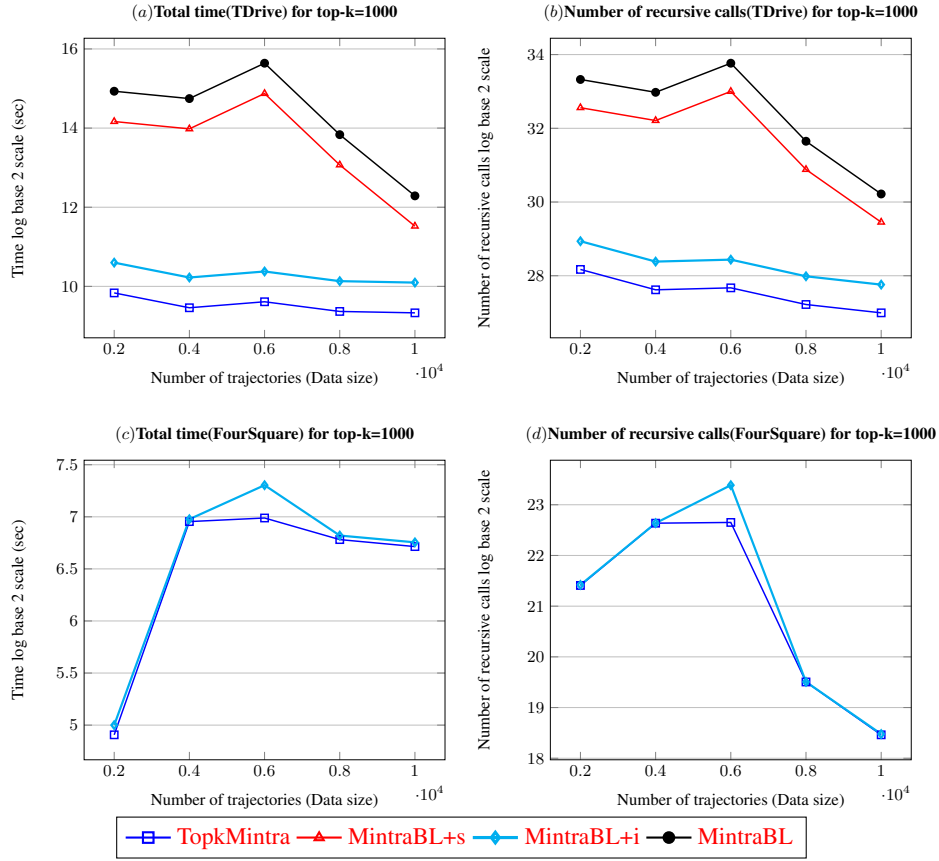
\begin{figure}[tbp]
	\begin{center}
		\begin{tabular}{cc}

			\begin{minipage}{.5\textwidth}
				\begin{tikzpicture}[scale=0.7]
				\begin{axis}[
				title={$(a)$\textbf{Total time(TDrive) for top-k=1000}},
				xlabel={Number of trajectories (Data size)},
				ylabel={Time log base 2 scale (sec)},
				enlarge y limits=true,
				ymajorgrids=true,
				legend to name=t3,
				legend style={at={($(0,0)+(1cm,1cm)$)},legend columns=4,fill=none,draw=black,anchor=center, align=center},
				]
				
				\addplot
				[
				color=blue,
				mark=square,
				style=thick,
				]
				coordinates {
					(2000,{log2(913)})(4000,{log2(704)})(6000,{log2(782)})(8000,{log2(660)})(10000,{log2(644)})
				};
				\addlegendentry{TopkMintra}
				\addplot
				[
				color=red,
				mark=triangle,
				style=thick,
				]
				coordinates {
					(2000,{log2(18381)})(4000,{log2(16162)})(6000,{log2(30029)})(8000,{log2(8591)})(10000,{log2(2940)})
				};    
				\addlegendentry{MintraBL+s}
				\addplot
				[
				color=cyan,
				mark=diamond,
				style=very thick,
				]
				coordinates {
					(2000,{log2(1552)})(4000,{log2(1196)})(6000,{log2(1329)})(8000,{log2(1122)})(10000,{log2(1094)})
				};    
				\addlegendentry{MintraBL+i}
				\addplot
				[
				color=black,
				mark=*,
				style=thick,
				]
				coordinates {
					(2000,{log2(31247)})(4000,{log2(27475)})(6000,{log2(51049)})(8000,{log2(14604)})(10000,{log2(4998)})
				};    
				\addlegendentry{MintraBL}

				\end{axis}
				
				\end{tikzpicture}
			\end{minipage}
			
			&
			
			\begin{minipage}{.5\textwidth}
				\begin{tikzpicture}[scale=0.7]
				\begin{axis}[
				title={$(b)$\textbf{Number of recursive calls(TDrive) for top-k=1000}},
				xlabel={Number of trajectories (Data size) },
				ylabel={Number of recursive calls log base 2 scale},
				enlarge y limits=true,
				ymajorgrids=true,
				]
				
				\addplot
				[
				color=blue,
				mark=square,
				style=thick,
				]
				coordinates {
					(2000,{log2(302303401)})(4000,{log2(206071616)})(6000,{log2(213772406)})(8000,{log2(156319859)})(10000,{log2(133620241)})
				};
				\addplot
				[
				color=red,
				mark=triangle,
				style=thick,
				]
				coordinates {
					(2000,{log2(6328251280)})(4000,{log2(4969828868)})(6000,{log2(8591866725)})(8000,{log2(1980262649)})(10000,{log2(734975901)})
				};    
				\addplot
				[
				color=cyan,
				mark=diamond,
				style=very thick,
				]
				coordinates {
					(2000,{log2(513915781)})(4000,{log2(350321747)})(6000,{log2(363413090)})(8000,{log2(265743760)})(10000,{log2(227154409)})
				};    
				\addplot
				[
				color=black,
				mark=*,
				style=thick,
				]
				coordinates {
					(2000,{log2(10758027176)})(4000,{log2(8448709075)})(6000,{log2(14606173432)})(8000,{log2(3366446504)})(10000,{log2(1249459031)})
				};  
				
				\end{axis}
				
				\end{tikzpicture}
			\end{minipage}\\ \\

			\begin{minipage}{.5\textwidth}
				\begin{tikzpicture}[scale=0.7]
				\begin{axis}[
				title={$(c)$\textbf{Total time(FourSquare) for top-k=1000}},
				xlabel={Number of trajectories (Data size) },
				ylabel={Time log base 2 scale (sec)},
				enlarge y limits=true,
				ymajorgrids=true,
				]
				
				\addplot[
				color=blue,
				mark=square,
				style=thick,
				]
				coordinates {
					(2000,{log2(30)})(4000,{log2(124)})(6000,{log2(127)})(8000,{log2(110)})(10000,{log2(105)})
				};
				\addplot[
				color=cyan,
				mark=diamond,
				style=very thick,
				]
				coordinates {
					(2000,{log2(32)})(4000,{log2(126)})(6000,{log2(158)})(8000,{log2(113)})(10000,{log2(108)})
				};    
				
				\end{axis}
				
				\end{tikzpicture}
			\end{minipage}
			
			&
			
			\begin{minipage}{.5\textwidth}
				\begin{tikzpicture}[scale=0.7]
				\begin{axis}[
				title={$(d)$\textbf{Number of recursive calls(FourSquare) for top-k=1000}},
				xlabel={Number of trajectories (Data size)},
				ylabel={Number of recursive calls log base 2 scale},
				enlarge y limits=true,
				ymajorgrids=true,
				]
				
				\addplot[
				color=blue,
				mark=square,
				style=thick,
				]
				coordinates {
					(2000,{log2(2783737)})(4000,{log2(6521337)})(6000,{log2(6587818)})(8000,{log2(743995)})(10000,{log2(361189)})
				};
				\addplot[
				color=cyan,
				mark=diamond,
				style=very thick,
				]
				coordinates {
					(2000,{log2(2804511)})(4000,{log2(6527968)})(6000,{log2(10964015)})(8000,{log2(744872)})(10000,{log2(365082)})
				};
				
				\end{axis}
				
				\end{tikzpicture}
			\end{minipage}
			
		\end{tabular}
		\ref{t3}
		\caption{Scalability on TDrive and FourSquare datasets for top-k=1000 and cell size (0.0020,0.0040). MintraBL+s and MintraBL didn't terminate on FourSquare dataset for more than 24 hours. Hence, we don't report their readings for FourSquare dataset.}
		\label{fig:tdrive_foursquare_scalability}
	\end{center}
\end{figure}
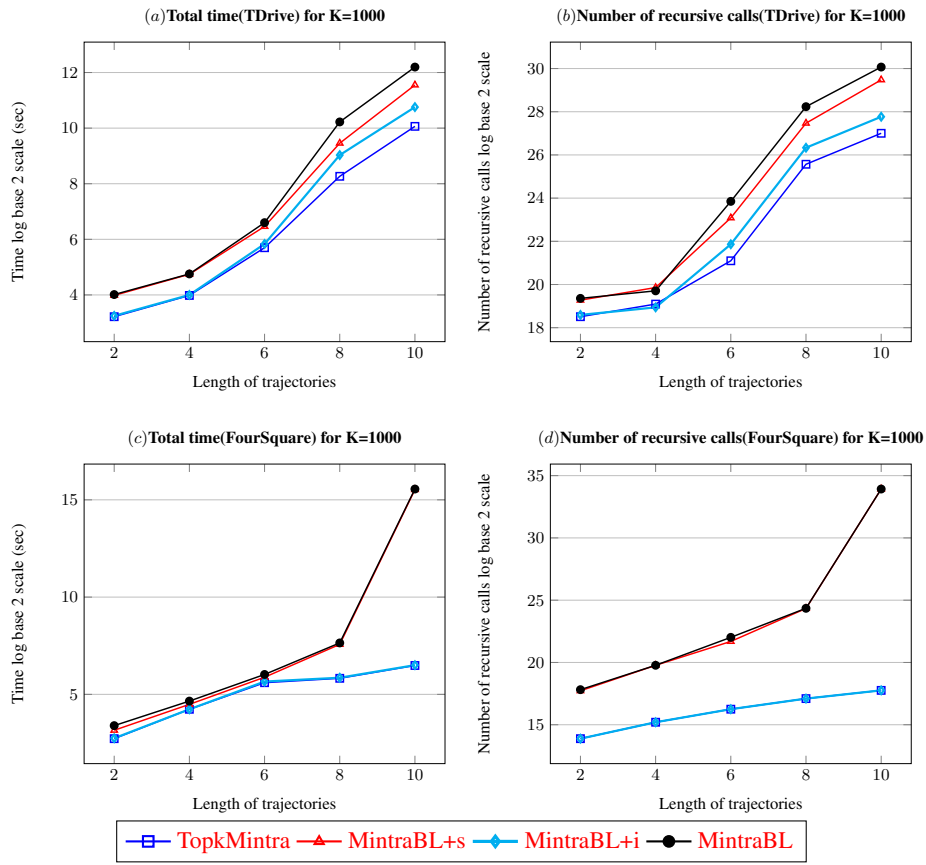
\begin{figure}[tbp]
	\begin{center}
		\begin{tabular}{cc}

			\begin{minipage}{.5\textwidth}
				\begin{tikzpicture}[scale=0.7]
				\begin{axis}[
				title={$(a)$\textbf{Total time(TDrive) for K=1000}},
				xlabel={Length of trajectories},
				ylabel={Time log base 2 scale (sec)},
				enlarge y limits=true,
				ymajorgrids=true,
				legend to name=t4,
				legend style={at={($(0,0)+(1cm,1cm)$)},legend columns=4,fill=none,draw=black,anchor=center, align=center},
				]
				
				\addplot
				[
				color=blue,
				mark=square,
				style=thick,
				]
				coordinates {
					(2,{log2(9.3)})(4,{log2(15.8)})(6,{log2(51.9)})(8,{log2(308)})(10,{log2(1070)})
				};
				\addlegendentry{TopkMintra}
				\addplot
				[
				color=red,
				mark=triangle,
				style=thick,
				]
				coordinates {
					(2,{log2(15.81)})(4,{log2(26.8)})(6,{log2(88.2)})(8,{log2(703)})(10,{log2(3000)})
				
				};    
				\addlegendentry{MintraBL+s}
				\addplot
				[
				color=cyan,
				mark=diamond,
				style=very thick,
				]
				coordinates {
						(2,{log2(9.5)})(4,{log2(16)})(6,{log2(57)})(8,{log2(523)})(10,{log2(1730)})
				};    
				\addlegendentry{MintraBL+i}
				\addplot
				[
				color=black,
				mark=*,
				style=thick,
				]
				coordinates {
					(2,{log2(16.15)})(4,{log2(27)})(6,{log2(97)})(8,{log2(1195)})(10,{log2(4693)})
				};    
				\addlegendentry{MintraBL}

				\end{axis}
				
				\end{tikzpicture}
			\end{minipage}
			
			&
			
			\begin{minipage}{.5\textwidth}
				\begin{tikzpicture}[scale=0.7]
				\begin{axis}[
				title={$(b)$\textbf{Number of recursive calls(TDrive) for K=1000}},
				xlabel={Length of trajectories},
				ylabel={Number of recursive calls log base 2 scale},
				enlarge y limits=true,
				ymajorgrids=true,
				]
				
				\addplot
				[
				color=blue,
				mark=square,
				style=thick,
				]
				coordinates {
					(2,{log2(373219)})(4,{log2(559419)})(6,{log2(2245201)})(8,{log2(49740656)})(10,{log2(134380583)})
				};
				\addplot
				[
				color=red,
				mark=triangle,
				style=thick,
				]
				coordinates {
					(2,{log2(634472)})(4,{log2(951012)})(6,{log2(8892131)})(8,{log2(185726718)})(10,{log2(746369331)})
				
				};    
				\addplot
				[
				color=cyan,
				mark=diamond,
				style=very thick,
				]
				coordinates {
						(2,{log2(395247)})(4,{log2(503345)})(6,{log2(3816841)})(8,{log2(84559115)})(10,{log2(228446991)})
				};    
				\addplot
				[
				color=black,
				mark=*,
				style=thick,
				]
				coordinates {
					(2,{log2(671919)})(4,{log2(855686)})(6,{log2(15116622)})(8,{log2(315735420)})(10,{log2(1126038310)})
				};  
				
				\end{axis}
				
				\end{tikzpicture}
			\end{minipage}\\ \\

			\begin{minipage}{.5\textwidth}
				\begin{tikzpicture}[scale=0.7]
				\begin{axis}[
				title={$(c)$\textbf{Total time(FourSquare) for K=1000}},
				xlabel={Length of trajectories },
				ylabel={Time log base 2 scale (sec)},
				enlarge y limits=true,
				ymajorgrids=true,
				]
				
				\addplot[
				color=blue,
				mark=square,
				style=thick,
				]
				coordinates {
					(2,{log2(6.6)})(4,{log2(18.8)})(6,{log2(48.5)})(8,{log2(56.7)})(10,{log2(89.5)})
				};
				\addplot
				[
				color=red,
				mark=triangle,
				style=thick,
				]
				coordinates {
				(2,{log2(8.9)})(4,{log2(22.4)})(6,{log2(58.9)})(8,{log2(190)})(10,{log2(47172)})
				};  
				\addplot[
				color=cyan,
				mark=diamond,
				style=very thick,
				]
				coordinates {
					(2,{log2(6.7)})(4,{log2(18.8)})(6,{log2(50.5)})(8,{log2(57.9)})(10,{log2(90.6)})
				};    
				\addplot[
				color=black,
				mark=*,
				style=thick,
				]
				coordinates {
				(2,{log2(10.5)})(4,{log2(25.17)})(6,{log2(65)})(8,{log2(200)})(10,{log2(48157)})
				};
				
				\end{axis}
				
				\end{tikzpicture}
			\end{minipage}
			
			&
			
			\begin{minipage}{.5\textwidth}
				\begin{tikzpicture}[scale=0.7]
				\begin{axis}[
				title={$(d)$\textbf{Number of recursive calls(FourSquare) for K=1000}},
				xlabel={Length of trajectories},
				ylabel={Number of recursive calls log base 2 scale},
				enlarge y limits=true,
				ymajorgrids=true,
				]
				
				\addplot[
				color=blue,
				mark=square,
				style=thick,
				]
				coordinates {
					(2,{log2(15105)})(4,{log2(37682)})(6,{log2(77884)})(8,{log2(140440)})(10,{log2(221872)})
				};
				\addplot
				[
				color=red,
				mark=triangle,
				style=thick,
				]
				coordinates {
				(2,{log2(216969)})(4,{log2(894526)})(6,{log2(3379907)})(8,{log2(21194837)})(10,{log2(16300557250)})
				};  
				\addplot[
				color=cyan,
				mark=diamond,
				style=very thick,
				]
				coordinates {
					(2,{log2(15108)})(4,{log2(37686)})(6,{log2(77897)})(8,{log2(140395)})(10,{log2(221915)})
				};
				\addplot[
				color=black,
				mark=*,
				style=thick,
				]
				coordinates {
					(2,{log2(230782)})(4,{log2(896527)})(6,{log2(4241720)})(8,{log2(21350145)})(10,{log2(16300557250)})
				};
				\end{axis}
				
				\end{tikzpicture}
			\end{minipage}
			
		\end{tabular}
		\ref{t4}
		\caption{Varying length of trajectory on TDrive and FourSquare datasets for top-k=1000 and cell size (0.0020,0.0040).}
		\label{fig:tdrive_foursquare_length}
	\end{center}
\end{figure}

\textbf{Effect of varying top-k value:} Fig.~\ref{fig:tdrive_foursquareperf1} shows the performance of the algorithms on TDrive, and FourSquare datasets for different top-k values. The number of trajectories and sequence length remain the same as specified in Table \ref{Tab:datasets}. It can be observed that the execution time and the number of candidates increase on increasing the top-k value. The threshold value depends upon the size of the TopKList. More candidates in the TopKList will lower the threshold value, and result in more exploration of the search space. We observe that TopKMintra performs the best and MintraBL is the slowest algorithm for both datasets. MintraBL+i performs better than MintraBL+s for TDrive dataset. We don't report the results for MintraBL+s and MintraBL on FourSquare dataset as the algorithms did not terminate for more than 24 hours. We also observed the memory requirement of our algorithms measured using Java-VisualVM which comes bundled with the JDK. The results are shown in Figure \ref{fig:tdrive_foursquarememory}. The results show a positive correlation between the number of candidates generated and the memory consumed by different algorithms.

\textbf{Effect of Scalability:} To measure the scalability of the two algorithms, we generate a
dataset containing different number of trajectories. The generated datasets have sequences in the range $2$K to
$10$K. The sequence length remains the same as specified in Table \ref{Tab:datasets} and top-k is set to 1000. The datasets were constructed by selecting the first 2K, 4K, 6K, and 8K transactions from TDrive and FourSquare datasets. The results are shown in Fig.\ref{fig:tdrive_foursquare_scalability}. The graph clearly shows that TopKMintra performs the best compared to other algorithms. It is interesting to observe that for both the data sets the running time and number of candidates increase with the number of transactions and then decrease. This is because we have fixed the region for extracting patterns, and as we increase the number of anonymized trajectories it increases the trajectory density in the fixed region.  TopKMintra can find top-k patterns quickly in denser regions as more trajectories overlap and the top-k threshold rises quickly during the mining process. This clearly indicates that patterns can be mined efficiently in the dense regions and, therefore, with more and more anonymized data mining of the pattern become easy.   


\textbf{Effect of Varying trajectory length:} In Fig.\ref{fig:tdrive_foursquare_length}, we show the effect of length of anonymized trajectory for different trajectory length. The number of transactions remains the same as specified in Table \ref{Tab:datasets} and top-k is set to 1000. It is clear from the graph that as we increase the length, the number of recursive calls and execution time increases with the length of the trajectory. It is trivial to observe that the search space increases with the trajectory length as more a-concatenation, s-concatenation, and l-concatenation operations are performed as inferred by the increase in candidates generated. We observe that TopKMintra and MintraBL+i perform significantly better than other variants for longer trajectory length. MintraBL+s performs similarly to TopKMintra for shorter length trajectories. However, it's performance degrades as trajectory length increases in terms of absolute time and number of recursive calls.

\textbf{Effect of Varying $k$-anonymity:} Figure \ref{fig:tdrive_kanon} shows the effect of varying $k$-anonymity on the performance of TopKMintra, and other algorithms. The size of anonymized regions associated with a trajectory increases to satisfy the increase in $k$-anonymity. Hence, the number of overlapping cells on the ground with the anonymized regions will increase leading to rise in the number of term-concatenations operations.  

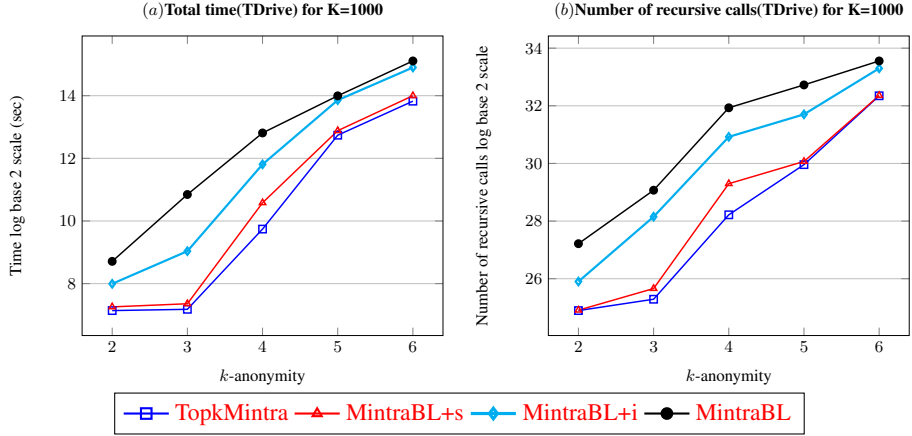
\begin{figure}[tbp]
	\begin{center}
		\begin{tabular}{cc}

			\begin{minipage}{.5\textwidth}
				\begin{tikzpicture}[scale=0.7]
				\begin{axis}[
				title={$(a)$\textbf{Total time(TDrive) for K=1000}},
				xlabel={$k$-anonymity},
				ylabel={Time log base 2 scale (sec)},
				enlarge y limits=true,
				ymajorgrids=true,
				legend to name=t5,
				legend style={at={($(0,0)+(1cm,1cm)$)},legend columns=4,fill=none,draw=black,anchor=center, align=center},
				]
				
				\addplot
				[
				color=blue,
				mark=square,
				style=thick,
				]
				coordinates {
					(2,{log2(141)})(3,{log2(145)})(4,{log2(856)})(5,{log2(6834)})(6,{log2(14512)})
				};
				\addlegendentry{TopkMintra}
				\addplot
				[
				color=red,
				mark=triangle,
				style=thick,
				]
				coordinates {
					(2,{log2(153)})(3,{log2(164)})(4,{log2(1532)})(5,{log2(7554)})(6,{log2(16377)})
				
				};    
				\addlegendentry{MintraBL+s}
				\addplot
				[
				color=cyan,
				mark=diamond,
				style=very thick,
				]
				coordinates {
						(2,{log2(255)})(3,{log2(525)})(4,{log2(3586)})(5,{log2(14845)})(6,{log2(30574)})
				};    
				\addlegendentry{MintraBL+i}
				\addplot
				[
				color=black,
				mark=*,
				style=thick,
				]
				coordinates {
					(2,{log2(419)})(3,{log2(1840)})(4,{log2(7184)})(5,{log2(16337)})(6,{log2(35357)})
				};    
				\addlegendentry{MintraBL}

				\end{axis}
				
				\end{tikzpicture}
			\end{minipage}
			
			&
			
			\begin{minipage}{.5\textwidth}
				\begin{tikzpicture}[scale=0.7]
				\begin{axis}[
				title={$(b)$\textbf{Number of recursive calls(TDrive) for K=1000}},
				xlabel={$k$-anonymity},
				ylabel={Number of recursive calls log base 2 scale},
				enlarge y limits=true,
				ymajorgrids=true,
				]
				
				\addplot
				[
				color=blue,
				mark=square,
				style=thick,
				]
				coordinates {
					(2,{log2(31197892)})(3,{log2(40966142)})(4,{log2(312192843)})(5,{log2(1043362114)})(6,{log2(5464754202)})
				};
				\addplot
				[
				color=red,
				mark=triangle,
				style=thick,
				]
				coordinates {
					(2,{log2(31500107)})(3,{log2(53002193)})(4,{log2(660020155)})(5,{log2(1126008298)})(6,{log2(5497030432)})
				
				};    
				\addplot
				[
				color=cyan,
				mark=diamond,
				style=very thick,
				]
				coordinates {
						(2,{log2(62699997)})(3,{log2(298160783)})(4,{log2(2035649808)})(5,{log2(3497030432)})(6,{log2(10569732122)})
				};    
				\addplot
				[
				color=black,
				mark=*,
				style=thick,
				]
				coordinates {
					(2,{log2(155869716)})(3,{log2(564056149)})(4,{log2(4095844015)})(5,{log2(7095832041)})(6,{log2(12633038681)})
				};  
				
				\end{axis}
				
				\end{tikzpicture}
			\end{minipage}
			
		\end{tabular}
		\ref{t5}
		\caption{Varying $k$-anonymity on TDrive dataset with trajectory length 5 for top-k=1000 and cell-size (0.0020,0.0040). }
		\label{fig:tdrive_kanon}
	\end{center}
\end{figure}

\textbf{Effect of Varying $l$-diversity:} A similar effect is observed on the performance of algorithms when $l$-diversity is increased as shown in Figure \ref{fig:tdrive_ldiv}. The size of anonymized region associated with each trajectory and the number of activities associated with an anonymized region increase to satisfy $l$-diversity. We observe that MintraBL+s performs better than MintraBL+i for higher values of $l$-diversity. We believe that the strategy to choose a better path to find top-k patterns has more impact on the performance of top-k algorithm compared to the pre-insertion strategy with increase in the search space.

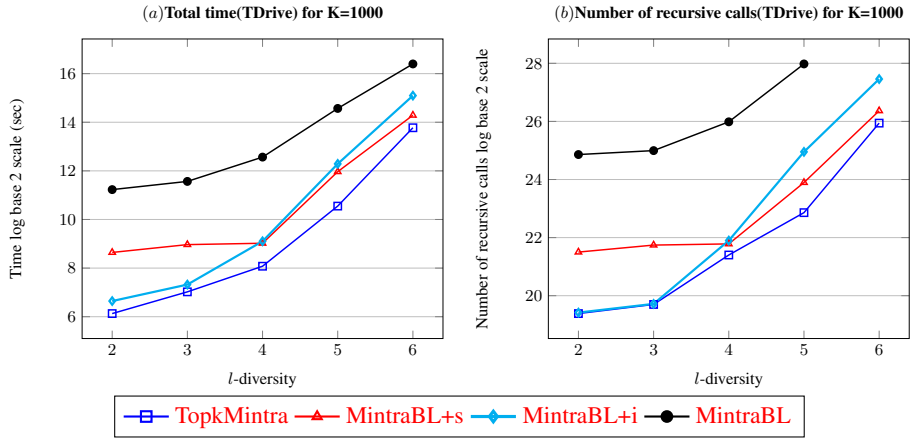
\begin{figure}[tbp]
	\begin{center}
		\begin{tabular}{cc}

			\begin{minipage}{.5\textwidth}
				\begin{tikzpicture}[scale=0.7]
				\begin{axis}[
				title={$(a)$\textbf{Total time(TDrive) for K=1000}},
				xlabel={$l$-diversity},
				ylabel={Time log base 2 scale (sec)},
				enlarge y limits=true,
				ymajorgrids=true,
				legend to name=t6,
				legend style={at={($(0,0)+(1cm,1cm)$)},legend columns=4,fill=none,draw=black,anchor=center, align=center},
				]
				
				\addplot
				[
				color=blue,
				mark=square,
				style=thick,
				]
				coordinates {
					(2,{log2(70)})(3,{log2(130)})(4,{log2(270)})(5,{log2(1500)})(6,{log2(14000)})
				};
				\addlegendentry{TopkMintra}
				\addplot
				[
				color=red,
				mark=triangle,
				style=thick,
				]
				coordinates {
					(2,{log2(400)})(3,{log2(500)})(4,{log2(520)})(5,{log2(4000)})(6,{log2(20000)})
				
				};    
				\addlegendentry{MintraBL+s}
				\addplot
				[
				color=cyan,
				mark=diamond,
				style=very thick,
				]
				coordinates {
					(2,{log2(100)})(3,{log2(160)})(4,{log2(550)})(5,{log2(5000)})(6,{log2(35000)})
				};    
				\addlegendentry{MintraBL+i}
				\addplot
				[
				color=black,
				mark=*,
				style=thick,
				]
				coordinates {
					(2,{log2(2400)})(3,{log2(3032)})(4,{log2(6065)})(5,{log2(24295)})(6,{log2(86400)})
				};    
				\addlegendentry{MintraBL}

				\end{axis}
				
				\end{tikzpicture}
			\end{minipage}
			
			&
			
			\begin{minipage}{.5\textwidth}
				\begin{tikzpicture}[scale=0.7]
				\begin{axis}[
				title={$(b)$\textbf{Number of recursive calls(TDrive) for K=1000}},
				xlabel={$l$-diversity},
				ylabel={Number of recursive calls log base 2 scale},
				enlarge y limits=true,
				ymajorgrids=true,
				]
				
				\addplot
				[
				color=blue,
				mark=square,
				style=thick,
				]
				coordinates {
					(2,{log2(684633)})(3,{log2(850721)})(4,{log2(2769029)})(5,{log2(7610053)})(6,{log2(64440212)})
				};
				\addplot
				[
				color=red,
				mark=triangle,
				style=thick,
				]
				coordinates {
					(2,{log2(2969029)})(3,{log2(3509029)})(4,{log2(3609029)})(5,{log2(15610053)})(6,{log2(86440212)})
				
				};    
				\addplot
				[
				color=cyan,
				mark=diamond,
				style=very thick,
				]
				coordinates {
						(2,{log2(700721)})(3,{log2(862567)})(4,{log2(3909029)})(5,{log2(32440212)})(6,{log2(184440212)})
				};    
				\addplot
				[
				color=black,
				mark=*,
				style=thick,
				]
				coordinates {
					(2,{log2(30440212)})(3,{log2(33440212)})(4,{log2(66440212)})(5,{log2(264440212)})
				};  
				
				\end{axis}
				
				\end{tikzpicture}
			\end{minipage}
			
		\end{tabular}
		\ref{t6}
		\caption{Varying $l$-diversity on TDrive dataset with trajectory length 5 for top-k=1000 and cell size (0.0020,0.0040). MintraBL didn't terminate for $l$-diversity equal to 6 for more than 24 hours.}
		\label{fig:tdrive_ldiv}
	\end{center}
\end{figure}

\textbf{Effect of Varying cell size:} The total execution time and number of candidates generated by different algorithms increase with the decrease in cell size as shown in Figure \ref{fig:tdrive_cellsize}. The number of overlapping cells with each anonymized region of the trajectory increase when the ground is divided into more fine-grained cells i.e. the cell size is decreased. In a nutshell, the size of the search space measured by the number of generated candidates govern the execution time taken by pattern mining algorithms.

\begin{figure}[tbp]
	\begin{center}
		\begin{tabular}{cc}

			\begin{minipage}{.5\textwidth}
				\begin{tikzpicture}[scale=0.7]
				\begin{axis}[
				title={$(a)$\textbf{Total time(TDrive) for top-k=1000}},
				xlabel={Cell size},
				ylabel={Time log base 2 scale (sec)},
				enlarge y limits=true,
				ymajorgrids=true,
				legend to name=t7,
				legend style={at={($(0,0)+(1cm,1cm)$)},legend columns=4,fill=none,draw=black,anchor=center, align=center},
				]
				
				\addplot
				[
				color=blue,
				mark=square,
				style=thick,
				]
				coordinates {
					(1,{log2(127)})(2,{log2(256)})(3,{log2(400)})(4,{log2(1869)})(5,{log2(2871)})
				};
				\addlegendentry{TopkMintra}
				\addplot
				[
				color=red,
				mark=triangle,
				style=thick,
				]
				coordinates {
					(1,{log2(337)})(2,{log2(507)})(3,{log2(800)})(4,{log2(2600)})(5,{log2(3578)})
				
				};    
				\addlegendentry{MintraBL+s}
				\addplot
				[
				color=cyan,
				mark=diamond,
				style=very thick,
				]
				coordinates {
						(1,{log2(569)})(2,{log2(780)})(3,{log2(1380)})(4,{log2(3800)})(5,{log2(6485)})
				};    
				\addlegendentry{MintraBL+i}
				\addplot
				[
				color=black,
				mark=*,
				style=thick,
				]
				coordinates {
					(1,{log2(2000)})(2,{log2(4500)})(3,{log2(9147)})(4,{log2(22257)})(5,{log2(32589)})
				};    
				\addlegendentry{MintraBL}

				\end{axis}
				
				\end{tikzpicture}
			\end{minipage}
			
			&
			
			\begin{minipage}{.5\textwidth}
				\begin{tikzpicture}[scale=0.7]
				\begin{axis}[
				title={$(b)$\textbf{Number of recursive calls(TDrive) for top-k=1000}},
				xlabel={Cell size},
				ylabel={Number of recursive calls log base 2 scale},
				enlarge y limits=true,
				ymajorgrids=true,
				]
				
				\addplot
				[
				color=blue,
				mark=square,
				style=thick,
				]
				coordinates {
					(1,{log2(7417789)})(2,{log2(12675242)})(3,{log2(71973191)})(4,{log2(584969234)})(5,{log2(1184909000)})
				};
				\addplot
				[
				color=red,
				mark=triangle,
				style=thick,
				]
				coordinates {
					(1,{log2(14421958)})(2,{log2(51365390)})(3,{log2(91365390)})(4,{log2(769938468)})(5,{log2(1442195645)})
				
				};    
				\addplot
				[
				color=cyan,
				mark=diamond,
				style=very thick,
				]
				coordinates {
						(1,{log2(35678247)})(2,{log2(142511240)})(3,{log2(222511240)})(4,{log2(1469938468)})(5,{log2(1856939850)})
				};    
				\addplot
				[
				color=black,
				mark=*,
				style=thick,
				]
				coordinates {
					(1,{log2(3604912535)})(2,{log2(6604912835)})(3,{log2(12204913532)})(4,{log2(18203912539)})(5,{log2(22203912539)})
				};  
				
				\end{axis}
				
				\end{tikzpicture}
			\end{minipage}
			
		\end{tabular}
		\ref{t7}
		\caption{Varying cell size on TDrive dataset with trajectory length 5 for top-k=1000. The cell size is decreased along the x-axis from left to right. Cell size for our experiments are: (0.0020,0.0040), (0.0018,0.0036), (0.0016, 0.0032), (0.0014, 0.0028), and (0.0012, 0.0024).}
		\label{fig:tdrive_cellsize}
	\end{center}
\end{figure}
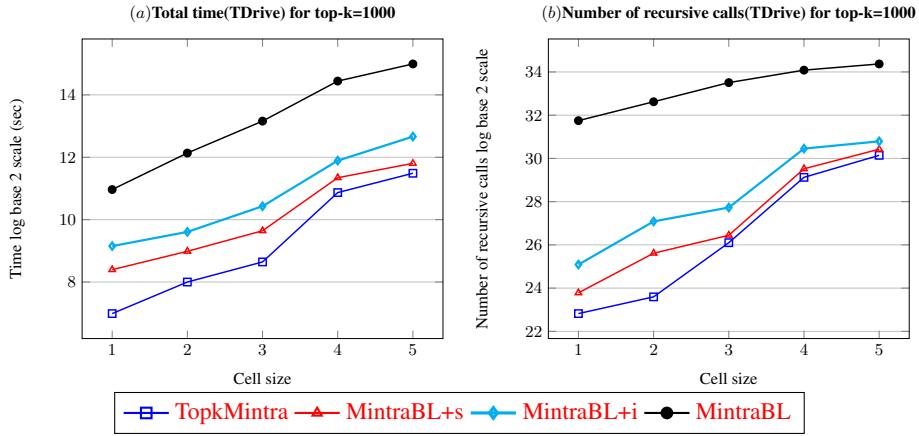

%% file: conclusion.tex
This work studies a problem of pattern-mining over an anonymized database of annotated trajectories. We introduce an approach TopKMintra that mine top-k patterns from the anonymized data which is encoded into a wLAS-sequence form. TopKMintra is a pattern-growth based sequential pattern-mining algorithm which explores two-dimensional wLAS-sequence data by restricting duplicate pattern generation. To improve the efficiency of TopKMintra we have adopted two efficiency strategies, namely threshold initialization (TI) and threshold updation (TU), and two pruning strategies, namely width pruning, and depth pruning. Since there exists no work which discusses pattern mining over anonymized data, we compare TopKMintra with an adapted version of USpan (namely $MintraBL$). We demonstrate through experiments that, though inherently it is difficult to extract correlated information from an anonymized database, TopKMintra finds the top-k most relevant trajectories efficiently as compared to the baseline algorithm in term of execution time, memory usage and size of the search tree. We have also shown that TopKMintra is better scalable to the data size. 
With increasingly more geo-tagged time-series data day-by-day, much of which is anonymized, we believe that our top-k framework can be adopted for practical usage. As a possible extension to this work, we suggest associating a weight with activities in the wLAS-sequence data. The weight associated with activities may break the information loss about its presence in the anonymized region, due to which we assumed activities to be uniformly presence over the region. Two types of weight can be associated-- a global weight capturing the popularity of various activities and a local weight associated with the (location, activity) pair capturing the rank of an outlet. Associating weight with the activity will increase the quality of the mined pattern. This work shows that meaningful patterns from the anonymized data can be mined efficiently. We hope to generate interest in the data mining research community for more and better ways to extract information from anonymized databases.

%% file: samplepaper.bbl
\begin{thebibliography}{10}

\bibitem{ardagna07}
C.~A. Ardagna, M.~Cremonini, E.~Damiani, S.~De~Capitani di~Vimercati, and
  P.~Samarati.
\newblock Location privacy protection through obfuscation-based techniques.
\newblock In {\em Proceedings of the 21st Annual IFIP WG 11.3 Working
  Conference on Data and Applications Security}, pages 47--60, Berlin,
  Heidelberg, 2007. Springer-Verlag.

\bibitem{bamba08}
Bhuvan Bamba, Ling Liu, Peter Pesti, and Ting Wang.
\newblock Supporting anonymous location queries in mobile environments with
  privacygrid.
\newblock In {\em Proceedings of the 17th International Conference on World
  Wide Web}, WWW '08, pages 237--246, New York, NY, USA, 2008. ACM.

\bibitem{beresford03}
A.~R. Beresford and F.~Stajano.
\newblock Location privacy in pervasive computing.
\newblock {\em IEEE Pervasive Computing}, 2:46--55, 01 2003.

\bibitem{chow07}
Chi-Yin Chow and Mohamed~F. Mokbel.
\newblock Enabling private continuous queries for revealed user locations.
\newblock In {\em Proceedings of the 10th International Conference on Advances
  in Spatial and Temporal Databases}, SSTD'07, pages 258--273, Berlin,
  Heidelberg, 2007. Springer-Verlag.

\bibitem{mokbel08}
Chi-Yin Chow, Mohamed~F. Mokbel, and Tian He.
\newblock Tinycasper: A privacy-preserving aggregate location monitoring system
  in wireless sensor networks.
\newblock In {\em Proceedings of the 2008 ACM SIGMOD International Conference
  on Management of Data}, SIGMOD '08, pages 1307--1310, New York, NY, USA,
  2008. ACM.

\bibitem{gruteser03}
Marco Gruteser and Dirk Grunwald.
\newblock Anonymous usage of location-based services through spatial and
  temporal cloaking.
\newblock In {\em Proceedings of the 1st International Conference on Mobile
  Systems, Applications and Services}, MobiSys '03, pages 31--42, New York, NY,
  USA, 2003. ACM.

\bibitem{Hoh05}
Baik Hoh and M.~Gruteser.
\newblock Protecting location privacy through path confusion.
\newblock In {\em First International Conference on Security and Privacy for
  Emerging Areas in Communications Networks (SECURECOMM'05)}, pages 194--205,
  Sep. 2005.

\bibitem{dummy05}
H.~Kido, Y.~Yanagisawa, and T.~Satoh.
\newblock An anonymous communication technique using dummies for location-based
  services.
\newblock In {\em International Conference on Pervasive Services}, pages
  88--97, July 2005.

\bibitem{Machanavajjhala07}
Ashwin Machanavajjhala, Daniel Kifer, Johannes Gehrke, and Muthuramakrishnan
  Venkitasubramaniam.
\newblock L-diversity: Privacy beyond k-anonymity.
\newblock {\em ACM Transactions on Knowledge Discovery from Data}, 1(1), March
  2007.

\bibitem{Samarati01}
P.~Samarati.
\newblock Protecting respondents identities in microdata release.
\newblock {\em IEEE Transactions on Knowledge and Data Engineering},
  13(6):1010--1027, Nov 2001.

\bibitem{xu08}
T.~Xu and Y.~Cai.
\newblock Exploring historical location data for anonymity preservation in
  location-based services.
\newblock In {\em IEEE 27th Conference on Computer Communications}, pages
  547--555, April 2008.

\bibitem{You07}
Tun-Hao You, Wen-Chih Peng, and Wang-Chien Lee.
\newblock Protecting moving trajectories with dummies.
\newblock In {\em Proceedings of the 2007 International Conference on Mobile
  Data Management}, MDM '07, pages 278--282. IEEE Computer Society, 2007.

\bibitem{Mintra16}
Anuj~S. Saxena, Vikram Goyal, and Debajyoti Bera.
\newblock Mintra: Mining anonymized trajectories with annotations.
\newblock In {\em Proceedings of the 20th International Database Engineering;
  Applications Symposium}, IDEAS '16, pages 105--114. ACM, 2016.

\bibitem{ZHENG12}
Vincent~W. Zheng, Yu~Zheng, Xing Xie, and Qiang Yang.
\newblock Towards mobile intelligence: Learning from gps history data for
  collaborative recommendation.
\newblock {\em Artificial Intelligence}, 184-185:17 -- 37, 2012.

\bibitem{BIERLAIRE08}
Michel Bierlaire and Emma Frejinger.
\newblock Route choice modeling with network-free data.
\newblock {\em Transportation Research Part C: Emerging Technologies},
  16(2):187 -- 198, 2008.

\bibitem{Rai07}
R.~K. Rai, Michael Balmer, Marcel Rieser, V.~S. Vaze, Stefan Schönfelder, and
  Kay~W. Axhausen.
\newblock Capturing human activity spaces: New geometries.
\newblock {\em Transportation Research Record}, 2021(1):70--80, 2007.

\bibitem{Zheng08}
Yu~Zheng, Quannan Li, Yukun Chen, Xing Xie, and Wei-Ying Ma.
\newblock Understanding mobility based on gps data.
\newblock In {\em Proceedings of the 10th International Conference on
  Ubiquitous Computing}, UbiComp '08, pages 312--321, New York, NY, USA, 2008.
  ACM.

\bibitem{Horvitz05}
Eric Horvitz, Johnson Apacible, Raman Sarin, and Lin Liao.
\newblock Prediction, expectation, and surprise: Methods, designs, and study of
  a deployed traffic forecasting service.
\newblock {\em CoRR}, abs/1207.1352, 2012.

\bibitem{HERRERA10}
Juan~C. Herrera, Daniel~B. Work, Ryan Herring, Xuegang~(Jeff) Ban, Quinn
  Jacobson, and Alexandre~M. Bayen.
\newblock Evaluation of traffic data obtained via gps-enabled mobile phones:
  The mobile century field experiment.
\newblock {\em Transportation Research Part C: Emerging Technologies},
  18(4):568 -- 583, 2010.

\bibitem{Goh12}
C.~Y. Goh, J.~Dauwels, N.~Mitrovic, M.~T. Asif, A.~Oran, and P.~Jaillet.
\newblock Online map-matching based on hidden markov model for real-time
  traffic sensing applications.
\newblock In {\em 15th International IEEE Conference on Intelligent
  Transportation Systems}, pages 776--781, Sep. 2012.

\bibitem{MONTOYA03}
Lorena Montoya.
\newblock Geo-data acquisition through mobile gis and digital video: an urban
  disaster management perspective.
\newblock {\em Environmental Modelling \& Software}, 18(10):869 -- 876, 2003.
\newblock Integrating Environmental Modelling and GI-Technology.

\bibitem{MANSOURIAN06}
A.~Mansourian, A.~Rajabifard, M.J.~Valadan Zoej, and I.~Williamson.
\newblock Using sdi and web-based system to facilitate disaster management.
\newblock {\em Computers \& Geosciences}, 32(3):303 -- 315, 2006.

\bibitem{Dai15}
J.~Dai, B.~Yang, C.~Guo, and Z.~Ding.
\newblock Personalized route recommendation using big trajectory data.
\newblock In {\em IEEE 31st International Conference on Data Engineering},
  pages 543--554, April 2015.

\bibitem{Kurashima10}
Takeshi Kurashima, Tomoharu Iwata, Go~Irie, and Ko~Fujimura.
\newblock Travel route recommendation using geotags in photo sharing sites.
\newblock In {\em Proceedings of the 19th ACM International Conference on
  Information and Knowledge Management}, CIKM '10, pages 579--588, New York,
  NY, USA, 2010. ACM.

\bibitem{nergiz2008towards}
Mehmet~Ercan Nergiz, Maurizio Atzori, and Yucel Saygin.
\newblock Towards trajectory anonymization: A generalization-based approach.
\newblock In {\em Proceedings of the SIGSPATIAL ACM GIS 2008 International
  Workshop on Security and Privacy in GIS and LBS}, SPRINGL '08, pages 52--61,
  New York, NY, USA, 2008. ACM.

\bibitem{OUR_ICISS}
Anuj~S. Saxena, Mayank Pundir, Vikram Goyal, and Debajyoti Bera.
\newblock Preserving location privacy for continuous queries on known route.
\newblock In {\em Proceedings of the 7th International Conference on
  Information Systems Security}, ICISS'11, pages 265--279, Berlin, Heidelberg,
  2011. Springer-Verlag.

\bibitem{gedik08}
B.~Gedik and L.~Liu.
\newblock Protecting location privacy with personalized k-anonymity:
  Architecture and algorithms.
\newblock {\em IEEE Transactions on Mobile Computing}, 7(1):1--18, Jan 2008.

\bibitem{Saxena13}
Anuj~Shanker Saxena, Vikram Goyal, and Debajyoti Bera.
\newblock Efficient enforcement of privacy for moving object trajectories.
\newblock In {\em International Conference on Information Systems Security},
  pages 360--374. Springer, 2013.

\bibitem{han2000freespan}
Jiawei Han, Jian Pei, Behzad Mortazavi-Asl, Qiming Chen, Umeshwar Dayal, and
  Mei-Chun Hsu.
\newblock Freespan: Frequent pattern-projected sequential pattern mining.
\newblock In {\em Proceedings of the Sixth ACM SIGKDD International Conference
  on Knowledge Discovery and Data Mining}, KDD '00, pages 355--359, New York,
  NY, USA, 2000. ACM.

\bibitem{han2004mining}
Jiawei Han, Jian Pei, Yiwen Yin, and Runying Mao.
\newblock Mining frequent patterns without candidate generation: A
  frequent-pattern tree approach.
\newblock {\em Data mining and knowledge discovery}, 8(1):53--87, January 2004.

\bibitem{pei2001prefixspan}
Jian Pei, Jiawei Han, Behzad Mortazavi-Asl, Helen Pinto, Qiming Chen, Umeshwar
  Dayal, and Meichun Hsu.
\newblock Prefixspan: Mining sequential patterns by prefix-projected growth.
\newblock In {\em Proceedings of the 17th International Conference on Data
  Engineering}, pages 215--224, Washington, DC, USA, 2001. IEEE Computer
  Society.

\bibitem{array_eswa}
Bac Le, Ut~Huynh, and Duy-Tai Dinh.
\newblock A pure array structure and parallel strategy for high-utility
  sequential pattern mining.
\newblock {\em Expert Systems with Applications}, 104:107 -- 120, 2018.

\bibitem{mmu}
Jerry Chun-Wei Lin, Jiexiong Zhang, and Philippe Fournier-Viger.
\newblock High-utility sequential pattern mining with multiple minimum utility
  thresholds.
\newblock In {\em Asia-Pacific Web (APWeb) and Web-Age Information Management
  (WAIM) Joint Conference on Web and Big Data}, pages 215--229. Springer, 2017.

\bibitem{tseng2010up}
Vincent~S. Tseng, Cheng-Wei Wu, Bai-En Shie, and Philip~S. Yu.
\newblock Up-growth: An efficient algorithm for high utility itemset mining.
\newblock In {\em Proceedings of the 16th ACM SIGKDD International Conference
  on Knowledge Discovery and Data Mining}, KDD '10, pages 253--262, New York,
  NY, USA, 2010. ACM.

\bibitem{inc_huspm}
Jun-Zhe Wang and Jiun-Long Huang.
\newblock On incremental high utility sequential pattern mining.
\newblock {\em ACM Transactions on Intelligent Systems and Technology},
  9(5):55:1--55:26, June 2018.

\bibitem{kais}
Jun-Zhe Wang, Jiun-Long Huang, and Yi-Cheng Chen.
\newblock On efficiently mining high utility sequential patterns.
\newblock {\em Knowledge and Information Systems}, 49(2):597--627, Nov 2016.

\bibitem{yin2012uspan}
Junfu Yin, Zhigang Zheng, and Longbing Cao.
\newblock Uspan: An efficient algorithm for mining high utility sequential
  patterns.
\newblock In {\em Proceedings of the 18th ACM SIGKDD International Conference
  on Knowledge Discovery and Data Mining}, KDD '12, pages 660--668, New York,
  NY, USA, 2012. ACM.

\bibitem{icdm}
J.~Yin, Z.~Zheng, L.~Cao, Y.~Song, and W.~Wei.
\newblock Efficiently mining top-k high utility sequential patterns.
\newblock In {\em 2013 IEEE 13th International Conference on Data Mining},
  pages 1259--1264, Dec 2013.

\bibitem{arya15}
Krishan~K Arya, Vikram Goyal, Shamkant~B Navathe, and Sushil Prasad.
\newblock Mining frequent spatial-textual sequence patterns.
\newblock In {\em International Conference on Database Systems for Advanced
  Applications}, pages 123--138. Springer, 2015.

\bibitem{cao2005mining}
Huiping Cao, N.~Mamoulis, and D.~W. Cheung.
\newblock Mining frequent spatio-temporal sequential patterns.
\newblock In {\em Fifth IEEE International Conference on Data Mining
  (ICDM'05)}, pages 8--pp, Nov 2005.

\bibitem{giannotti2007trajectory}
Fosca Giannotti, Mirco Nanni, Fabio Pinelli, and Dino Pedreschi.
\newblock Trajectory pattern mining.
\newblock In {\em Proceedings of the 13th ACM SIGKDD International Conference
  on Knowledge Discovery and Data Mining}, KDD '07, pages 330--339, New York,
  NY, USA, 2007. ACM.

\bibitem{tsoukatos2001efficient}
Ilias Tsoukatos and Dimitrios Gunopulos.
\newblock Efficient mining of spatiotemporal patterns.
\newblock In {\em Proceedings of the 7th International Symposium on Advances in
  Spatial and Temporal Databases}, SSTD '01, pages 425--442, London, UK, UK,
  2001. Springer-Verlag.

\bibitem{Agrawal1994}
Rakesh Agrawal and Ramakrishnan Srikant.
\newblock Fast algorithms for mining association rules in large databases.
\newblock In {\em Proceedings of the 20th International Conference on Very
  Large Data Bases}, VLDB '94, pages 487--499, 1994.

\bibitem{tks}
Philippe Fournier-Viger, Antonio Gomariz, Ted Gueniche, Esp{\'e}rance
  Mwamikazi, and Rincy Thomas.
\newblock Tks: efficient mining of top-k sequential patterns.
\newblock In {\em International Conference on Advanced Data Mining and
  Applications}, pages 109--120. Springer, 2013.

\bibitem{badge92}
Roy Want, Andy Hopper, Veronica Falc\~{a}o, and Jonathan Gibbons.
\newblock The active badge location system.
\newblock {\em ACM Transactions on Information Systems}, 10(1):91--102, January
  1992.

\bibitem{ji}
Jason~I. Hong and James~A. Landay.
\newblock An architecture for privacy-sensitive ubiquitous computing.
\newblock In {\em Proceedings of the 2Nd International Conference on Mobile
  Systems, Applications, and Services}, MobiSys '04, pages 177--189, New York,
  NY, USA, 2004. ACM.

\bibitem{MLY08}
M.~L. Yiu, C.~S. Jensen, X.~Huang, and H.~Lu.
\newblock Spacetwist: Managing the trade-offs among location privacy, query
  performance, and query accuracy in mobile services.
\newblock In {\em IEEE 24th International Conference on Data Engineering},
  pages 366--375, April 2008.

\bibitem{chow06}
Chi-Yin Chow, Mohamed~F. Mokbel, and Xuan Liu.
\newblock A peer-to-peer spatial cloaking algorithm for anonymous
  location-based service.
\newblock In {\em Proceedings of the 14th Annual ACM International Symposium on
  Advances in Geographic Information Systems}, GIS '06, pages 171--178, New
  York, NY, USA, 2006. ACM.

\bibitem{casper06}
Mohamed~F. Mokbel, Chi-Yin Chow, and Walid~G. Aref.
\newblock The new casper: Query processing for location services without
  compromising privacy.
\newblock In {\em Proceedings of the 32Nd International Conference on Very
  Large Data Bases}, VLDB '06, pages 763--774. VLDB Endowment, 2006.

\bibitem{hu09}
H.~{Hu} and J.~{Xu}.
\newblock Non-exposure location anonymity.
\newblock In {\em 2009 IEEE 25th International Conference on Data Engineering},
  pages 1120--1131, March 2009.

\bibitem{DewriRRW10}
R.~Dewri, I.~Ray, I.~Ray, and D.~Whitley.
\newblock Query m-invariance: Preventing query disclosures in continuous
  location-based services.
\newblock In {\em 2010 Eleventh International Conference on Mobile Data
  Management}, pages 95--104, May 2010.

\bibitem{Chow2011}
Chi-Yin Chow, Mohamed~F. Mokbel, Jie Bao, and Xuan Liu.
\newblock Query-aware location anonymization for road networks.
\newblock {\em GeoInformatica}, 15(3):571--607, Jul 2011.

\bibitem{GHINITA10}
Gabriel Ghinita, Keliang Zhao, Dimitris Papadias, and Panos Kalnis.
\newblock A reciprocal framework for spatial k-anonymity.
\newblock {\em Information Systems}, 35(3):299 -- 314, 2010.

\bibitem{lee11}
Byoungyoung Lee, Jinoh Oh, Hwanjo Yu, and Jong Kim.
\newblock Protecting location privacy using location semantics.
\newblock In {\em Proceedings of the 17th ACM SIGKDD International Conference
  on Knowledge Discovery and Data Mining}, KDD '11, pages 1289--1297, 2011.

\bibitem{Xue09}
Mingqiang Xue, Panos Kalnis, and Hung~Keng Pung.
\newblock Location diversity: Enhanced privacy protection in location based
  services.
\newblock In {\em Proceedings of the 4th International Symposium on Location
  and Context Awareness}, LoCA '09, pages 70--87, 2009.

\bibitem{agrawal1995mining}
R.~Agrawal and R.~Srikant.
\newblock Mining sequential patterns.
\newblock In {\em Proceedings of the Eleventh International Conference on Data
  Engineering}, pages 3--14, March 1995.

\bibitem{han2000mining}
Jiawei Han, Jian Pei, and Yiwen Yin.
\newblock Mining frequent patterns without candidate generation.
\newblock In {\em Proceedings of the 2000 ACM SIGMOD International Conference
  on Management of Data}, SIGMOD '00, pages 1--12, New York, NY, USA, 2000.
  ACM.

\bibitem{zaki2001spade}
Mohammed~J. Zaki.
\newblock Spade: An efficient algorithm for mining frequent sequences.
\newblock {\em Machine Learning}, 42(1-2):31--60, January 2001.

\bibitem{ahmed2009efficient}
C.~F. Ahmed, S.~K. Tanbeer, B.~Jeong, and Y.~Lee.
\newblock Efficient tree structures for high utility pattern mining in
  incremental databases.
\newblock {\em IEEE Transactions on Knowledge and Data Engineering},
  21(12):1708--1721, Dec 2009.

\bibitem{ahmed}
Chowdhury~Farhan Ahmed, Syed~Khairuzzaman Tanbeer, and Byeong-Soo Jeong.
\newblock A novel approach for mining high-utility sequential patterns in
  sequence databases.
\newblock {\em ETRI Journal}, 32(5):676--686, 2010.

\bibitem{crom}
O.~K. Alkan and P.~Karagoz.
\newblock Crom and huspext: Improving efficiency of high utility sequential
  pattern extraction.
\newblock {\em IEEE Transactions on Knowledge and Data Engineering},
  27(10):2645--2657, Oct 2015.

\bibitem{seq_rules}
Souleymane Zida, Philippe Fournier-Viger, Cheng-Wei Wu, Jerry Chun-Wei Lin, and
  Vincent~S. Tseng.
\newblock Efficient mining of high-utility sequential rules.
\newblock In Petra Perner, editor, {\em Machine Learning and Data Mining in
  Pattern Recognition}, pages 157--171, 2015.

\bibitem{koperski1995discovery}
Krzysztof Koperski and Jiawei Han.
\newblock Discovery of spatial association rules in geographic information
  databases.
\newblock In {\em Proceedings of the 4th International Symposium on Advances in
  Spatial Databases}, SSD '95, pages 47--66, Berlin, Heidelberg, 1995.
  Springer-Verlag.

\bibitem{chung2004evolutionary}
Fu-Lai Chung, Tak-Chung Fu, V.~Ng, and R.~W.~P. Luk.
\newblock An evolutionary approach to pattern-based time series segmentation.
\newblock {\em IEEE Transactions on Evolutionary Computation}, 8(5):471--489,
  Oct 2004.

\bibitem{Gidofalvi07}
G.~Gidofalvi, X.~Huang, and T.~B. Pedersen.
\newblock Privacy-preserving data mining on moving object trajectories.
\newblock In {\em International Conference on Mobile Data Management}, pages
  60--68, May 2007.

\bibitem{Abul08}
O.~Abul, F.~Bonchi, and M.~Nanni.
\newblock Never walk alone: Uncertainty for anonymity in moving objects
  databases.
\newblock In {\em IEEE 24th International Conference on Data Engineering},
  pages 376--385, April 2008.

\bibitem{kalnis07}
P.~{Kalnis}, G.~{Ghinita}, K.~{Mouratidis}, and D.~{Papadias}.
\newblock Preventing location-based identity inference in anonymous spatial
  queries.
\newblock {\em IEEE Transactions on Knowledge and Data Engineering},
  19(12):1719--1733, Dec 2007.

\bibitem{tdrive_2}
Jing Yuan, Yu~Zheng, Chengyang Zhang, Wenlei Xie, Xing Xie, Guangzhong Sun, and
  Yan Huang.
\newblock T-drive: Driving directions based on taxi trajectories.
\newblock In {\em Proceedings of the 18th SIGSPATIAL International Conference
  on Advances in Geographic Information Systems}, pages 99--108, 2010.

\bibitem{tdrive_1}
Jing Yuan, Yu~Zheng, Xing Xie, and Guangzhong Sun.
\newblock Driving with knowledge from the physical world.
\newblock In {\em Proceedings of the 17th ACM SIGKDD International Conference
  on Knowledge Discovery and Data Mining}, KDD '11, pages 316--324. ACM, 2011.

\bibitem{ankita17}
Ankita Likhyani, Srikanta Bedathur, and Deepak P.
\newblock Locate: Influence quantification for location promotion in
  location-based social networks.
\newblock In {\em Proceedings of the Twenty-Sixth International Joint
  Conference on Artificial Intelligence, {IJCAI-17}}, pages 2259--2265, 2017.

\end{thebibliography}
